\documentclass[a4paper,11pt,DIV=calc,parskip=false]{scrartcl}

\usepackage[utf8]{inputenc}
\usepackage{amsmath,amsfonts,amssymb,amsthm,mathtools}
\usepackage{txfonts}
\usepackage{dsfont}
\usepackage[ngerman,english]{babel}
\usepackage[inline]{enumitem}
\usepackage[all]{xy}
\usepackage{microtype}
\usepackage{mathrsfs,cite}
\usepackage{color}
\numberwithin{equation}{section}

\DeclareMathOperator{\tr}{Tr}

\DeclareMathOperator{\hc}{h.c.}

\newtheorem{theorem}{Theorem}[section]
\newtheorem{proposition}[theorem]{Proposition}
\newtheorem{lemma}[theorem]{Lemma}

\theoremstyle{definition}
\newtheorem{remark}[theorem]{Remark}

\newtheorem{assumption}[theorem]{Assumptions}

\newcommand{\dda}{\mathrm{d}}
\newcommand{\de}{\,\dda}

\newcommand{\cc}{\overline}

\renewcommand{\Im}{\textrm{Im}}
\newcommand{\rmd}{\mathrm{d}}

\renewcommand\rho\varrho
\renewcommand\epsilon\varepsilon

\definecolor{darkred}{rgb}{0.9,0,0.3}
\definecolor{darkblue}{rgb}{0,0.3,0.9}

\newcommand{\ii}{\mathrm{i}}
\newcommand{\dd}{\mathrm{d}}
\newcommand*{\deq}{\mathrel{\vcenter{\baselineskip0.65ex \lineskiplimit0pt \hbox{.}\hbox{.}}}=}

\definecolor{vdarkred}{rgb}{0.7,0,0.2}
\definecolor{vdarkblue}{rgb}{0,0.2,0.7}
\usepackage[pdftex, colorlinks, linkcolor=vdarkblue,citecolor=vdarkred]{hyperref}
\renewcommand{\cal}{\mathcal}

\def\fh{{\mathfrak h}}
\def\bR{\mathbb{R}}

\def\bN{\mathbb{N}}
\def\NN{\mathbb{N}}

\def\cU{\mathcal{U}}
\def\cQ{\mathcal{Q}}
\def\cD{\mathcal{D}}

\def\cR{\mathcal{R}}
\def\cW{\mathcal{W}}
\def\cV{\mathcal{V}}

\def\cF{\mathcal{F}}
\def\cG{\mathcal{G}}
\def\cL{\mathcal{L}}

\def\cI{\mathcal{I}}
\def\cN{\mathcal{N}}
\def\cE{\mathcal{E}}

\def\cH{\mathcal{H}}
\def\cS{\mathcal{S}}

\def\cT{\mathcal{T}}
\def\cB{\mathcal{B}}
\def\eps{\varepsilon}

\def\RRR{\mathbb{R}} 
\def\RR{\mathbb{R}}
\def\CC{\mathbb{C}}

\def\wt{\widetilde}
\def\wh#1{\widehat{#1}}
\def\dpr{\partial}
\def\dt{\dpr_t}
\def\id{\mathds{1}} 
\def\fock{\mathscr{F}}

\def \xx{\mathbf x}

\def \yy{\mathbf y}
\def \zz{\mathbf z}
\def \xy{(\xx\longleftrightarrow \yy)}

\def \pphi{\pmb{\phi}}
\def \Ux{U_{t,\xx}}
\def \Uy{U_{t,\yy}}
\def \Vx{\cc{V_{t,\xx}}}
\def \Vy{\cc{V_{t,\yy}}}
\def \Uxs{U_{s,\xx}}
\def \Uys{U_{s,\yy}}
\def \Vxs{\cc{V_{s,\xx}}}
\def \Vys{\cc{V_{s,\yy}}}
\def \xt{\xi_{t}}
\def \xat{\xi_{\h,t}}

\def \xas{\xi_{\h,s}}

\def \Ulx{U_{t, \ell, x}}
\def \Uly{U_{t, \ell, y}}
\def \Vlx{\cc{V_{t, \ell, x}}}
\def \Vly{\cc{V_{t, \ell, y}}}
\def \gpt{\g_t^{\phi_t}}
\def \apt{\a_t^{\phi_t}}
\def \Iplus#1{\cI_{N,t}^{(#1,\eta),>}}

\let\a=\alpha     \let\g=\gamma          
  \let\h=\eta     \let\th=\vartheta        
    \let\n=\nu      \let\x=\xi                \let\r=\rho
\let\s=\sigma             
   \let\o=\omega     
\let\G=\Gamma \let\D=\Delta   \let\Th=\Theta    \let\L=\Lambda    
             
\let\O=\Omega

\theoremstyle{definition}

\newcommand{\beq}{\begin{equation}}
\newcommand{\eeq}{\end{equation}}
\usepackage{pdfsync}

\begin{document}

\title{Dynamics of mean-field bosons at positive temperature}

\author{Marco Caporaletti, Andreas Deuchert, Benjamin Schlein}

\date{\today}

\maketitle

\begin{abstract} 
	We study the time evolution of an initially trapped weakly interacting Bose gas at positive temperature, after the trapping potential has been switched off. It has been recently shown in \cite{DeuSei2021} that the one-particle density matrix of Gibbs states of the interacting trapped gas is given, to leading order in $N$, as $N \to \infty$, by the one of the ideal gas, with the condensate wave function replaced by the minimizer of the Hartree energy functional. We show that this structure is stable with respect to the many-body evolution in the following sense: the dynamics can be approximated in terms of the time-dependent Hartree equation for the condensate wave function and in terms of the free evolution for the thermally excited particles. The main technical novelty of our work is the use of the Hartree–Fock–Bogoliubov equations to define a  fluctuation dynamics. 
\end{abstract}

\setcounter{tocdepth}{2}
\tableofcontents

\section{Introduction and main results}

\subsection{Background and summary}
\label{sec:introduction}
The time evolution of bosonic many-particle systems in the mean-field limit has received a considerable amount of attention in recent years. The main theme in most works is the study of the dynamics of Bose gases that are initially prepared in a ground state (equilibrium state at zero temperature) of a trapped Hamiltonian after the trap has been switched off. It is well-known, see e.g. \cite{LewNamRou2014,Rou2015} and references therein, that under appropriate conditions on the interaction potential, ground states of interacting mean-field systems display complete Bose--Einstein condensation (BEC), that is, the largest eigenvalue $\lambda_N$ of the one-particle density matrix (1-pdm) of the ground state wave function satisfies $\lambda_N/N \to 1$ as $N \to \infty$. It turns out that this structure is stable w.r.t. the dynamics generated by the time-dependent many-particle Schr\"odinger equation. More precisely, it can be shown that the system displays complete BEC also at time $t>0$, where the condensate wave function (the eigenfunction related to the largest eigenvalue of the 1-pdm) is given by the solution to the time-dependent Hartree equation.

The first results in this direction were obtained in \cite{Hepp1974} and then in \cite{GinVel1979a,GinVel1979b,Spohn80}. The case of Coulomb interactions has been studied more recently in \cite{BaGolMau2000,ErdYau2001} and in \cite{ElgSchlein2007} with a relativistic dispersion relation. In \cite{FroKnoPizzo2001,FroKnoSchwarz2009} the convergence towards the Hartree equation has been interpreted as an Egorov type theorem, \cite{AmmaNier2009} focuses on the propagation of Wigner measures, and bounds on the rate of convergence towards the Hartree dynamics have been obtained in  \cite{RodSchlein2009,KnowlesPickl2010,ChenLeeSchl2011,AmmaFalPawi2016}. Once the convergence towards the time-dependent Hartree equation is established, it is natural to ask whether it is possible to make statements about the fluctuations around the Hartree dynamics. As shown in \cite{Hepp1974,GrillMacheMarg2010,GrillMacheMarg2011,BenKirkSchl2013,BuSafSchl2014,LewNamSchlein2015}, for mean-field systems it is possible to approximate the dynamics of particles outside the condensate by the time evolution generated by a Bogoliubov Hamiltonian. This, in particular, allows one to obtain a norm approximation for the wave function, and not only for the 1-pdm. A systematic expansion for correlation functions in terms of the solution of the Hartree equation and a Bogoliubov two-point function that approximates the original dynamics to arbitrary precision has been obtained in \cite{BossPavPicklSoff2020}.

Another interesting scaling limit for the Bose gas is the Gross--Pitaevskii (GP) limit, which can be used to describe modern experiments with cold alkali gases. In contrast to the mean-field limit, where particle collisions are frequent but weak, particles interact rarely but strongly in the GP limit. These strong collisions induce microscopic correlations between the particles that are the main reason why the GP limit is more challenging from a mathematical point of view than the mean-field limit. It has been shown in \cite{LiebSeiYng2000,LiebSei2002} that ground states of trapped gases display complete BEC also in the GP limit. Moreover, in the series of works \cite{ErdSchlYau2006,ErdSchlYau2007,ErdSchlYau2009,ErdSchlYau2010} it has been established that there is complete BEC also at later times and that the condensate wave function is given by the solution to the time-dependent GP equation. For other results in this direction, we refer to \cite{KlaiMach2008,KirpSchlStaff2011,CheHolm2016a,CheHolm2016b,CheHaiPavlSei2014,CheHaiPavlSei2015,Pickl2015,BenOlivSchl2015,JeblLeopPickl2019,BreSchl2019,GriMach2013,Kuz2017,GriMach2017,GriMach2019,NamNap2017a,NamNap2017b,BocCenaSchl2017,BreNamNapSchl2018,ChonZhao2020}. More references concerning the dynamics of BECs in the above scaling limits can be found in the lecture notes \cite{BenPorSchl2015} and in the review article \cite{Nap2021}. 

Here we are interested in the dynamics of mean-field bosons, initially prepared in an equilibrium state at positive temperature (Gibbs state). Just above the critical temperature for BEC, correlation functions of the many-body Gibbs state have been shown to converge, in the mean-field limit, to correlation functions of the invariant measure associated with the nonlinear Hartree equation, see \cite{LewNamRou2015,FroKnoSchleinSoh2017,LewNamRou2021,FroKnoSchleinSoh2021,FroKnoSchleinSoh2022}; observe that, in dimensions $d =2,3$, this requires an appropriate renormalization of the interaction. In \cite{FroKnoSchleinSoh2019}, also time-dependent correlations have been proven to converge. In the present paper, we are primarily interested in initial data prepared at equilibrium below the critical temperature, where the BEC coexists with a thermal cloud and both phases contain a macroscopic number of particles. In \cite{DeuSei2021}, it has been recently shown that at temperatures lower than (but comparable with) the critical temperature of the corresponding ideal (non-interacting) gas, the 1-pdm of the Gibbs state of a trapped mean-field system can be approximated by the 1-pdm of the ideal gas, where the condensate wave function has been replaced by the minimizer of the Hartree energy functional. Our goal here is to investigate the dynamics of such positive temperature states, after the trapping potential has been switched off. In our main theorem we show that the 1-pdm of the evolved state is given, to leading order, by the one of the ideal gas, propagated with the non-interacting time evolution, with the condensate wave function replaced by the solution of the time-dependent Hartree equation. In particular, the number of particles inside the condensate and outside the condensate is preserved by the time evolution, to leading order.

\subsection{Notation}
\label{sec:notation}
We denote by $\langle \cdot, \cdot \rangle$ the inner product in $L^2(\RR^d)$, and by $\mathcal{B}(L^2(\mathbb{R}^d))$ the space of bounded operators on $L^2(\mathbb{R}^d)$. The $L^2$ norm of a function and the norm of an operator in $\mathcal{B}(L^2(\mathbb{R}^d))$ are both denoted by $\Vert \cdot \Vert$. If $p\neq 2$, we denote by $\Vert \cdot \Vert_p$ the $L^p$ norms of a function, and $\Vert A \Vert_{\mathcal{L}^p} = (\tr[ (A^* A)^{p/2} ])^{1/p}$ is the $p$-th Schatten norm of the operator $A \in \mathcal{B}(L^2(\mathbb{R}^d))$, for $p \in [1 , \infty) $. The corresponding Banach spaces of compact operators with finite $\mathcal{L}^p$ norm are denoted by $\mathcal{L}^p$. Sobolev spaces are denoted by $W^{n,p}$ or by $H^n$ if $p=2$, with $\Vert \cdot \Vert_{W^{n,p}}$ and $\Vert \cdot \Vert_{H^n}$ denoting the corresponding Sobolev norms. For the Fourier transform in $\RR^d$ we use the notation, and the convention,
\[\cF f(k)=\hat f (k)=\frac 1 {(2\pi)^{d/2}}\int_{\mathbb{R}^d} e^{\mathrm{i} k \cdot x}f(x)\de x\,.\] For quantities $a,b$ depending on $N$ and/or $t$, we use the notation $a \lesssim b$ to say that there exists a constant $C>0$ with $a \leq C b$. If $a \lesssim b$ and $b \lesssim a$ we write $a \sim b$.
\subsection{The model}
\label{sec:model}
In this article we consider the time evolution of an initially trapped Bose gas, prepared in a positive temperature state (an approximate Gibbs state), after the trapping potential has been switched off. We start by introducing our set-up. 

We will be interested in quantum states with a fluctuating particle number. These states are naturally defined on the bosonic Fock space 
\begin{equation}
	\mathscr{F}(L^2(\mathbb{R}^3)) = \bigoplus_{n=0}^{\infty} L^2_{\mathrm{sym}}(\mathbb{R}^{3n}) \, ,
	\label{eq:bosonicFockspace}
\end{equation}  
where $L^2_{\mathrm{sym}}(\mathbb{R}^{3n})$ denotes the closed linear subspace of $L^2(\mathbb{R}^{3n})$ consisting of those functions $\Psi(x_1, ..., x_n)$ that are invariant under any permutation of the coordinates $x_1, ..., x_N \in \mathbb{R}^3$. As usual, we have $L^2_{\mathrm{sym}}(\mathbb{R}^{0}) = \mathbb{C}$. By $a_x^*$ and $a_x$ we denote the creation and annihilation operators (actually operator-valued distributions) on $\mathscr{F}$ that create and annihilate a particle at point $x \in \mathbb{R}^3$, respectively. They satisfy the canonical commutation relations (CCR) 
\begin{equation*}
	[a_x,a_y^*] = \delta(x-y), \quad \quad [a_x,a_y] = 0 = [a_x^*,a_y^*] \,,
\end{equation*}
where $\delta(x)$ denotes the delta distribution with unit mass at $x=0$. Quantum states on the bosonic Fock space with an expected number of particles equal to $N$ are elements of the set
\begin{equation*}
	\mathcal{S}_N = \{ \Gamma \in \mathcal{B}(\mathscr{F}) \ | \ \Gamma \geq 0, \tr \Gamma = 1, \tr[ \mathcal{N} \Gamma ] = N \} \, ,
\end{equation*}
where $\mathcal{N} = \int_{\mathbb{R}^3} a_x^* a_x \de x$ denotes the particle number operator on $\mathscr{F}$. The one-particle reduced density matrix (1-pdm) $\gamma$ of a state $\Gamma \in \mathcal{S}_N$ is defined via its integral kernel by
\begin{equation*}
	\gamma(x,y) = \tr[ a_y^* a_x \Gamma ]
\end{equation*}
and defines a nonnegative operator on $L^2(\mathbb{R}^3)$ with $\tr \gamma = N$.

The time evolution we are interested in is governed by the Heisenberg equation
\begin{equation}
	\mathrm{i} \partial_t \Gamma_t = [\mathcal{H}_N, \Gamma_t] \,, 
	\label{eq:time-dep-Schroedinger}
\end{equation}
where
\begin{equation}
	\mathcal{H}_N = \int_{\mathbb{R}^3} \nabla a_x^* \nabla a_x \de x + \frac{1}{2 N} \int_{\mathbb{R}^6} v(x-y) a^*_x a^*_y a_y a_x \de x \de y 
	\label{eq:Hamiltonianti}
\end{equation}
and $[A,B] = AB-BA$ denotes the commutator of two linear operators $A$ and $B$. The prefactor $N^{-1}$ in front of $v$ guarantees that, for particles in the condensate, the kinetic and the potential energy are comparable. Any solution to \eqref{eq:time-dep-Schroedinger} with initial condition $\Gamma_0 \in \mathcal{S}_N$ can be written in terms of the strongly continuous unitary group $e^{-\mathrm{i} \mathcal{H}_N t}$ as 
\begin{equation}
	\Gamma_t = e^{-\mathrm{i} \mathcal{H}_N t} \Gamma_0 e^{\mathrm{i} \mathcal{H}_N t} \,.
	\label{eq:solutionschroedingereq}
\end{equation}
The fact that $\mathcal{H}_N$ commutes with $\mathcal{N}$ guarantees $\tr[ \mathcal{N} \Gamma_t ] = \tr[ \mathcal{N} \Gamma_0 ]$. As explained in Section~\ref{sec:introduction}, we are interested in the solution \eqref{eq:solutionschroedingereq} for initial data describing equilibrium states of trapped gases at positive temperature. Let us now recall some known properties of such states.
\subsection{Equilibrium states of trapped Bose gases}
\label{sec:Equilibriumstates}
Equilibrium states of trapped Bose gases can be defined as minimizers, in $\mathcal{S}_N$, of the Gibbs free energy functional 
\begin{equation}
	\mathcal{F}(\Gamma) = \tr [ \mathcal{H}_N^{\mathrm{trap}} \Gamma ] - T S(\Gamma) \, , \quad \text{ with the von Neumann entropy } \quad S(\Gamma) = - \tr[\Gamma \ln (\Gamma)] \,.
	\label{eq:Gibbsfreenergyfunctional}
\end{equation}
Here
\begin{equation}
	\mathcal{H}_N^{\mathrm{trap}} = \mathcal{H}_N +  \int_{\mathbb{R}^3} w(x) a_x^* a_x \de x = \int_{\mathbb{R}^3} a_x^* \left( -\Delta + w(x) \right) a_x \de x + \frac{1}{2 N} \int_{\mathbb{R}^6} v(x-y) a^*_x a^*_y a_y a_x \de x \de y 
	\label{eq:Hamiltonian}
\end{equation}
denotes the Hamiltonian of the trapped system and $T > 0$ is the temperature. The trapping potential $w(x)$ satisfies $w(x) \to \infty$ for $|x| \to \infty$. The minimum of \eqref{eq:Gibbsfreenergyfunctional} is attained for the Gibbs state 
\begin{equation*}
	\Gamma_{\mathrm{G}} = \frac{ \exp(- (\mathcal{H}_N^{\mathrm{trap}} - \mu \mathcal{N})/T )}{\tr [ \exp(- (\mathcal{H}_N^{\mathrm{trap}} - \mu \mathcal{N})/T ) ]} \,,
\end{equation*}
where the chemical potential $\mu \in \mathbb{R}$ is chosen (depending on $N$) so that $\tr [ \mathcal{N} \Gamma ] = N$. This leads to the (grand canonical) free energy
\begin{equation}
	F(T,N) = \min_{\Gamma \in \mathcal{S}_N} \mathcal{F}(\Gamma) = - T \ln \tr [ \exp(- (\mathcal{H}_N^{\mathrm{trap}} - \mu \mathcal{N})/T ) ] + \mu N \,.
	\label{eq:freeenergy}
\end{equation}

Before we discuss what is known about approximate minimizers of the Gibbs free energy functional $\mathcal{F}$, let us briefly recall some well-known facts about the ideal (non-interacting) trapped Bose gas, that is, the system described by the Hamiltonian in \eqref{eq:Hamiltonian} with $v=0$. In the following we assume that the trapping potential $w$ satisfies 
\begin{equation}
	\lim_{|x| \to \infty} \frac{w(x)}{L |x|^s } = 1
	\label{eq:asymptw}
\end{equation}
for two constants $s > 0$ and $L > 0$. The one-particle density matrix $\gamma^{\mathrm{id}}$ of the Gibbs state $\Gamma_{\mathrm{G}}$ of the ideal gas equals
\begin{equation}
	\gamma^{\mathrm{id}} = \frac{1}{\exp((-\Delta + w -\mu_N)/T)-1} \, .
	\label{eq:idealgas1-pdm}
\end{equation}
The chemical potential $\mu_N$ satisfies $\mu_N < e$, where $e$ denotes the lowest eigenvalue of $-\Delta + w$, and is chosen s.t. $\tr \gamma^{\mathrm{id}} = N$ holds. By definition, the system displays BEC if and only if the largest eigenvalue $N_0(T,N)$ of $\gamma^{\mathrm{id}}$ is of order $N$ in the limit $N \to\ \infty$, that is, if and only if
\begin{equation}
	\lim_{N \to \infty} \frac{N_0(T,N)}{N} = g > 0 \, , \quad \text{ where } \quad N_0(T,N) = \frac{1}{\exp((e-\mu_N)/T)-1} \, .
	\label{eq:condfract}
\end{equation}
The trapped ideal Bose gas is well-known to exhibit a BEC phase transition with critical temperature given by
\begin{equation}
	T_{\mathrm{c}}(s) = \frac{N^{1/\alpha}}{(\kappa \alpha \Gamma(\alpha) \zeta(\alpha) )^{1/{\alpha}}} \, , \quad \text{ where } \quad \kappa = \frac{2 L^{-3/s}}{3 \pi} \int_{0}^1 (1-x)^{3/2} x^2 \de x \quad \text{ and } \quad \alpha = \frac{6+3s}{2s} \, .
	\label{eq:idealgascrittemp}
\end{equation}
Here $\Gamma(\alpha)$ denotes the Gamma function and $\zeta(\alpha)$ is the Riemann zeta function. More precisely, let us define $t_{\mathrm{c}}(s) = [(\kappa \alpha \Gamma(\alpha) \zeta(\alpha) )^{1/{\alpha}}]^{-1}$, that is, $t_{\mathrm{c}}(s) = N^{-1/\alpha} T_{\mathrm{c}}(s)$. If $T = \lambda N^{1/\alpha}$ with $\lambda \geq 0$, then we have $g = [1-(\lambda/t_{\mathrm{c}}(s))^{\alpha}]_+$ with $[x]_+ = \max\{ x,0 \}$. This statement follows from the definition of $N_0(T,N)$ and the eigenvalue asymptotics of $-\Delta + w$, see e.g. \cite[Theorem~XIII.81]{RS4}\footnote{The theorem in the reference is stated for $s>1$ but the proof applies for all $s>0$.}. 

An important property of the ideal gas at temperature $T = \lambda N^{1/\alpha}$ with fixed $ \lambda \in (0,t_{\mathrm c}(s)) $ is that its density $\gamma^{\mathrm{id}}(x,x)$ displays a two-scale structure. The condensed particles are described by the ground state of $-\Delta + w$. Accordingly, their density varies on a length scale of order one, independent of $N$. Excited particles (the thermal cloud) are characterized by the much larger length scale $R \gg 1$ because their energy per particle is of the order $T$. The length $R$ can be obtain by equating $T = R^s$, that is, $R \sim N^{2/(6+3s)}$. The density $\varrho_{\mathrm{th}}$ of the thermal cloud is proportional to $\varrho_{\mathrm{th}} \sim N/R^3 \sim T^{3/2} \sim N^{3s/(6+3s)} \ll N$, that is, the thermal cloud is much more dilute than the condensate. 

As recently shown in \cite{DeuSei2021}, this two scale structure is preserved if we add a weak interaction to the Hamiltonian, as in \eqref{eq:Hamiltonian}. To make this statement precise, we introduce for a function $\phi$ in the form domain of $-\Delta + w$ and $g \in [0,1]$ the Hartree energy functional 
\begin{equation}
	\mathcal{E}^{\mathrm{H}}(\phi) = \langle \phi, (-\Delta + w) \phi \rangle + \frac{g}{2} \int_{\mathbb{R}^6} |\phi(x)|^2 v(x-y) | \phi(y) |^2 \de x \de y \,.
	\label{eq:Hartreefunctional}
\end{equation}
In the following we assume that $v$ is such that the functional $\mathcal{E}^{\mathrm{H}}$ is bounded from below and admits a unique minimizer in the set of functions with $L^2(\mathbb{R}^3)$-norm equal to one. We denote the minimum and the minimizer of $\mathcal{E}^{\mathrm{H}}$ by $E^{\mathrm{H}}(g)$ and $\phi^{\mathrm{H}}$, respectively. Using Theorem~1.3 and Lemma~7.1 in \cite{DeuSei2021}, it can be shown\footnote{The result in the reference is stated for the special choice $w(x) = x^2$. However, the techniques that have been used to prove this result also apply to the more general trapping potentials we introduced in the discussion of the ideal gas.}  that for $0 \leq T \lesssim T_{\mathrm{c}}(s)$, with $T_{\mathrm{c}}(s)$ defined in \eqref{eq:idealgascrittemp} as the critical temperature of the ideal trapped gas, the free energy in \eqref{eq:freeenergy} satisfies
\begin{equation}
	\lim_{N \to \infty} N^{-1} \left| F(T,N) - F_0(T,N) - N_0(T,N) E^{\mathrm{H}}(g) \right| = 0 \,.
	\label{eq:freeenergyasymptotics}
\end{equation}
Here $N_0(N,T)$ and $g$ are chosen as in \eqref{eq:condfract}, and $F_0(T,N)$ denotes the free energy in \eqref{eq:freeenergy} with $v = 0$. The effect of the interaction can be observed, to leading order in $N$, only in the condensate because the thermal cloud is much more dilute. Another consequence of the results in \cite{DeuSei2021} is that for any approximately minimising sequence $\Gamma_N \in \mathcal{S}_N$ satisfying 
\begin{equation}
	\lim_{N \to \infty} N^{-1} \left| \mathcal{F}(\Gamma_N) - F_0(T,N) - N_0(T,N) E^{\mathrm{H}}(g) \right| = 0 \,,
	\label{eq:asymptoticsapproxmin}
\end{equation}
the corresponding 1-pdm $\gamma_N$ is such that 
\begin{equation}
	\lim_{N \to \infty} N^{-1} \left\Vert \gamma_N - \gamma^{\mathrm{id}} Q - N_0(T,N) | \phi^{\mathrm{H}} \rangle \langle \phi^{\mathrm{H}} | \ \right\Vert_{\mathcal{L}^1} = 0 
	\label{eq:asymptotics1pdm}
\end{equation}
holds, where $Q = 1 - | \psi_0 \rangle \langle \psi_0 |$ with the ground state $\psi_0$ of $-\Delta + w$. That is, to leading order in $N$, $\gamma_N$ is the 1-pdm of the ideal gas, where the eigenfunction of its largest eigenvalue (the condensate wave function) has been replaced by the minimizer of the Hartree energy functional. This, in particular, implies that the interacting systems displays a BEC phase transition with critical temperature given by the one of the ideal gas to leading order. It also indicates that the density of the Gibbs state $\Gamma_{\mathrm{G}}$ has the same two-scale structure as that of the ideal trapped gas. 
\subsection{Construction of the initial data} 
\label{sec:initialdata}
In this section we define our class of initial states. Our definition is motivated by the discussion in Section~\ref{sec:Equilibriumstates}.

Let $\phi \in L^2(\mathbb{R}^3)$ and let $\gamma \in \mathcal{B}(L^2(\mathbb{R}^3))$ be a positive trace class operator. We think of $\phi$ as a condensate wave function and of $\gamma$ as a 1-pdm describing excitations of the condensate. The particle number $n(\phi,\gamma)$ of the pair $(\phi,\gamma)$ is defined by
\begin{equation}
	n(\phi,\gamma) = \int_{\mathbb{R}^3} | \phi(x) |^2 \de x + \tr \gamma \,.
	\label{eq:particleconstraint}
\end{equation}
Let $G$ be the unique quasi-free state on $\mathscr{F}$ that satisfies $[G,\mathcal{N}] = 0$ and has $\gamma$ as 1-pdm. We also define the Weyl transformation
\begin{equation}
	W(\phi) = \exp( a^*(\phi) - a(\phi) ) \,.
	\label{eq:Weyltrafo}
\end{equation}
The above definitions allow us to associate to each pair $(\phi,\gamma)$ the state 
\begin{equation}
	\Gamma(\phi,\gamma) = W(\phi) G W(\phi)^*
	\label{eq:stateclean}
\end{equation}
on the bosonic Fock space. The kernel of its 1-pdm is given by
\begin{equation}
	\tr[ a_y^* a_x \Gamma(\phi,\gamma) ] = \phi(x) \overline{\phi(y)} + \gamma(x,y) \,,
	\label{eq:1pdmstateclean}
\end{equation}
which, in particular, implies $\tr[ \mathcal{N} \Gamma(\phi,\gamma) ] = n(\phi,\gamma)$. The state $\Gamma (\phi,\gamma)$ is still a quasi-free state, in the sense that higher order correlations can be computed through the Wick theorem, summing over all possible partitions in groups of one or two creation and annihilation operators. Choosing $\phi = \sqrt{ N_0 (T,N) } \phi^{\mathrm{H}}$ and $\gamma = \gamma^{\mathrm {id}} Q$, it is simple to check that $\Gamma (\phi, \gamma)$ provides a good approximation for the free energy of the system in the sense of \eqref{eq:asymptoticsapproxmin}. In our main theorem, we will describe the evolution in \eqref{eq:solutionschroedingereq} for initial data given by suitable perturbations of a quasi-free state of the form $\Gamma (\phi, \gamma)$, under appropriate assumptions on $(\phi,\gamma)$.


To this end, we first use the Araki--Woods representation, see \cite{ArakiWoods1963,DerGer2013}, to recast a mixed state $\Gamma \in \mathcal{S}_N$ as a vector $\Psi$ on the doubled Fock space $\mathscr{F} \otimes \mathscr{F}$. This will allow us to apply the formalism of Bogoliubov transformations and fluctuation dynamics. Note that a similar strategy was used in \cite{Fermions} to study the mean-field evolution of fermionic mixed states. The vector $\Psi$ is conveniently defined in terms of the spectral decomposition $\Gamma = \sum_{\alpha = 1}^{\infty} \lambda_{\alpha} | \psi_{\alpha} \rangle \langle \psi_{\alpha} |$ as
\begin{equation}
	\Psi = \sum_{\alpha=1}^{\infty} \lambda_{\alpha}^{1/2} \psi_{\alpha} \otimes \overline{\psi_{\alpha}}, 
	\label{eq:Araki-Wyss}
\end{equation}
where we define the complex conjugate of $\psi_{\alpha}$ in the position-space representation. It satisfies
\begin{equation}\label{eq:Araki-Wyss-main-property}
	\tr[A \Gamma] = \langle \Psi, \left( A \otimes \mathds{1} \right) \Psi \rangle
\end{equation}
for $A \in \mathcal{B}(\mathscr{F})$. It should be noted that this representation is far from being unique. When we replace the vectors $\psi_{\alpha}$ in the second tensor factor in \eqref{eq:Araki-Wyss} by elements of another orthonormal basis of $\mathscr{F}$, we obtain an equivalent representation of $\Gamma(\phi,\gamma)$ on $\mathscr{F} \otimes \mathscr{F}$ in the sense that \eqref{eq:Araki-Wyss-main-property} continues to hold. 

For the next step in our construction we need the unitary equivalence 
\begin{equation*}
	\mathcal{U} \mathscr{F}(L^2(\mathbb{R}^3)) \otimes \mathscr{F}(L^2(\mathbb{R}^3)) = \mathscr{F}(L^2(\mathbb{R}^3) \oplus L^2(\mathbb{R}^3))
\end{equation*}
defined by the relations $\mathcal{U} \Omega \otimes \Omega = \widetilde{\Omega}$, where $\widetilde{\Omega}$ denotes the vacuum vector of $\mathscr{F}(L^2(\mathbb{R}^3) \oplus L^2(\mathbb{R}^3))$, 
\begin{equation}
	\mathcal{U} \big( a(f) \otimes \mathds{1} \big) \cU^* = a( f \oplus 0 ) \eqqcolon a_{\mathrm{\ell}}(f), \quad \text{and} \quad \mathcal{U} \big( \mathds{1} \otimes a(f) \big) \cU^*= a( 0 \oplus f ) \eqqcolon a_{r}(f). 
	\label{eq:Utrafo}
\end{equation}

Thus, every state $\Gamma \in \mathcal{S}_N$ can be represented by a normalized vector $\widetilde{\Psi} = \cU \Psi \in \mathscr{F} (L^2(\mathbb{R}^3) \oplus L^2(\mathbb{R}^3))$ satisfying 
\begin{equation*}
	\tr [ A \Gamma ] = \langle \wt{\Psi} \cU (A \otimes \mathds{1} ) \cU^* \wt{\Psi} \rangle
\end{equation*}
If $A$ is expressed in terms of creation and annihilation operators $a^*, a$, then $\cU (A \otimes 1) \cU^*$ is obtained from $A$ by replacing all operators $a^*$ with $a^*_{\ell}$ and all operators $a$ with $a_{\ell}$.  Let us now consider, in particular, the state \eqref{eq:stateclean}. An important observation is that the quasi-free state $G$ can be described, on $\mathcal{F} (L^2(\mathbb{R}^3) \oplus L^2(\mathbb{R}^3))$, by the vector $\wt{\Psi}_G = \cT (\gamma) \wt \Omega$, where
\begin{equation}\label{eq:def_T(gamma)}
	\mathcal{T}(\gamma) = \exp\left( \int_{\mathbb{R}^6} k_{\gamma}(x,y) a^*_{\mathrm{\ell},x} a^*_{r,y} \de x \de y - \text{h.c.} \right) 
\end{equation} 
and $k_\gamma (x,y)$ denotes the integral kernel of the operator $k_\gamma = \mathrm{arcsinh} (\sqrt{\gamma})$. We observe that $\cT (\gamma)$ is a Bogoliubov transformation and satisfies the relations 
\begin{equation}
	\mathcal{T}^*(\gamma) a_{\ell}(f) \mathcal{T}(\gamma) = a_{\ell}(u f) + a_{r}^*(\overline{v f}) \,, \hspace{1.5cm} \mathcal{T}^*(\gamma) a_{r}(f) \mathcal{T}(\gamma) = a_{r}(u f) + a_{\ell}^*(\overline{v f}) \,,
	\label{eq:Bogtrafo}
\end{equation}
where $u = \sqrt{1+\gamma}$ and $v = \sqrt{\gamma}$. With \eqref{eq:Bogtrafo}, it is clear that $\wt{\Psi}_G$ is quasi-free (it satisfies the Wick theorem) and that it has the same 1-pdm as $G$, in the sense that 
\begin{equation*}
	\langle  \wt{\Psi}_G, a_{\ell,y}^* a_{\ell,x} \wt{\Psi}_G \rangle = \gamma(x,y) \,.
\end{equation*}
Recalling \eqref{eq:stateclean}, we find that the quasi-free state $\Gamma (\phi, \gamma)$, now with a condensate described by the wave function $\phi$, can be described, on $\mathscr{F} (L^2(\mathbb{R}^3) \oplus L^2(\mathbb{R}^3))$ by the vector $\wt{\Psi} (\phi, \gamma) = \cW (\phi) \cT (\gamma) \wt{\Omega}$ (sometimes called the purification of $\Gamma (\phi, \gamma)$), in the sense that 
\begin{equation}
	\tr[ A \Gamma(\phi,\gamma) ] = \langle \widetilde{\Omega}, \mathcal{T}^*(\gamma) \mathcal{W}^*(\phi) \mathcal{U} \left( A \otimes \mathds{1} \right) \mathcal{U}^* \mathcal{W}(\phi) \mathcal{T}(\gamma) \widetilde{\Omega} \rangle \,.
	\label{eq:Gamma}
\end{equation}
Here, we introduced, on $\mathscr{F} (L^2(\mathbb{R}^3) \oplus L^2(\mathbb{R}^3))$, the Weyl operator 
\begin{equation}
	\mathcal{W}(\phi) = \exp\left( a^*_{\mathrm{\ell}}(\phi) + a^*_{r}(\overline{\phi}) - \mathrm{h.c.} \right)
	\label{eq:doubleFPtrafos}
\end{equation}
satisfying the relations 
\begin{equation}
	\mathcal{W}^*(\phi) a_{\ell}(f) \mathcal{W}(\phi) = a_{\ell}(f) + \langle f, \phi \rangle \,, \hspace{1.5cm} \mathcal{W}^*(\phi) a_{r}(f) \mathcal{W}(\phi) = a_{r}(f) + \langle f, \overline{\phi} \rangle \,.
	\label{Weyltrafo}
\end{equation}
The validity of \eqref{eq:Gamma} can be checked similarly as the comparable equality for $G$, using \eqref{eq:Bogtrafo} and \eqref{Weyltrafo}.

In our main theorem, we will study the dynamics of initial states $\Gamma_{\xi}(\phi, \gamma)$, represented, on $\mathscr{F} (L^2(\mathbb{R}^3) \oplus L^2(\mathbb{R}^3))$, by the vector $\mathcal{W} (\phi) \mathcal{T} (\gamma) \xi$, for a normalized $\xi \in \mathscr{F} (L^2(\mathbb{R}^3) \oplus L^2(\mathbb{R}^3))$. That is, the state $\Gamma_\xi (\phi, \gamma)$ is defined so that, for every observable $A \in \mathcal{B} (\mathscr{F})$, 
\begin{equation}
	\tr[ A \Gamma_{\xi}(\phi,\gamma) ] = \langle \xi, \mathcal{T}^*(\gamma) \mathcal{W}^*(\phi) \mathcal{U} \left( A \otimes \mathds{1} \right) \mathcal{U}^* \mathcal{W}(\phi) \mathcal{T}(\gamma) \xi \rangle \,.
	\label{eq:Gammaxi}
\end{equation}
Our assumptions on $\phi, \gamma, \xi$ will make sure that $\Gamma_\xi (\phi , \gamma)$ is a perturbation of the quasi-free state $\Gamma (\phi, \gamma)$ (in particular, we will require the expectation of the number of particles operator and sufficiently many of its moments to stay bounded, in the state $\xi$, uniformly in $N$).
From \eqref{eq:solutionschroedingereq}, we conclude that the evolution $\Gamma_{\xi, t} (\phi, \gamma) = e^{- \mathrm{i} \mathcal{H}_N t} \Gamma_\xi (\phi, \gamma) e^{\mathrm{i} \mathcal{H}_N t}$ of the initial data $\Gamma_\xi (\phi, \gamma)$ is such that 
\begin{equation}
	\tr[ A \Gamma_{\xi,t}(\phi,\gamma) ] = \langle \xi, \mathcal{T}^*(\gamma) \mathcal{W}^*(\phi) \exp(\mathrm{i} \mathcal{L}_N t) \mathcal{U} \left( A \otimes \mathds{1} \right) \mathcal{U}^* \exp(-\mathrm{i} \mathcal{L}_N t) \mathcal{W}(\phi) \mathcal{T}(\gamma) \xi \rangle \,,
	\label{eq:Gammaxit}
\end{equation} 
with the Liouvillian
\begin{equation}
	\mathcal{L}_N = \mathcal{H}_{N,\mathrm{\ell}} - \mathcal{H}_{N,r} \,.
	\label{eq:Liouvillian}
\end{equation}
The operators $\mathcal{H}_{N,\mathrm{\ell}}$ and $\mathcal{H}_{N,r}$ are defined as the original Hamiltonian $\mathcal{H}_N$ in \eqref{eq:Hamiltonianti} with $a_x^*, a_x$ replaced by $a_{\mathrm{\ell},x}^*, a_{\mathrm{\ell},x}$ and $a_{r,x}^*, a_{r,x}$, respectively. In other words, on $\mathscr{F} (L^2(\mathbb{R}^3) \oplus L^2(\mathbb{R}^3))$, we are interested in the evolution $e^{-\mathrm{i} \ \mathcal{L}_N t} \mathcal{W} (\phi) \mathcal{T} (\gamma) \xi$ generated by $\mathcal{L}_N$ on the initial data $\mathcal{W} (\phi) \mathcal{T} (\gamma) \xi$. 

Motivated by the properties of approximate equilibrium states of trapped Bose gases with a trapping potential satisfying \eqref{eq:asymptw} for some $s > 0$, we will make the following assumptions on the pair $(\phi , \gamma)$.

\begin{assumption}
	\label{ass:scales}
	We assume that the pair $(\phi,\gamma)$ is s.t. $\phi \in H^3(\mathbb{R}^3)$ and $\tr[ (1-\Delta)^{3/2} \gamma (1-\Delta)^{3/2} ] < +\infty$. Moreover, $n(\phi,\gamma) = N$ and the following holds for $s>0$:

	\textbf{(A)} The condensate wave function $\phi$ can be written as $\phi = c(N) \wt{\phi}$, with a fixed normalized $\widetilde{\phi} \in L^2 (\bR^3)$. The $N$-dependent constant $c(N)$ determines the expected number of particles in the condensate. 
	
	\textbf{(B)} The 1-pdm $\gamma$ obeys 
	\begin{equation*}
		\int_{\mathbb{R}^6} | \hat{\gamma}(p,q) | \de p \de q \lesssim T_{\mathrm{c}}^{3/2}(s)
	\end{equation*}
	with $T_{\mathrm{c}}(s)$ as in \eqref{eq:idealgascrittemp}, and where $\hat{\gamma}(p,q)$ denotes the integral kernel of $\gamma$ in Fourier space. 
	
	\textbf{(C)} The operator norm of $\gamma$ satisfies
	\begin{equation*}
		\Vert \gamma \Vert \lesssim T_{\mathrm{c}}(s).
	\end{equation*}
\end{assumption}

Part \textbf{(A)} in the above assumption allows us to choose $\widetilde{\phi} = \phi^H$ minimising \eqref{eq:Hartreefunctional} with $g$ in \eqref{eq:condfract}, and $c(N) = \sqrt{N_0(N,T)}$, as defined in \eqref{eq:condfract}. Point \textbf{(B)} implies that the density satisfies $|\gamma (x,x)| \lesssim T_{\mathrm{c}}^{3/2}(s)$, uniformly in $x \in \mathbb{R}^3$. In fact, it also implies (and this will be important for us) a bound (uniform in $t$) for the density of the free evolution $e^{\mathrm{i} \Delta t} \gamma e^{-\mathrm{i} \Delta t}$ and more generally 
\begin{equation*}
	\sup_{x,y \in \mathbb{R}^3} | (e^{\mathrm{i} \Delta t} \gamma e^{- \mathrm{i} \Delta t})(x,y) | \leq \int_{\mathbb{R}^6} | \hat{\gamma}(p,q) | \de p \de q \lesssim T_{\mathrm{c}}^{3/2}(s) \,.
\end{equation*}
As we will see in Section~\ref{sec:assumptions}, point \textbf{(B)} is satisfied by $\gamma^{\mathrm{id}} Q$ with the 1-pdm $\gamma^{\mathrm{id}}$ of the ideal gas in \eqref{eq:idealgas1-pdm} with the choice $w(x) = |x|^s$ and with $Q = 1 - |\psi_0 \rangle \langle \psi_0 |$, where $\psi_0$ denotes the ground state of $-\Delta + |x|^s$. Point \textbf{(C)} in Assumption~\ref{ass:scales} is satisfied by $\gamma^{\mathrm{id}} Q$, too. 

We conclude this section with a brief discussion of the relevant scales related to our time evolution. In the whole discussion we assume that $T = \kappa T_{\mathrm{c}}(s)$ with $\kappa \in (0,1)$. This in particular implies that we have order $N$ particles in the condensate and order $N$ particles in the thermal cloud. We are interested in following the dynamics of the condensate, and therefore consider times of order one. This is motivated by the fact that the energy per particle in the condensate and the related length scale are both of order one, independent of $N$, see Section~\ref{sec:Equilibriumstates}. In contrast, the energy per particle in the thermal cloud is proportional to $T$ and the related velocity is of the order $T^{1/2} \sim N^{s/(6+3s)} \gg 1$. Thus, particles in the thermal cloud move much faster than particles in the condensate. For particles in the condensate, the $N^{-1}$ factor appearing in \eqref{eq:Hamiltonianti} in front of the interaction produces a mean-field potential and leads to the Hartree equation. In contrast, the typical length scale $R$ of the thermal cloud scales as $N^{2/(6+3s)} \gg 1$ (see Section~\ref{sec:Equilibriumstates}); hence the thermal cloud is much more dilute and we can expect that, to leading order in $N$, it moves according to the free evolution. Our main result confirms this heuristic picture. 
\subsection{Main result}
\label{sec:main}
Our main result is captured in the following theorem.

\begin{theorem}\label{thm:main_2}
	We assume that $v \in L^1(\mathbb{R}^3) \cap W^{1,p}(\mathbb{R}^3)$ with $p>3$ satisfies $v(x) = v(-x)$ for a.e. $x \in \mathbb{R}^3$ and that $\hat{v} \in L^1(\mathbb{R}^3)$, where $\hat{v}$ denotes the Fourier transform of v. Moreover, we assume that the pair $(\phi,\gamma)$ satisfies Assumption~\ref{ass:scales} with $0 < s \leq 3/2$ and that the fluctuation vector $\xi$ is s.t. $\langle \xi, (\mathcal{N}_{\ell} + \mathcal{N}_{r})^{44} \xi \rangle \lesssim 1$ holds with $\mathcal{N}_{\ell/r} = \int a^*_{\ell/r,x} a_{\ell/r,x} \de x$. Then there exists a constant $c>0$ s.t. the 1-pdm $\gamma_{\x,t}$ of the state $\Gamma_{\xi,t}(\phi,\gamma)$ defined in \eqref{eq:Gammaxit} satisfies
	\begin{equation}
		\left\Vert \gamma_{\x,t} - e^{\mathrm{i} \Delta t} \gamma e^{-\mathrm{i} \Delta t} - | \phi_t \rangle \langle \phi_t | \  \right\Vert_{\cL^1} \lesssim \sqrt{ N } T_{\mathrm{c}}^{3/4}(s) \exp(c \exp(c \exp(c t))),
		\label{eq:thmbound1pdm}
	\end{equation} 
	where $\Vert \cdot \Vert_{\cL^1}$ denotes the trace norm. The function $\phi_t$ is the solution to the time-dependent Hartree equation
	\begin{equation}
		\mathrm{i} \partial_t \phi_t(x) = \left( -\Delta + N^{-1} v \ast | \phi_t(x) |^2 \right) \phi_t(x) \quad \text{ with initial datum } \quad \phi_0(x) = \phi(x).
		\label{eq:timedepHartree}
	\end{equation}
\end{theorem}

We have the following remarks concerning Theorem~\ref{thm:main_2}.

\begin{remark}\label{rmk:restriction_s} 
	\begin{enumerate}
		\item The r.h.s. of \eqref{eq:thmbound1pdm} scales for $0<s \leq 3/2$ as $N^{(1+s)/(2+s)} \leq N^{5/7} \ll N$, which should be compared to $\tr \gamma_{\x,t} = N = \tr[ e^{\mathrm{i} \Delta t} \gamma e^{-\mathrm{i} \Delta t} + | \phi_t \rangle \langle \phi_t | ]$. That is, the 1-pdm $e^{\mathrm{i} \Delta t} \gamma e^{-\mathrm{i} \Delta t} + | \phi_t \rangle \langle \phi_t |$ yields a good approximation for $\gamma_{\x,t}$ for fixed $t>0$ and large $N$.
		\item The time-dependence of the r.h.s. of \eqref{eq:thmbound1pdm} can be replaced by $C_{\epsilon} \exp( c \exp(c(t+t^{3+\epsilon}))) $ for $\epsilon > 0$ provided we have either $v \geq 0$ or $\hat{v} \geq 0$. More details can be found in Remark~\ref{remark:HartreeH2} in Section~\ref{sec:apprximationHFB} below. 
		\item We can obtain a better approximation for $\gamma_{\x,t}$ by letting $\omega = |\phi \rangle \langle \phi| + \gamma$ and considering the solution to the time-dependent Hartree equation
		\begin{equation}\label{eq:Hartree_1pdm}
			\mathrm{i} \partial_t \omega_t = [ -\Delta + N^{-1} v * \rho_{\omega_t}(x)  , \omega_t ]
		\end{equation}  
		with $\omega_{t=0} = \omega$ and $\rho_{\omega_t} (x) = \omega_t (x,x)$. Then, proceeding as in our proof of Theorem~\ref{thm:main_2}, we could show that 
		\begin{equation*}
			\Vert \gamma_{\x,t} - \omega_t \Vert_{\mathcal{L}^1} \lesssim \sqrt{ N T_{\mathrm{c}}(s) } \exp(c \exp(c \exp(c t))) \,.
		\end{equation*} 
		In this case the dependence on $N$ on the r.h.s. is optimal, already for $t = 0$. In fact, taking $\xi = \mathcal{U} (1+a_{r}^*(\overline{\phi}/\Vert \phi \Vert)) \Omega/\sqrt{2}$, with $\mathcal U$ defined in \eqref{eq:Utrafo}, a short computation shows that 
		\begin{equation}
			\tr[ \mathcal{N} \Gamma_{\xi}(\phi,\g)] - \tr[\mathcal{N} \Gamma_{\Omega}(\phi,\g)] \geq \frac{\langle \phi, \sqrt{\gamma} \phi \rangle}{\Vert \phi \Vert} \,. 
			\label{eq:finalchanges1}
		\end{equation} 
		Assume $\Vert \gamma \Vert \sim T_c (s)$ and let $\psi$ denote the normalized eigenfunction of $\gamma$ associated with its largest eigenvalue. Choosing $\phi = c(N) \psi$ so that $\Vert \phi \Vert \sim \sqrt{N}$, the r.h.s. of \eqref{eq:finalchanges1} is of order $\sqrt{N T_c(s)}$.

		\item The assumption for $\xi$ appearing in Theorem is needed for $s=3/2$. For $s < 3/2$, less moments are sufficient. 
		\item The theorem holds, in particular, for the following choice of $(\phi,\gamma)$: Let $h = -\Delta + |x|^s$ with $0<s\leq3/2$. The 1-pdm $\gamma$ is defined by $\gamma = \gamma^{\mathrm{id}} Q$ with $\gamma^{\mathrm{id}}$ in \eqref{eq:idealgas1-pdm}, where the projection $Q$ removes the largest eigenvalue of $\gamma^{\mathrm{id}}$ (the condensate). Moreover, the temperature $T$ in the definition of $\gamma^{\mathrm{id}}$ obeys $T \lesssim T_{\mathrm{c}}(s)$. The condensate wave function $\phi$ is given by $\sqrt{N-\tr[\gamma^{\mathrm{id}} Q]}$ times a minimizer of the Hartree energy functional in \eqref{eq:Hartreefunctional}. In Section~\ref{sec:assumptions} we show that the pair $(\phi,\gamma)$ satisfies all assumptions of Theorem~\ref{thm:main_2}.
		\item The condition  $0 < s \leq 3/2$ is needed to consider initial data like those described in Remark~$5$, for $T = \kappa T_c (s)$, $\kappa \in (0,1)$. This restriction is required, because, if $s > 3/2$, the largest eigenvalue of $\gamma^{\mathrm{id}} Q$ grows so fast with $N$ that our analysis breaks down. It is possible to consider traps with $s > 3/2$, in particular the harmonic trap $w (x) = x^2$, if we restrict our attention to temperatures $T \lesssim T_c (3/2) \ll T_c (2)$. Moreover, our analysis also applies above the critical temperature, where $\phi = 0$ and $\| \gamma \| \sim 1$; in this case, we can take any $s \in (0;\infty)$ and we only need weaker assumptions on $\xi$ (and we get better bounds). 
	\end{enumerate}
\end{remark}
\subsection{Proof strategy and organization of the article}
\label{sec:proofstrategy}
For the convenience of the reader we give in this section a short summary of the main steps leading to a proof of Theorem~\ref{thm:main_2}. 




In Section~\ref{sec:assumptions} we show that the initial data in Remark~\ref{rmk:restriction_s}.5, consisting of the minimizer of the Hartree energy functional in \eqref{eq:Hartreefunctional} and the 1-pdm of the ideal gas trapped by the potential $w(x) = |x|^s$ with $s \in (0,3/2]$, satisfies Assumption~\ref{ass:scales}.

Our many-body analysis is based on the definition of fluctuation dynamics around the dynamics generated by the Hartree--Fock--Bogoliubov (HFB) equations, see e.g. \cite{BachBretaux,BenedikterSolovej_Bog_de_Gennes}. We do not give a derivation of these equations from quantum mechanics but rather use them as a technical tool. In Section~\ref{sec:properties_HFB} we introduce the HFB equations, we recall their well-posedness, as established in \cite{BachBretaux}, and we show how they can be approximated by the much simpler effective dynamics appearing in Theorem~\ref{thm:main_2}, if the initial data satisfies Assumption~\ref{ass:scales}. We also prove a bound for solutions of the HFB equations that guarantees the diluteness of the thermal cloud during the time evolution and is a crucial ingredient for our many-body analysis.


Section~\ref{sec:fluctdyn} is devoted to the construction of the fluctation dynamics around the HFB equations and to the proof of Theorem~\ref{thm:main_2}. We start by recalling some basic facts about Weyl and Bogoliubov transformations in Section~\ref{sec:weylandBogtrafos}, and afterwards define the fluctuation dynamics in Section~\ref{sec:constrcutionfluctdyn}. In Section~\ref{sec:Boundbyparticlenumber} we prove a bound for the trace-norm difference of $\gamma_{\xi,t}$ and the 1-pdm related to the solution to the HFB equations in terms of the expected number of excitations in the time-dependent fluctuation vector. A bound for this quantity is stated without proof in Proposition~\ref{mainprop} in Section~\ref{sec:proofoftheorem}. We end Section~\ref{sec:fluctdyn} by showing how Proposition~\ref{mainprop} implies Theorem~\ref{thm:main_2}. 

In Section~\ref{sec:studyfluctdyn} we give a proof of Proposition~\ref{mainprop} that is based on a Gr\"onwall argument for the expected number of particles in the fluctuation dynamics, and we start our analysis in Section~\ref{sec:generator} with the computation of its generator. In Section~\ref{sec4.1} we introduce a fluctuation dynamics with a cut-off in the particle number and prove the equivalent of Proposition~\ref{mainprop} for this dynamics. Section~\ref{sec:weakbounds} is devoted to the proof of weak a-priori bounds for the original fluctation dynamics without a cut-off, that are later used in Section~\ref{sec:comparisonfluctuationdynamics} to show that the two fluctuation dynamics are close in a suitable sense. Finally, in Section~\ref{sec:proofmainprop} we use the results of the previous sections to give a proof of Proposition~\ref{mainprop}, thus concluding the proof of Theorem~\ref{thm:main_2}.
\section{Properties of the 1-pdm of the ideal gas and of Hartree minimizers}
\label{sec:assumptions}

In this section we provide an example of a physically relevant initial pair $(\phi, \gamma)$ satisfying Assumption~\ref{ass:scales}. To this end, we choose $\gamma$ as the 1-pdm of the ideal gas described by the Schr\"odinger operator $h = -\Delta + |x|^s$ for some $s\in(0,2]$ and $\phi$ as the minimizer of the Hartree functional \eqref{eq:Hartreefunctional}, again with $h = -\Delta + |x|^s$ and with a sufficiently regular interaction potential $v$ (Theorem \ref{thm:main_2} describes the evolution of perturbations of the state $\Gamma(\phi,\gamma)$ in \eqref{eq:stateclean}, if $s \in (0;3/2]$). The next proposition shows the validity of Remark~\ref{rmk:restriction_s}.5. 

\begin{proposition} \label{prop:gaphi} 
	Let $h = -\Delta + |x|^s$, for some $s \in [0,2]$ and 
	\[ \gamma = \frac{1}{e^{\beta (h-\mu)} - 1} Q \]
	with a chemical potential $\mu \in \bR$ satisfying $h > \mu$ and with $Q = 1 - |\psi_0 \rangle \langle \psi_0|$ denoting the projection on the orthogonal complement of the ground state $\psi_0$ of $h$. Then, for every $\beta > 0$ there exists a constant $C < \infty$ (depending on $\beta$), with 
	\begin{equation}\label{eq:prop1} \tr \, (1-\Delta)^{3/2} \, \gamma (1-\Delta)^{3/2} \leq C \,. \end{equation} 
	Moreover, denoting by $\hat{\gamma} (p,q)$ the integral kernel of $\gamma$ in Fourier space, we find $\| \gamma \| \lesssim T_{\mathrm c} (s)$ and  
	\begin{equation}\label{eq:prop2} \int_{\bR^6} |\hat{\gamma} (p,q)| \, \de p \de q \lesssim T_{\mathrm c}^{3/2} (s) \end{equation} 
	for all $\beta > 0$ with $\beta^{-1} \lesssim T_{\mathrm c} (s)$. 
	
	Let $v \in W^{1,p} (\bR^3)$ for some $p > 3$ and let $\phi \in H^1 (\bR^3)$ with $\| \phi \| = 1$ solve the Hartree equation 
	\begin{equation}\label{eq:hartree} \Big( -\Delta + |x|^s + g (v * |\phi|^2) (x) \Big) \phi (x) = \mu^{\mathrm{H}} \phi (x) \end{equation} 
	with $g \in [0,1]$ in the sense of distributions, for some $\mu^H \in \bR$. Then we have $\| \phi \|_{H^3} \lesssim (1+ \Vert v \Vert_{\infty} + \mu^{\mathrm{H}})^{3/2}$. 
\end{proposition} 

In order to prove  Prop. \ref{prop:gaphi}, we will make use of the following lemma. 
\begin{lemma}
	\label{lem:secreg1}
	Let $w \in W^{1,1}_\text{loc} (\bR^3)$ be real-valued, with $w (x) \geq - C$, for some $C \geq 0$. Let $h_w = -\Delta + w$, denote $\kappa = 1+C$, and assume the operator inequality   
	\begin{equation}\label{eq:nabw} |\nabla w|^2 \lesssim (\kappa + h_w) \,. \end{equation} 
	Then, we have 
	\begin{equation}\label{eq:lm1} (1 - \Delta)^2 , w^2 \lesssim (\kappa+h_w)^2, \qquad (1-\Delta)^3 \lesssim (\kappa+h_w)^3 \,.
	\end{equation} 
\end{lemma} 
\textit{Remark:} In applications, we will use Lemma \ref{lem:secreg1}, with $w (x) = |x|^s$ and with $w(x) = |x|^s + g (v * |\phi|^2) (x)$, for $s \in [0,2]$. 
\begin{proof} 
	Pick $u \in \mathcal{C}_{\mathrm{c}}^{\infty}(\mathbb{R}^3)$. A short computation shows
	\begin{equation}
		\langle u, (\kappa -\Delta + w)^2 u \rangle = \langle u, [ (1-\Delta)^2 + (C+w)^2 ] u \rangle + 2 \langle \nabla u, (C+w) \nabla u \rangle + 2 \text{Re} \langle u, (\nabla w) \nabla u \rangle + 2 \langle u, (C+w) u \rangle \, . 
		\label{eq:A00}
	\end{equation}  
	With (\ref{eq:nabw}), we find 
	\[ |\langle u, (\nabla w) \nabla u \rangle | \leq \| (\nabla w) u \| \,  \| \nabla u \| \lesssim \langle u, (\kappa + h_w ) u \rangle \,. \]
	Estimating $(\kappa+h_w) \leq \eps (\kappa+h_w)^2 + C_\eps$ and using $w \geq -C$ we conclude, from (\ref{eq:A00}), that 
	\begin{equation}\label{eq:first-bd} \langle u, (\kappa+h_w)^2 u \rangle \gtrsim  \langle u, \big[ (1-\Delta)^2 + (C+w)^2 \big] u \rangle \, . \end{equation} 	
	Since $h_w$ is essentially self-adjoint on $C^\infty_c (\RR^3)$ (see for instance \cite[Theorem X.28]{RS2}), this proves the first bound in (\ref{eq:lm1}). To show the second bound, consider
	\[ (1-\Delta)^3 = (1-\Delta)^2 - \nabla \cdot (1-\Delta)^2 \,  \nabla \, . \]
	From (\ref{eq:first-bd}), we have $(1-\Delta)^2 \lesssim (\kappa+h_w)^2$ and therefore that  
	\[ \begin{split} -\nabla \cdot (1-\Delta)^2 \, \nabla  &\lesssim - \nabla \cdot (\kappa+h_w)^2 \nabla = - ((h_w+\kappa)  \nabla + [\nabla, h_w ]) \cdot (\nabla (h_w + \kappa) -[\nabla , h_w]) \\
		&\lesssim (\kappa+h_w) (-\Delta) (\kappa+h_w) + | [ \nabla , h_w ]|^2 \,. \end{split} \]
	With $-\Delta \leq \kappa + h_w$, $[\nabla , h_w] = \nabla w$ and the assumption (\ref{eq:nabw}), we conclude that 
	\[ (1-\Delta)^3 \lesssim (\kappa + h_w)^3 \, . \] 
\end{proof}

We will also need the following estimate on the Fourier transform of the heat kernel associated with $h$. 
\begin{lemma}\label{lem:fractional_heat_kernel_bounds}
	Let $k_t(p,q)$ denote the Fourier transform of the kernel of the operator $\exp(-th)$ for $t>0$. The function $(p,q) \mapsto k_t(p,q)$ is nonnegative a.e. and satisfies
	\begin{equation}
		\int_{\mathbb{R}^6} k_t(p,q) \de p \de q \leq (\pi/t)^{3/2} \,. \label{Kernel_Bound_2}
	\end{equation}
\end{lemma}
\begin{proof}
	Let us define $P_t(v)$ as $(2\pi)^{-3/2}$ times the Fourier transform of the function $\exp(-t |x|^s)$. We have $\int P_t(v) \de v = 1$, and from Theorem~XIII.52 and Example~2 on p.~220 in \cite{RS4} we know that it is nonnegative. Accordingly, $P_t(v-w)$ defines a probability transition kernel\footnote{The function $P_t(v-w)$ is the transition kernel of an $s$-stable Levy process, see e.g. \cite[Example~6.5]{Levy}.}. Moreover, we have the Markov property
	\begin{equation}
		\int_{\mathbb{R}^3} P_{t_1}(v-z) P_{t_2}(z-w) \de z = P_{t_1+t_2}(v-w)
		\label{eq:markov}
	\end{equation}
	for $t_1,t_2 \geq 0$. 
	
	From the Trotter product formula, see e.g. \cite[Theorem VII.31]{RS1}, we know that
	\begin{equation}
		\exp(-t h) = \lim_{n \to \infty} \left( \exp(-t |x|^s/n) \exp(t \Delta/n) \right)^n
		\label{eq:Trotter}
	\end{equation}
	holds in the strong operator topology. Let us denote by $K_{n}(v,w)$ the Fourier transform of the integral kernel of the $n$-dependent operator on the r.h.s. of the above equation. With the notation $q_0 = v$ and $q_n = w$, we can write
	\begin{align}
		0 \leq K_n(v,w) =& \int \left( \prod_{i=1}^n P_{t/n}(q_{i-1}-q_i) \right) \exp\left(-\sum_{i=1}^{n} \frac{t q_i^2}{n} \right) \de q_1\dots \de q_{n-1} \nonumber \\
		\leq& \int \left(\prod_{i=1}^n P_{t/n}(q_{i-1}-q_i)\right)\left( \frac{1}{n}\sum_{i=1}^{n} \exp\left( -t q_i^2 \right) \right)\de q_1\dots\de q_{n-1} \nonumber \\
		=& \frac{1}{n}\sum_{i=1}^{n}\int P_{it/n}(v-q) \exp\left( -t q^2 \right) P_{(n-i)t/n}(q-w)\de q \,. \label{eq:A21}
	\end{align}
	To come from the first to the second line, we applied Jensen's inequality. Let $\phi,\psi \in L^2(\RR^3) \cap L^{\infty}(\mathbb{R}^3)$ be two nonnegative functions. In combination, \eqref{eq:Trotter} and \eqref{eq:A21} imply the bound
	\begin{align}
		0 \leq \int_{\mathbb{R}^6} \phi(v) k_t(v,w) \psi(w) \de v \de w &= \lim_{n \to \infty} \int_{\mathbb{R}^6} \phi(v) K_n(v,w) \psi(w) \de v \de w \nonumber \\
		&\leq \liminf_{n \to \infty} \frac{1}{n} \sum_{i=1}^n \int_{\mathbb{R}^9} \phi(v)  P_{it/n}(v-q) \exp\left( -t q^2 \right) P_{(n-i)t/n}(q-w) \psi(w) \de v \de q \de w \nonumber \\
		&\leq \Vert \psi \Vert_{\infty} \ \Vert \phi \Vert_{\infty} \int_{\mathbb{R}^3} \exp\left(-t q^2\right) \de q = \Vert \psi \Vert_{\infty} \ \Vert \phi \Vert_{\infty} \ \left( \pi/t \right)^{3/2}.
		\label{eq:A22}
	\end{align}
	The first bound in \eqref{eq:A22} implies that $k_t(p,q) \geq 0$ holds for a.e. $(p,q) \in \mathbb{R}^6$. The bound in \eqref{Kernel_Bound_2} follows from \eqref{eq:A22}, when we take the supremum over all functions $\phi, \psi$ with $0 \leq \phi, \psi \leq 1$ on both sides. This proves Lemma~\ref{lem:fractional_heat_kernel_bounds}. 
\end{proof}

We are now ready to show Proposition~\ref{prop:gaphi}. 
\begin{proof}[Proof of Proposition~\ref{prop:gaphi}] 
	To show (\ref{eq:prop1}), we remark that, with the choice $w (x) = |x|^s$, for $s \in (0;2]$, we have $|\nabla w (x)| = s |x|^{s-1}$ and therefore 
	\[ |\nabla w (x)|^2 \lesssim \big[ w (x) + |x|^{-2} + 1 \big] \lesssim ( 1 + h) \]
	with $h = -\Delta + w (x)$. In the last step, we applied Hardy's inequality. Thus, we can apply Lemma \ref{lem:secreg1} to estimate 
	\[  \tr \, (1-\Delta)^{3/2} \gamma (1-\Delta)^{3/2} \leq  \Big\| (1-\Delta)^{3/2} \frac{1}{(1+h)^{3/2}} \Big\|^2 \,  \tr \, (1+h)^{3} \gamma \, . \]
	The claim follows by noticing (with $\{ e_j \}_{j =0}^\infty$ indicating the eigenvalues of $h$) that 
	\begin{equation}\label{eq:FK} \tr \, (1+h)^3 \gamma = \sum_{j=1}^\infty  \frac{1}{e^{\, \beta (e_j - \mu)} - 1}  (1+e_j)^3 \lesssim \tr e^{-ch} \lesssim \int_{\bR^3} e^{-cw(x)/2} dx < \infty \end{equation} 
	for some constant $c>0$, where we used $e_1 > e_0 > \mu $ in the first step, see \cite[Theorem~XIII.47]{RS4}, and, in the second step, the bound for the trace of the propagator of the heat equation below Eq.~(6.5) in \cite{BraLie1976}. 
	
	The bound $\| \gamma \| \lesssim T_c (s)$ follows from $(h-\mu)Q \geq e_1 - e_0 =: \Delta e > 0$, which implies that 
	\begin{equation*}
		\Vert \gamma \Vert \leq \frac{1}{\exp(\beta \Delta e)-1} \leq \frac{1}{\beta \Delta e} \lesssim T_{\mathrm{c}}(s) \,.
	\end{equation*}
	
	Next, we show (\ref{eq:prop2}). We use here the notation $ \psi_j $ for the eigenfunction of $h$ corresponding to the eigenvalue $ e_j $, for $j\in\NN$. The identity $(\exp(x)-1)^{-1} = \sum_{\alpha=1}^{\infty} \exp(-\alpha x)$ for $x > 0$ allows us to write 
	\begin{equation*}
		\gamma = \sum_{\alpha = 1}^{\infty} \left( \exp(-\beta  (h-\mu) \alpha) - \exp(-\beta (e_0 - \mu)\a) | \psi_0 \rangle \langle \psi_0 | \right) \,.
	\end{equation*}
	Using the above representation of $\gamma$ and the notation $k_t$ from Lemma \ref{lem:fractional_heat_kernel_bounds}, we estimate
	\begin{align}
		| \hat{\gamma}(p,q) | \leq& \sum_{\alpha=1}^M \exp(\beta \mu \alpha) k_{\beta \alpha}(p,q) + \sum_{\alpha=M+1}^{\infty} \sum_{i=1}^{\infty} \exp(-\beta \alpha (e_i - \mu)) | \hat{\psi}_i(p) | \ | \hat{\psi}_i(q) | \nonumber \\
		&+ \sum_{\alpha=1}^M \exp(-\beta(e_0-\mu)\alpha) |\hat{\psi}_0(p)|\ | \hat{\psi}_0(q)| \,, \label{eq:A23}
	\end{align}
	where $M \in \mathbb{N}$. An application of Lemma~\ref{lem:fractional_heat_kernel_bounds} shows that the $L^1(\mathbb{R}^6)$-norm of the first term on the r.h.s. is bounded by
	\begin{equation}
		\sum_{\alpha=1}^{M} \exp(\beta \mu \alpha)  \int_{\mathbb{R}^6} k_{\beta \alpha}(p,q) \de p \de q \leq \sum_{\alpha=1}^{M} \exp(\beta e_0 \alpha ) \left( \frac{\pi}{\beta \alpha} \right)^{3/2} \leq \left( \pi/\beta \right)^{3/2} \exp(\beta e_0 M ) \sum_{\alpha=1}^{\infty} \alpha^{-3/2}.
		\label{eq:A24}
	\end{equation}
	To obtain the first bound we used $e_0 - \mu > 0$, and for the second one $e_0 > 0$. The $L^1(\mathbb{R}^6)$-norm of the term in the last line is bounded by
	\begin{equation}
		\sum_{\alpha=1}^M \exp(-\beta(e_0-\mu)\alpha) \int_{\mathbb{R}^6} |\hat{\psi}_0(p)|\ | \hat{\psi}_0(q)| \leq M \Vert \hat{\psi}_0 \Vert_1^2 \,.
		\label{eq:A24b}
	\end{equation}
	The norm on the r.h.s. is bounded by the $H^2$-norm of $\psi_0$. Since $\psi_0$ is an eigenfunction of $h$, we conclude with Lemma~\ref{lem:secreg1} that it is finite. It remains to consider the second term on the r.h.s. of \eqref{eq:A23}. To that end, with $\Delta e = e_1 - e_0 > 0$ we estimate
	\begin{align} 
		\sum_{\alpha=M+1}^{\infty} \sum_{i=1}^{\infty} & \exp(-\beta \alpha (e_i - \mu)) | \hat{\psi}_i(p) | \ | \hat{\psi}_i(q) | \nonumber \\ 
		&= \sum_{i=1}^\infty \frac{e^{-\beta (M+1) (e_i - \mu)}}{1- e^{-\beta (e_i - \mu)}}  | \hat{\psi}_i(p) | \ | \hat{\psi}_i(q) | \leq \frac{1}{\beta \Delta e} 
		\sum_{i=1}^\infty e^{-\beta M (e_i - \mu)} | \hat{\psi}_i(p) | \ | \hat{\psi}_i(q) | \,.
		\label{eq:A25}
	\end{align} 
	Integrating over $p,q$, the r.h.s. is bounded by 
	\begin{align*} \frac{1}{\beta \Delta e} 
		\sum_{i=1}^\infty e^{-\beta M (e_i - \mu)} \left( \int | \hat{\psi}_i(p) | \de p \right)^2 \lesssim \frac{1}{\beta \Delta e}  \sum_{i=1}^{\infty} e^{-\beta M (e_i-\mu)} \, \langle \psi_i, (1-\Delta)^2 \psi_i \rangle \lesssim \frac{1}{\beta \Delta e} \sum_{i=1}^{\infty} e^{-\beta M (e_i - \mu)} (1+e_i)^2 \, .
	\end{align*}
	Choosing $M \in \bN$ with $M > \beta^{-1}$, (\ref{eq:prop2}) follows from the last equation (proceeding as in (\ref{eq:FK}) to bound the sum over $i$), together with (\ref{eq:A23}), (\ref{eq:A24}), and (\ref{eq:A24b}).
	
	Finally, let us show the bound on $\| \phi \|_{H^3}$, for  a normalized $\phi$ solving the Hartree equation (\ref{eq:hartree}). Let $w (x) = |x|^s + g ( v * |\phi|^2) (x)$ and $h_w = -\Delta + w$. Since $v \in L^{\infty}(\mathbb{R}^3)$ (by Sobolev embedding), we have $\| v * |\phi|^2 \|_\infty \leq \Vert v \Vert_{\infty}$. We also have (choosing $p' < 3/2$ such that $1/p + 1/p' = 1$) 
	\[ \| \nabla (v * |\phi|^2) \|_\infty  \leq \|\nabla v \|_p \| \phi \|_{2p'}^2 \lesssim  \|\nabla v \|_p \| \phi \|^2_{H_1} \,. \] 
	Therefore, we can apply Lemma \ref{lem:secreg1}, with $w = |x|^s + g (v *|\phi|^2) (x)$ and $\kappa = 1 + \| v \|_\infty $
	to show that 
	\[ \| (1- \Delta)^{3/2} \phi \| \lesssim \| (\kappa+h_w)^{3/2} \phi \| = (\kappa+\mu^{\mathrm{H}})^{3/2} \,,  \]
	which proves the claim.
\end{proof}

\section{The Hartree--Fock--Bogoliubov equations}
\label{sec:properties_HFB}

As explained in Section~\ref{sec:proofstrategy}, we use the HFB equations to define our fluctuation dynamics in Section~\ref{sec:constrcutionfluctdyn}. Here we introduce the equations and collect some of their properties that we need for the many-body analysis. 

For the triple $(\phi_t,\gamma_t,\alpha_t)$ consisting of a condensate wave function $\phi_t \in L^2(\mathbb{R}^3)$, a positive trace class operator $\gamma_t$ (a 1-pdm), and a pairing function $\alpha_t \in L^2(\mathbb{R}^6)$, the HFB equations take the form 
\begin{align}
	\mathrm{i}\dpr_t\phi_t =& h(\gamma_t)\phi_t+k(\a_t^{\phi_t})\cc{\phi_t} \nonumber \\
	\mathrm{i}\dpr_t\gamma_t =&\; [h(\gamma_t^{\phi_t}), \gamma_t] + k(\a_t^{\phi_t})\a_t^*-\a_tk(\a_t^{\phi_t})^* \nonumber \\
	\mathrm{i}\dpr_t\a_t =&\; [h(\gamma_t^{\phi_t}), \a_t]_+ + [k(\a_t^{\phi_t}), \gamma_t]_++k(\a_t^{\phi_t}), \label{eq:HFB}
\end{align}
where $[A,B]_+=A\cc{B}^*+B\cc A^*$, $\gamma^{\phi}=\gamma+|\phi\rangle\langle\phi|$, and $\a^{\phi}=\a+|\phi\rangle\langle\cc\phi|$. By $\overline{A}$ we denote the operator with the complex conjugate integral kernel in position space. Moreover, we use the notations
\begin{equation} \label{eq:h(g)_k(g)_def}
	h(\gamma)=-\D+b(\gamma),\qquad k(\a) = \frac 1 N v\sharp\a\,,
\end{equation} where \[b(\gamma)=\frac 1 N v\ast \rho_\g+\frac 1 N v\sharp\gamma\,,\] $v\sharp \s$ denotes the operator with kernel $v(x-y)\s(x,y)$, and the density associated with the 1-pdm $\gamma$ is given by $\rho_\g(x)=\gamma(x,x)$. Here $v$ is an interaction potential satisfying suitable assumptions to be specified later. The mean-field scaling we are interested in is reflected in the factor $N^{-1}$ multiplying $v$ in the equations.

The HFB equations arise naturally as a quasi-free approximation of the many-body dynamics (which does not preserve quasi-freeness), for quasi-free initial data. For more details concerning this argument and more information about the HFB equations we refer the reader to \cite{BachBretaux, BenedikterSolovej_Bog_de_Gennes}.


\subsection{Well-posedness}
\label{sec:wellposedness}
In this section we state conditions that guarantee the well-posedness of the HFB equations. We define $M=(1-\D)^{1/2}$, and for $j\ge 0$ we introduce the linear spaces
\begin{equation*}
	\cH^{j,1}=M^{-j}\cL^1M^{-j} \quad \text{ and } \quad \cH^{j,2}=\left\{\a\in\cL^2\,| \, M^j\a,\a M^j\in\cL^2 \right\}
\end{equation*}
with $\mathcal{L}^j$ defined in Section~\ref{sec:notation}, which are Banach spaces when endowed with their natural norms 
\begin{equation*}
	\| \gamma \|_{\cH^{j,1}}=\|M^j\gamma M^j\|_{\cL^1} \quad \text{ and } \quad \|\a\|_{\cH^{j,2}}=\|M^j\a\|_{\cL^2}+\|\a M^j \|_{\cL^2} \,,
\end{equation*}
respectively. 
We also define the Banach space $(X^j, \Vert \cdot \Vert_{X^j})$, where 
\begin{equation*}
	X^j=H^j\times\cH^{j,1}\times\cH^{j,2} \, \quad \text{ and } \quad \|(\phi,\gamma,\a)\|_{X^j}=\|\phi\|_{H^{j}} + \|\gamma\|_{\cH^{j,1}} + \|\a\|_{\cH^{j,2}} \,.
\end{equation*} The following holds true, see \cite[Theorem 5.1]{BachBretaux}.
\begin{theorem}\label{thm:well-posedness-HFB}
	Let $(\phi_0,\gamma_0,\a_0)\in X^3$ and suppose that
	\begin{equation*}
		\G_0^{(1)} =\begin{pmatrix}
			\gamma_0& \a_0 \\
			\cc\a_0 & 1+\cc\gamma_0
		\end{pmatrix} 
	\end{equation*}
	is positive, i.e. $\G_0^{(1)} \geq 0$ as an operator on $L^2(\RR^3)\otimes L^2(\RR^3)$ (this is always the case if $\g_0$, $\a_0$ are the reduced 1-pdm and pairing density of a state $\G$ on $\fock(L^2(\RR^3))$). Assume also that the interaction potential $v$ satisfies $v(x) = v(-x)$ and $v\in W^{1,p}(\RR^3)$ for some $p>3$. Then there exists a unique triple $(\phi_t,\gamma_t,\a_t)\in C^1([0,\infty),X^1)\cap C^0([0,\infty),X^3)$ satisfying \eqref{eq:HFB} with initial datum $(\phi_0,\gamma_0,\a_0)$ in the classical sense in $X^1$. Moreover, the number of particles
	\begin{equation}\label{eq:HFB_number}
		n(\phi_t,\gamma_t) = \int_{\RRR^3}|\phi_t(x)|^2\de x + \tr\gamma_t
	\end{equation} and the energy 
	\begin{align}\nonumber
			\cE(\phi_t,\gamma_t,\alpha_t) = & \; \tr[h(\gamma_t^{\phi_t})\g_t]+\frac 1 N \tr[(v\ast|\phi_t|^2+v\sharp|\phi_t\rangle\langle\phi_t|)\gamma_t]\\
			& \; + \frac1{2N}\tr[(v\ast \rho_{\g_t}+v\sharp\gamma_t)\gamma_t]+\frac 1 N \int_{\RRR^6}v(x-y)|\a_t(x,y)+\phi_t(x)\phi_t(y)|^2\de x \de y \label{eq:HFB_energy}
	\end{align} are conserved along the evolution, i.e. $n(\phi_t,\gamma_t)=n(\phi_0,\gamma_0)$, $\cE(\phi_t,\gamma_t,\alpha_t) =\cE(\phi_0,\gamma_0,\alpha_0)$, for every $t>0$. Finally, the dynamics is positivity-preserving, that is, 
	\begin{equation} \label{eq:positivity_preserving}
		\G_t^{(1)} =\begin{pmatrix}
			\gamma_t & \a_t \\
			\cc\a_t & 1+\cc\gamma_t
		\end{pmatrix}\ge 0
	\end{equation}
	for every $t>0$.
\end{theorem}

\subsection{Approximation of the HFB dynamics} \label{sec:apprximationHFB}
Although we use the HFB equations to define our fluctuation dynamics in Section~\ref{sec:constrcutionfluctdyn}, we ultimately want to show convergence of the full many-body dynamics to a simpler effective dynamics, with $\phi_t$ solving the time-dependent Hartree equation \eqref{eq:timedepHartree}, $\gamma_t = e^{\mathrm{i}\Delta t} \gamma_0 e^{-\mathrm{i}\Delta t}$ evolving freely, and $\alpha_t = 0$. To reach this goal, we will use the next proposition. 
\begin{proposition}\label{lem:closeness_dynamics}
	Let the pair $(\phi_0,\gamma_0)$ satisfy Assumption~\ref{ass:scales} with $s \in (0,\infty)$, and let the interaction potential $v$ satisfy the assumptions of Theorem~\ref{thm:main_2}. Moreover, let $(\phi_t,\gamma_t,\a_t)$ denote the solution to the HFB equations \eqref{eq:HFB} with initial datum $(\phi_0,\g_0,0)$. Then there exists a constant $c>0$ independent of $N$ and $t$ such that 
	\begin{align}
		\left\Vert \gamma_t - \gamma_t^{\mathrm F} \right\Vert_{\cL^1} \lesssim N^{1/2} T_{\mathrm{c}}^{3/4}(s) t \exp(ct) \,,  &&	\left\Vert \phi_{t} - \phi_t^{\mathrm H}  \right\Vert \lesssim T_{\mathrm{c}}^{3/4}(s) t \exp(ct) \,, \label{eq:closeness_dynamics}
	\end{align}
	where $\phi_t^{\mathrm H} $ is the solution to the time-dependent Hartree equation in \eqref{eq:timedepHartree} with initial condition $\phi_0$ and $\gamma_{t}^{\mathrm F} =e^{\mathrm{i} \D t} \gamma_0 e^{-\mathrm{i} \D t}$. 
\end{proposition}
\begin{remark}\label{rmk:comparison_dynamics}
	\begin{enumerate}
		\item Our assumptions on $(\phi_0,\gamma_0)$ and $v$ guarantee that the HFB equations are well-posed, see Theorem~\ref{thm:well-posedness-HFB}. It is well-known that they also guarantee the global well-posedness of the time-dependent Hartree equation in $H^1(\RR^3)$, and that the $L^2$-mass of the solution is preserved along the Hartree flow.
		\item As explained in Section \ref{sec:Equilibriumstates}, the scaling of the critical temperature with the number of particles is given by $T_c(s)\sim N^{1/\a}=N^{\frac{2s}{6+3s}}$. In particular $T_{\mathrm c}^{3/4} (s)\ll N^{1/2}$ for every $s>0$, showing that the bounds given in \eqref{eq:closeness_dynamics} are non-trivial. 
		\item It is possible to get better rates than those in \eqref{eq:closeness_dynamics} if one compares to a slightly more complicated evolution for the 1-pdm. Indeed, let $\o_0=|\phi_0\rangle\langle\phi_0|+\g_0$, and define $\o_t$ to be the solution of the Hartree equation \eqref{eq:Hartree_1pdm} with initial datum $\o_{t=0}=\o_0$. Then we have the improved bound
		\begin{equation} 
			\label{eq:Hartree_1pdm_bound} \left\| \ |\phi_t\rangle\langle\phi_t|+\g_t - \o_t \right\|_{\cL^1} \lesssim N^{1/2} T_c^{1/2}(s) t\exp(ct)\,, 
		\end{equation} for some constant $c>0$ independent of $t$ and $N$. As a consequence, one gets a better (indeed, optimal) rate of convergence in \eqref{eq:thmbound1pdm} by replacing $|\phi_t^{\mathrm H} \rangle\langle\phi_t^{\mathrm H} |+e^{i\D t}\g_0 e^{-i\D t}$ with $\o_t$ in the statement, see Remark \ref{rmk:restriction_s}.3. 
		The estimate \eqref{eq:Hartree_1pdm_bound} is proven by a suitable adaptation of the Gr\"onwall-type argument given below, but since we do not need it for our main result, we omit the details. We only highlight that the gain from using \eqref{eq:Hartree_1pdm} is that it allows one to replace the direct term $[v\ast|\phi_t|^2,\g_t]$ appearing on the r.h.s. of \eqref{eq:Duhamel_omega}, which we can only bound by $N^{1/2}T_{\mathrm c}^{3/4}(s)$ (see \eqref{eq:Duh_gamma_bad_bound}), with some simpler terms that can be shown to be of order $N^{1/2}T_{\mathrm c}^{1/2}(s)$.
	\end{enumerate}
\end{remark}
\begin{proof}
	We introduce the notation $\theta_t=(\phi_t,\gamma_t,\a_t)$, and we write the HFB equations as 
	
	\begin{equation}\label{eq:HFB_in_X_space}
		\mathrm{i}\dt\theta_t=A\theta_t+f(\theta_t)\,,
	\end{equation}
	where 
	\begin{equation}
		A\theta=\left(-\D\phi,~[-\D,\gamma],~[-\D,\a]_+\right)
		\label{eq:linpartHFB}
	\end{equation}
	is a linear operator on $X^0$ with domain $\cD(A)=X^2$, and $f\in C(X^0)$ is defined by
	\[
	\begin{split}
		f(\theta)=\Big(b(\gamma)\phi+k(\a^{\phi})\cc\phi,~[b(\gamma^\phi),\gamma] + (k(\a^\phi)\a^*-\a k(\a^\phi)^*)\gamma,~[b(\gamma^\phi),\a]_++[k(\a^\phi),\a]_+ + k(\a^\phi)\Big)\,.
	\end{split}
	\] 
	Combining \eqref{eq:HFB_in_X_space} and \eqref{eq:timedepHartree} we get
	\begin{align}
		\phi_t-\phi_t^{\mathrm H}=&-\mathrm{i}\int_0^te^{\mathrm{i}(t-s)\D}\left(\left( \frac1Nv\ast|\phi_s|^2\phi_s - \frac1Nv\ast|\phi_s^{\mathrm H}|^2\phi_s^{\mathrm H}  \right)+ b(\gamma_s)\phi_s+k(\a_s)\cc{\phi_s}\right)\de s \, , \label{eq:Duhamel_phi}\\
		\gamma_t-\gamma_t^{\mathrm F}=&-\mathrm{i}\int_0^te^{\mathrm{i}(t-s)\D}\Big([b(\gamma_s^{\phi_s}),\gamma_s]+k(\a_s^{\phi_s})\a_s^*-\a_s k(\a_s^{\phi_s})^*   \Big)e^{-\mathrm{i}(t-s)\D}\de s \, ,\label{eq:Duhamel_omega}\\
		\a_t=&-\mathrm{i}\int_0^te^{\mathrm{i}(t-s)(\D_x+\D_y)}\Big([b(\gamma_s^{\phi_s}),\a_s]_++[k(\a_s^{\phi_s}),\gamma_s]_++k(\a_s^{\phi_s})\Big)\de s\,.\label{eq:Duhamel_alpha}
	\end{align}
	Here, we identify $\alpha_t$ with its kernel $\alpha_t (x,y)$ in $L^2 (\RR^6)$ and we denote by 
	$\Delta_x + \Delta_y$ the Laplacian acting on this space. 
	
	Let us start by deriving a bound for the r.h.s. of \eqref{eq:Duhamel_phi}. We use the identity 
	\begin{equation}\label{eq:identity_hartree}
		\frac1N v\ast|\phi_s|^2\phi_s - \frac1N v\ast|\phi_s^{\mathrm H}|^2\phi_s^{\mathrm H} = \left(\frac1N v\ast(|\phi_s|^2-|\phi_s^{\mathrm H}|^2)\right)\phi_s^{\mathrm H} + \left(\frac1N v\ast|\phi_s|^2\right)(\phi_s-\phi_s^{\mathrm H})\,,
	\end{equation}  Young's inequality, and the conservation of the particle number for the Hartree and HFB equations, to see that 
	\[
	\begin{split}
		\left\|  \frac1N v\ast|\phi_s|^2\phi_s - \frac1N v\ast|\phi_s^{\mathrm H}|^2\phi_s^{\mathrm H}  \right\|\lesssim &N^{-1}\|v\|_\infty\left( \|\phi_s\|^2+\| \phi_s^{\mathrm H} \|^2 \right)\| \phi_s-\phi_s^{\mathrm H} \|\\
		\lesssim & N^{-1}\|v\|_\infty\left( \|\phi_0\|^2+\| \g_0\|_{\cL^1} \right) \|\phi_s-\phi_s^{\mathrm H} \|\lesssim\|\phi_s-\phi_s^{\mathrm H} \| \,.
	\end{split}
	\] In the second term on the r.h.s. of \eqref{eq:Duhamel_phi} we write 
	\[
	b(\gamma_s)\phi_s=b(\gamma_s^{\mathrm F})\phi_s+b(\gamma_s-\gamma_s^{\mathrm F})\phi_s\,,
	\] and we estimate the two terms separately. First, we bound the operator norm of $b(\gamma_s^{\mathrm F})=N^{-1}v\ast \rho_{\gamma_s^{\mathrm F}}+N^{-1}v\sharp\gamma_s^{\mathrm F}$. Using Assumption~\ref{ass:scales}~(\textbf{B}), we see that 
	\begin{equation}\label{eq:bound_d_omega_operator}
		\|N^{-1}v\ast \rho_{\gamma_s^{\mathrm F}}\|= N^{-1}\|v\ast \rho_{\gamma_s^{\mathrm F}}\|_\infty\le  N^{-1}\|v\|_1 \sup_{x\in\RR^3}|\gamma_s^{\mathrm F}(x,x)|\lesssim N^{-1} T_{\mathrm{c}}^{3/2}(s)\,.
	\end{equation} Moreover, for any $\psi\in L^2(\RR^3)$ we can estimate
	\[
	\begin{split}
		N^{-1}|(v\sharp \gamma_s^{\mathrm F})\psi(x)| \le&N^{-1}\int_{\mathbb{R}^3} |v(x-y)|\ |\gamma_s^{\mathrm F}(x,y)|\ |\psi(y)|\de y\\
		\le & N^{-1}\sup_{x,y \in \mathbb{R}^3}|\gamma_s^{\mathrm F}(x,y)|\int_{\mathbb{R}^3} |v(x-y)|\ |\psi(y)|\de y \lesssim N^{-1} T_{\mathrm{c}}^{3/2}(s) (|v|\ast|\psi|)(x)\,,
	\end{split}
	\] which implies 
	\begin{equation}\label{eq:bound_sharp_omega_operator}
		N^{-1}\|(v\sharp \gamma_s^{\mathrm F})\psi\|\lesssim N^{-1} T_{\mathrm{c}}^{3/2}(s) \|\ |v|\ast|\psi|\ \|\lesssim N^{-1} T_{\mathrm{c}}^{3/2}(s) \|\psi\|\,.
	\end{equation} For the operator norm of $b(\gamma_s-\gamma_s^{\mathrm F})$ we have the bound
	\begin{align}
		N^{-1}\|v\ast (\r_{\gamma_s}-\r_{\gamma_s^{\mathrm F}})\|= &N^{-1}\|v\ast (\r_{\gamma_s}-\r_{\gamma_s^{\mathrm F}})\|_\infty \nonumber \\
		\le & N^{-1}\|v\|_\infty\int_{\mathbb{R}^3}|\gamma_s(x,x)-\gamma_s^{\mathrm F}(x,x)|\de x\lesssim  N^{-1}\|\gamma_s-\gamma_s^{\mathrm F}\|_{\cL^1}\,,
		\label{eq:bound_d_omega_diff_operator}
	\end{align} 
	and
	\begin{align}
		N^{-1}\|v\sharp(\gamma_s-\gamma_s^{\mathrm F})\|\le & N^{-1}\left(\int_{\mathbb{R}^6} v(x-y)^2|\gamma_s(x,y)-\gamma_s^{\mathrm F}(x,y)|^2\de x \de y\right)^{1/2} \nonumber \\
		\le & N^{-1}\|v\|_\infty\|\gamma_s-\gamma_s^{\mathrm F}\|_{\cL^2}\lesssim N^{-1}\|\gamma_s-\gamma_s^{\mathrm F}\|_{\cL^1}\,, \label{eq:bound_sharp_omega_diff_operator}
	\end{align} 
	where in the last step we used the fact that the Hilbert-Schmidt norm of a trace-class operator is bounded by its trace norm. 
	Collecting \eqref{eq:bound_d_omega_operator}--\eqref{eq:bound_sharp_omega_diff_operator} we deduce
	\begin{align}
		\|b(\gamma_s^{\mathrm F})\|\lesssim & N^{-1} T_{\mathrm{c}}^{3/2}(s) \nonumber \\
		\|b(\gamma_s-\gamma_s^{\mathrm F})\| \lesssim & N^{-1}\|\gamma_s-\gamma_s^{\mathrm F}\|_{\cL^1}\,. \label{eq:b_l2_bounds}
	\end{align} The above bounds and the conservation of the particle number imply
	\[
	\begin{split}
		\|b(\gamma_s)\phi_s\|\lesssim & N^{-1/2} \left( \|\phi_0\|^2+\|\g_0\|_{\cL^1}\right)^{1/2}\left(N^{-1/2} T_{\mathrm{c}}^{3/2}(s)\|v\|_1 + N^{-1/2}\|v\|_\infty\|\gamma_s-\gamma_s^{\mathrm F}\|_{\cL^1}\right)\\
		\lesssim & N^{-1/2} T_{\mathrm{c}}^{3/2}(s) \|v\|_1 + N^{-1/2}\|v\|_\infty\|\gamma_s-\gamma_s^{\mathrm F}\|_{\cL^1}  \,.
	\end{split}
	\] Similarly, we estimate
	\[
	\|k(\a_s)\phi_s\|\le N^{-1}\|v\|_\infty\|\a_s\|_{\cL^2}\|\phi_s\|\lesssim N^{-1/2}\|v\|_\infty\|\a_s\|_{\cL^2}\,.
	\] We have thus shown that 
	\begin{equation}\label{eq:Duhamel_phi_final_bound}
		\|\phi_t-\phi_t^{\mathrm H}\|\lesssim N^{-1/2} T_{\mathrm{c}}^{3/2}(s) t + \int_0^t\left(\|\phi_s-\phi_s^{\mathrm F}\|+N^{-1/2}\|\gamma_s-\gamma_s^{\mathrm F}\|_{\cL^1}+N^{-1/2}\|\a_s\|_{\cL^2}\right)\de s\,.
	\end{equation}
	
	We now turn to the r.h.s. of \eqref{eq:Duhamel_omega}.  Using again the bounds in \eqref{eq:b_l2_bounds}, we find
	\begin{align}
		\| [b(\gamma_s^{\phi_s}),\gamma_s] \|_{\cL^1}\le & \| b(\gamma_s^{\phi_s})\gamma_s \|_{\cL^1} + \|\gamma_s b(\gamma_s^{\phi_s}) \|_{\cL^1} \nonumber \\
		\le & 2\Big(\|\gamma_s\|_{\cL^1}\left(\| b(\gamma_s-\gamma_s^{\mathrm F})\| +\| b(\gamma_s^{\mathrm F})\|\right) + \| b(|\phi_s\rangle\langle\phi_s|)\gamma_s^{\mathrm F}\|_{\cL^1}+\|b(|\phi_s\rangle\langle\phi_s|)\| \ \|\gamma_s-\gamma_s^{\mathrm F}\|_{\cL^1}\Big) \nonumber \\
		\lesssim & \|v\|_\infty\|\gamma_s-\gamma_s^{\mathrm F}\|_{\cL^1} +T_{\mathrm{c}}^{3/2}(s)\|v\|_1 \nonumber \\
		& + N^{-1}\| (v\ast |\phi_s|^2)\gamma_s^{\mathrm F}\|_{\cL^1}  + N^{-1}\| (v\sharp |\phi_s\rangle\langle\phi_s|)\gamma_s^{\mathrm F}\|_{\cL^1} + \|b(|\phi_s\rangle\langle\phi_s|)\| \ \|\gamma_s-\gamma_s^{\mathrm F}\|_{\cL^1}\,. \label{eq:bound_duhamel_omega_1}
	\end{align} The first term in the last line is bounded by
	\begin{align}\nonumber
			N^{-1}\| (v\ast |\phi_s|^2)\gamma_s^{\mathrm F}\|_{\cL^1}\le & N^{-1}\|\gamma_s^{\mathrm F}\|_{\cL^1}^{1/2}\|(v\ast |\phi_s|^2)(\gamma_s^{\mathrm F})^{1/2}\|_{\cL^2}\\ \nonumber
			\lesssim & N^{-1/2} \left(\int_{\mathbb{R}^3} |(v\ast |\phi_s|^2)(x)|^2|\gamma_s^{\mathrm F}(x,x)|\de x\right)^{1/2}\\
			\lesssim & N^{-1/2} T_{\mathrm{c}}^{3/4}(s)\|v\ast |\phi_s|^2\|\lesssim N^{1/2} T_{\mathrm{c}}^{3/4}(s)\|v\|\lesssim N^{1/2} T_{\mathrm{c}}^{3/4}(s)\,. \label{eq:Duh_gamma_bad_bound}
	\end{align} The exchange term can be estimated by
	\[ \| v \sharp |\phi_s \rangle \langle \phi_s| \|_{\cL_1} \leq \int \hat{v} (p) \| \ |e^{ip \cdot . } \phi_s \rangle \langle e^{ip\cdot .} \phi_s |\ \|_{\cL_1}\de p \leq \| \hat{v} \|_1 \| \phi_s \|^2, \]
	thus, 
	\[ N^{-1} \|  v\sharp |\phi_s \rangle \langle \phi_s| \gamma_s^{\mathrm F} \|_{\cL_1} \lesssim \| \gamma_s^{\mathrm F} \| \lesssim T_c(s) \,.\] As for the last term on the last line of \eqref{eq:bound_duhamel_omega_1}, we have 
	\[
	\begin{split}
		\|b(|\phi_s\rangle\langle\phi_s|)\| \ \|\gamma_s-\gamma_s^{\mathrm F}\|_{\cL^1} \le N^{-1}\left(\| v\ast |\phi_s|^2\|_{\infty}+\| v\sharp |\phi_s\rangle\langle\phi_s| \ \|\right)\|\gamma_s-\gamma_s^{\mathrm F}\|_{\cL^1}\lesssim \|v\|_\infty\|\gamma_s-\gamma_s^{\mathrm F}\|_{\cL^1}\,.
	\end{split}
	\] Going back to \eqref{eq:bound_duhamel_omega_1} we find 
	\[
	\| [b(\gamma_s^{\phi_s}),\gamma_s] \|_{\cL^1}\lesssim \|\gamma_s-\gamma_s^{\mathrm F}\|_{\cL^1} +N^{1/2}T_{\mathrm{c}}^{3/4}(s)+T_{\mathrm{c}}^{3/2}(s)+T_{\mathrm{c}}(s) \,.
	\] To bound the last two terms in the integral on the r.h.s. of \eqref{eq:Duhamel_omega}, we use the elementary inequality $\|\a_s\|_{\cL^2}^2\le (1+\tr\gamma_s)\tr\gamma_s$, which follows from \eqref{eq:positivity_preserving}. We find 
	\begin{equation}\label{eq:bound_alpha_duhamel_omega}
		\|k(\a_s^{\phi_s})\a_s^*\|_{\cL^1}\le \|k(\a_s^{\phi_s})\|_{\cL^2}\|\a_s\|_{\cL^2}\le N^{-1}\|v\|_\infty\left(\|\phi_s\|^2+(1+\|\gamma_s\|_{\cL^1}^2)^{1/2}\right)\|\a_s\|_{\cL^2} \lesssim \|\a_s\|_{\cL^2}\,,
	\end{equation}and similarly 
	\[
	\|\a_s k(\a_s^{\phi_s})^*\|_{\cL^1}\lesssim \|\a_s\|_{\cL^2}\,.
	\] We conclude that 
	\begin{equation}\label{eq:Duhamel_omega_final_bound}
		\|\gamma_t-\gamma_t^{\mathrm F}\|_{\cL^1}\lesssim \left(N^{1/2} T_{\mathrm{c}}^{3/4}(s)+T_{\mathrm{c}}^{3/2}(s)+T_{\mathrm{c}}(s)\right) t + \int_0^t\left(\|\gamma_s-\gamma_s^{\mathrm F}\|_{\cL^1} + \|\a\|_{\cL^2}\right) \de s \,.
	\end{equation}
	
	We also need to control the operator norm of \eqref{eq:Duhamel_omega}. Thus, we observe that
	\begin{align}
		\|b(\gamma_s^{\phi_s})\|\le & N^{-1}\|v\ast (\r_{\gamma_s}+|\phi_s|^2)\|_{\infty}+ N^{-1}\|v\sharp(\gamma_s + |\phi_s\rangle\langle\phi_s|)\|_{\cL^2} \nonumber \\
		\lesssim & N^{-1}\|v\|_\infty\left(\|\gamma_s\|_{\cL^1}+\|\phi_s\|^2\right)\lesssim \|v\|_\infty\,, \label{eq:comparison_intermediate_bound_b_operator}
	\end{align}
	which implies
	\[
	\| [b(\gamma_s^{\phi_s}),\gamma_s] \|\le 2\|b(\gamma_s^{\phi_s})\| \ \|\gamma_s \|\lesssim \left(\|\gamma_s-\gamma_s^{\mathrm F}\|+T_{\mathrm{c}}(s) \right)\,.
	\] Similarly, we find
	\[
	\|k(\a_s^{\phi_s})\a_s^*\|\lesssim \|k(\a_s^{\phi_s})\a_s^*\|_{\cL^1}\lesssim \|\a_s\|_{\cL^2}\,, 
	\] where we used \eqref{eq:bound_alpha_duhamel_omega} in the second step. We conclude that
	\begin{equation}\label{eq:Duhamel_omega_final_operator}
		\|\gamma_t-\gamma_t^{\mathrm F}\|\lesssim T_{\mathrm{c}}(s) t + \int_0^t\left(\|\gamma_s-\gamma_s^{\mathrm F}\|+\|\a_s\|_{\cL^2}\right)\de s\,.
	\end{equation}
	
	Finally, we combine \eqref{eq:Duhamel_alpha}, \eqref{eq:comparison_intermediate_bound_b_operator}, $\|k(\a_s^{\phi_s})\|\le\|k(\a_s^{\phi_s})\|_{\cL^2}\lesssim \|v\|_\infty$, and $\|AB\|_{\cL^2}\le\|A\| \ \|B\|_{\cL^2}$ to see that
	\begin{align}
		\|\a_t\|_{\cL^2}\le & \int_0^t\left\| [b(\gamma_s^{\phi_s}),\a_s]_++[k(\a_s^{\phi_s}),\gamma_s]_+ +k(\a_s^{\phi_s})\right\|_{\cL^2}\de s \nonumber \\
		\lesssim &  T_{\mathrm{c}}(s) t + \int_0^t\left(\|\gamma_s-\gamma_s^{\mathrm F}\|+\|\a_s\|_{\cL^2}\right)\de s \label{eq:Duhamel_alpha_final_bound}
	\end{align} 
	holds.
	
	Let us define the norm
	\[
	\| \theta \|_*=T_{\mathrm{c}}^{-3/4}(s)\|\phi\|+N^{-1/2} T_{\mathrm{c}}^{-3/4}(s)\|\gamma\|_{\cL^1}+T_{\mathrm{c}}^{-1}(s)(\|\gamma\|+\|\a\|_{\cL^2})\,.
	\] 
	In combination, \eqref{eq:Duhamel_phi_final_bound},\eqref{eq:Duhamel_omega_final_bound}, \eqref{eq:Duhamel_omega_final_operator}, and \eqref{eq:Duhamel_alpha_final_bound} imply
	\[
	\|\theta_t-\theta_{t,0}\|_*\le Ct+K\int_0^t\|\theta_s-\theta_{s,0}\|_*\,,
	\] for constants $C,K>0$ independent of $t$ and $N$. Here, $\theta_t=(\phi_t,\gamma_t,\a_t)$, while $\theta_{t,0}=(\phi_t^{\mathrm H},\gamma_t^{\mathrm F},0)$. With an application of Gr\"onwall's Lemma we conclude that 
	\[
	\|\theta_t-\theta_{t,0}\|_*\le Ct e^{Kt}\,,
	\] which proves the claim.
\end{proof}
Assumption~\ref{ass:scales} guarantees the diluteness of the thermal cloud when propagated with the free time evolution at all times. We now prove that this property continues to hold for the 1-pdm evolved with the HFB equations up to a double exponential growth in time. The precise statement is captured in the following proposition.
\begin{proposition}
	\label{lem:diluteness_HFB}
	Let the pair $(\phi_0,\gamma_0)$ satisfy Assumption~\ref{ass:scales} with $s \in (0,2]$, and let $v$ be as in Theorem \ref{thm:main_2}. Moreover, let $(\phi_t,\gamma_t,\a_t)$ denote the solution to the HFB equations in \eqref{eq:HFB}, with initial datum $(\phi_0,\g_0,0)$. Then there exist constants $c,C>0$ s.t. 
	\begin{equation}
		\sup_{x,y \in\RRR^3 } |\gamma_t(x,y)| \lesssim T_{\mathrm{c}}^{{3/2}}(s) \exp(c \exp(ct)) 
	\end{equation}
	holds for $t>0$ and $N \geq C \exp(c \exp(ct))$.
\end{proposition}

In order to prove Proposition~\ref{lem:diluteness_HFB} we need two lemmas. The first one concerns the growth of the $L^1$-norm of the Fourier transform of the solution to the time-dependent Hartree equation in time.
\begin{lemma}\label{lem:diluteness_hartree}
	Assume that $\phi_t$ is a solution to the time-dependent Hartree equation in \eqref{eq:timedepHartree} with initial condition $\phi_0 \in H^2(\mathbb{R}^3)$ and $v$ satisfies the same assumptions as in Theorem \ref{thm:main_2}. There exists a constant $c>0$ s.t.
	\begin{equation*}
		\Vert \widehat{\phi}_t \Vert_1 \lesssim \Vert \widehat{\phi}_0 \Vert_1 \exp(ct)\,,
	\end{equation*} where $\widehat{\phi}_t$ denotes the Fourier transform of $\phi_t$.
\end{lemma}
\begin{proof}
	We use Duhamel's formula to write $\phi_t$ as
	\[
	\phi_t = e^{\mathrm{i}\D t}\phi_0 -\mathrm{i} \int_0^t e^{\mathrm{i}(t-s)\D}\left[\left( N^{-1} v\ast|\phi_s|^2\right)\phi_s\right]\de s\,.
	\] 
	When take the $L^1$-norm in Fourier space on both sides, we find 
	\[
	\begin{split}
		\|\wh{\phi}_t\|_1\le & \|\wh{\phi}_0\|_1 + N^{-1}\int_0^t\|\hat v \ \wh{|\phi_s|^2}\|_1\|\wh{\phi}_s\|_1\de s\\
		\le & \|\wh{\phi}_0 \|_1 + N^{-1}\| \hat v\|_1\int_0^t\| \ \wh{|\phi_s|^2}\|_\infty\|\wh{\phi}_s\|_1\de s \\ 
		\le &  \|\wh{\phi}_0\|_1 + N^{-1}\|\hat v\|_1\|\phi_0\|^2\int_0^t\|\wh{\phi}_s\|_1\de s\lesssim \|\wh{\phi}_0\|_1 + \int_0^t\|\wh{\phi}_s\|_1\de s \,,
	\end{split}
	\] 
	where we used the conservation of the $L^2$-mass in the last step. An application of Gr\"onwall's Lemma yields
	\begin{equation*}
		\|\wh{\phi}_s \|_1 \lesssim \| \wh {\phi}_0 \|_1 \exp(ct)\,,
	\end{equation*} 
	which proves the claim. 
\end{proof}
\begin{remark}
	\label{remark:HartreeH2}
	If $v \geq 0$ or $\hat{v} \geq 0$ it can be shown that there exists a constant $C>0$ depending only on $\Vert \phi_0 \Vert_2$, $\mathcal{E}^{\mathrm{H}}(\phi)$, with $\mathcal{E}^{\mathrm{H}}$ defined in \eqref{eq:Hartreefunctional}, and $\epsilon>0$ s.t.
	\begin{equation}
		\Vert \phi_t \Vert_{H^2(\mathbb{R}^3)} \leq C (1+|t|^{2+\epsilon}) \Vert \phi_0  \Vert_{H^2(\mathbb{R}^3)} \,.
		\label{eq:polynomialtimedep}
	\end{equation}
	This result follows from the analysis in \cite[Section~5]{Soh2011}. It allows us to replace the exponential time dependence in Lemma~\ref{lem:diluteness_hartree} by the polynomial time-dependence on the r.h.s. of \eqref{eq:polynomialtimedep} (because the $H^2(\mathbb{R}^3)$-norm of a function dominates the $L^1(\mathbb{R}^3)$-norm of its Fourier transform), the double exponential time-dependence in Proposition~\ref{lem:diluteness_HFB} by $\exp(c(t+t^{3+\epsilon}))$, and the triple exponential time-dependence in Theorem~\ref{thm:main_2} by $\exp(c \exp( c (t+t^{3+\epsilon})))$.
\end{remark} The second lemma is a generalization of Gr\"onwall's inequality which follows from \cite[Corollary 3.2]{GronwallQuad}.
\begin{lemma}\label{lem:GronwallQuad}
	Let $p\in\NN\setminus\{0,1\}$, $T>0$ and let $a,b,c$ be continuous, nonnegative and nondecreasing functions on $[0,T]$ such that $e^{a(t)(p-1)T}< 1+b(t)/(a(t)^{p-1}c(t))$, for every $t\in[0,T]$. If $u\in L^\infty[0,T]$ satisfies $u(t)\ge 0$ and
	\[
	u(t)\le a(t)+b(t)\int_0^tu(s)\de s+c(t)\int_0^tu^p(s)\de s
	\] for a.e. $t\in[0,T]$, then, for a.e. $t\in[0,T]$, we have 
	\begin{equation*}
		u(t)\le a(t)\left[1-\frac{a(t)^{p-1}c(t)}{b(t)}\left(e^{(p-1)b(t)t}-1\right)\right]^{-\frac 1{p-1}}e^{b(t)t}\,.
	\end{equation*}
\end{lemma}
\begin{remark}
	As a consequence of the power-law nonlinearity in the integral inequality in Lemma~\ref{lem:GronwallQuad}, the bound for $u(t)$ blows up as $t$ approaches the smallest time $T^*$ such that
	\begin{equation*}
		T^*=b(T^*)^{-1}(p-1)^{-1}\ln(1+b(T^*)/(a(T^*)^{p-1}c(T^*)))\,.
	\end{equation*}
	For times satisfying $0 < t \leq T^*/2$, the nonlinearity is irrelevant and functions obeying the integral inequality satisfy the standard Gr\"onwall estimate 
	\begin{equation*}
		u(t) \lesssim a(t) \exp(b(t)t)\,.
	\end{equation*}
\end{remark} 

We are now prepared to give the proof of Proposition~\ref{lem:diluteness_HFB}.

\begin{proof}[Proof of Proposition~\ref{lem:diluteness_HFB}]
	For a Hilbert-Schmidt operator $\gamma$ we use the notation $\|\gamma\|_d=\int|\hat \gamma(p,q)|\de p\de q$, where $\hat\gamma(p,q)$ denotes the Fourier transform of the integral kernel of $\gamma$. For $\phi\in L^2(\RR^3)$ we also define $\|\phi\|_\rmd=\|\wh\phi\|_1$. We will prove that there exists a constant $c > 0$, independent of $t$ and $N$, such that 
	\begin{align*}
		\|\gamma_t\|_\rmd \lesssim& T_{\mathrm{c}}^{3/2}(s) \exp(c  \exp(c t))\,, \\
		\|\a_t\|_\rmd \lesssim& T_{\mathrm{c}}^{3/2}(s) \exp(c  \exp(c t))\,, \\
		\|\phi_t-\phi_t^{\mathrm H}\|_\rmd \lesssim&  \exp(c \exp(c t))\,,
	\end{align*} 
	holds for $t>0$, and $N$ large enough depending on $t$, which implies the claim. To this end, we start by noticing that the operator $e^{-\mathrm{i}A(t-s)}$, with $A$ as in \eqref{eq:linpartHFB}, acts in Fourier space by multiplication with a complex phase, and therefore preserves the norm $\|\cdot\|_\rmd$ for $\phi$, $\gamma$, and $\alpha$, for every $t>s>0$. Thus \eqref{eq:Duhamel_omega} implies, for $t>0$,
	\begin{equation}\label{eq:Duhamel_gamma_1}
		\|\gamma_t\|_\rmd\le \|\gamma_t^{\mathrm F}\|_{\mathrm d} + \int_0^t\left\|b(\gamma_s^{\phi_s})\gamma_s+k(\a_s^{\phi_s})\a_s^*-\a_s k(\a_s^{\phi_s})^*\right\|_{\mathrm d}\de s \,,
	\end{equation}
	with $\gamma^{\mathrm F}_t$ defined below \eqref{eq:closeness_dynamics}. Young's convolution inequality implies 
	\begin{equation*}
		\left\|(N^{-1}v\ast \r_{\gamma_s^{\phi_s}})\gamma_s\right\|_\rmd = N^{-1}\left\|\left(\hat v\wh{\r_{\gamma_s^{\phi_s}}}\right)\ast_1\hat{\gamma}_s(\cdot,\cdot)\right\|_1 
		\le N^{-1}\| \hat v \|_{1}  \left( \|\phi_0\|^2+\|\g_0\|_1\right) \|\gamma_s\|_\rmd\lesssim \|\gamma_s\|_\rmd\,,
	\end{equation*}	
	where $f\ast_i\hat\gamma_s(\cdot,\cdot)$ denotes the convolution of $f$ with the $i$-th variable of the kernel $\hat\gamma_s(p,q)$. A straightforward computation shows that, for any Hilbert-Schmidt operator $\gamma$, we have 
	\begin{equation} \label{eq:exchange_fourier}
		\wh{(v\sharp\gamma)}(p,q)=(2\pi)^{-3}\int_{\mathbb{R}^3} \hat v(\x)\hat\gamma(q-\x,p-\x)\de\x\,.
	\end{equation} It follows that
	\[
	\begin{split}
		\|N^{-1}\left(v\sharp\gamma_s^{\phi_s}\right)\gamma_s\|_\rmd = & N^{-1}\int_{\mathbb{R}^6} \left|\int_{\mathbb{R}^3} \hat v(\x)\left(\hat \gamma_s(z-\x,p-\x)+\cc{\wh{\phi}_s(z-\x)}\wh{\phi}_s(p-\x)\right)\hat\gamma_s(z,q)\de\x\de z\right|\de p\de q \\
		\le & N^{-1}  \|\gamma_s\|_\rmd\cdot\left(\sup_{z\in\RRR^3}\int_{\mathbb{R}^6}|\hat v(\x)|\left(|\hat \gamma_s(z-\x,p-\x) |+|\wh{\phi}_s(z-\x)\wh{\phi}_s(p-\x)|  \right)\de\x\de p\right)\,. 
	\end{split}
	\] 
	Bounding $|\hat v(\x)|\le \|v\|_1$ and integrating first in $p$ and then in $\x$, we find 
	\[
	\begin{split}
		\|N^{-1}\left(v\sharp\gamma_s^{\phi_s}\right)\gamma_s\|_\rmd \le & N^{-1}\|v\|_1\left(\|\gamma_s\|_\rmd+ \|\phi_s\|_\rmd^2 \right) \|\gamma_s\|_\rmd \\
		\lesssim & N^{-1}\|v\|_1\left(\|\gamma_s\|_\rmd+ \|\phi_s-\phi_s^{\mathrm H}\|_\rmd^2+\|\phi_s^{\mathrm H}\|_\rmd^2 \right) \|\gamma_s\|_\rmd\\
		\lesssim & N^{-1}\|\gamma_s\|_\rmd^2+N^{-1}\|\gamma_s\|_\rmd\|\phi_s-\phi_s^{\mathrm H}\|_\rmd^2 + \exp(ct) \|\gamma_s\|_\rmd\,.
	\end{split}
	\]
	In the last step we used Lemma~\ref{lem:diluteness_hartree}. Similarly, we see that
	\[
	\begin{split}
		\left\|k(\a^{\phi_s})\a_s^*-\a_sk(\a^{\phi_s})^*\right\|_\rmd\le & N^{-1}\|v\|_1\left(\|\a_s\|_\rmd+ \|\phi_s\|_\rmd^2 \right) \|\a_s\|_\rmd \\
		\le & N^{-1}\|v\|_1\left(\|\a_s\|_\rmd+ \|\phi_s-\phi_s^{\mathrm H}\|_\rmd^2+\|\phi_s^{\mathrm H}\|_\rmd^2 \right) \|\a_s\|_\rmd\\
		\lesssim & N^{-1}\|\a_s\|_\rmd^2+N^{-1}\|\a_s\|_\rmd\|\phi_s-\phi_s^{\mathrm H}\|_\rmd^2 + \exp(ct) \|\a_s\|_\rmd\,.
	\end{split}
	\] In combination, \eqref{eq:Duhamel_gamma_1} and the above bounds imply
	\begin{align}
		\|\gamma_t\|_\rmd\le & \|\gamma_t^{\mathrm F}\|_{\mathrm d} + C \exp(ct) \int_0^t \left( \|\gamma_s\|_\rmd+\|\a_s\|_\rmd \right) \de s \nonumber \\
		& +  CN^{-1}\int_0^t \left[ \|\gamma_s\|_\rmd^2+\|\a_s\|_\rmd^2+\|\phi_s-\phi_s^{\mathrm H}\|_\rmd^2\left(\|\gamma_s\|_\rmd+\|\a_s\|_\rmd\right) \right] \de s \,.\label{eq:Duhamel_final_bound_omega_1}
	\end{align}
	
	Next, we consider the r.h.s. of \eqref{eq:Duhamel_alpha}. Estimates similars to the ones above show
	\begin{align*}
		\left\| [b(\gamma_s^{\phi_s}),\a_s]_++[k(\a_s^{\phi_s}),\gamma_s]_+ \right\|_\rmd\lesssim & N^{-1}\|v\|_1\Big( \left( \|\phi_s-\phi_s^{\mathrm H}\|_\rmd^2+\|\phi_s^{\mathrm H}\|_\rmd^2 \right)\left(\|\gamma_s\|_\rmd + \|\a_s\|_\rmd \right) +  \|\gamma_s\|_\rmd \|\a_s\|_\rmd \Big)\\
		\lesssim & N^{-1}\left(\|\gamma_s\|_\rmd+\|\a_s\|_\rmd\right)\|\phi_s-\phi_s^{\mathrm H}\|_\rmd^2 + \exp(ct) \left(\|\gamma_s\|_\rmd+\|\a_s\|_\rmd\right)+N^{-1}\|\gamma_s\|_\rmd\|\a_s\|_\rmd
	\end{align*}
	and 
	\[
	\|k(\a_s^{\phi_s})\|_\rmd\le N^{-1}\|\hat v\|_1\left( \|\a_s\|_\rmd+\|\phi_s-\phi_s^{\mathrm H}\|_\rmd^2+\|\phi_s^{\mathrm H}\|_\rmd^2 \right)\lesssim N^{-1}\|\a_s\|_\rmd+N^{-1}\|\phi_s-\phi_s^{\mathrm H}\|_\rmd^2+\exp(ct)\,,
	\] which implies 
	\begin{align}
		\|\a_t\|_\rmd\lesssim & \exp(ct)+\exp(ct) \int_0^t \left( \|\gamma_s\|_\rmd+\|\a_s\|_\rmd \right) \de s \nonumber \\
		& +  N^{-1}\int_0^t \left[ \|\gamma_s\|_\rmd^2+\|\a_s\|_\rmd^2+\|\phi_s-\phi_s^{\mathrm H}\|_\rmd^2\left(1+\|\gamma_s\|_\rmd+\|\a_s\|_\rmd\right) \right] \de s \,. \label{eq:Duhamel_final_bound_alpha_1}
	\end{align}
	
	Finally, we consider \eqref{eq:Duhamel_phi}. We use the identity 
	\[
	|\phi_s|^2-|\phi_s^{\mathrm H}|^2=\phi_s(\cc{\phi_s-\phi_s^{\mathrm H}})+(\phi_s-\phi_s^{\mathrm H})\cc{\phi_s^{\mathrm H}}\,,
	\] Lemma~\ref{lem:diluteness_hartree}, and Young's inequality to see that
	\[
	\begin{split}
		\left\|\left( N^{-1} v\ast(|\phi_s|^2-|\phi_s^{\mathrm H}|^2)\right)\phi_s^{\mathrm H} \right\|_\rmd \le & N^{-1} \left\| \hat v\left(\left(\wh{\cc{\phi_s}}-\wh{\cc{\phi_s^{\mathrm H}}}\right)\ast\wh{\phi_s}+\left(\phi_s-\phi_s^{\mathrm H}\right)\wh{\cc{\phi_s^{\mathrm H}}}\right) \right\|_1\|\phi_s^{\mathrm H}\|_\rmd\\
		\lesssim & N^{-1/2} \exp(ct) \|v\|_2\|\phi_s-\phi_s^{\mathrm H}\|_\rmd\left(\|\phi_s\|+\|\phi_s^{\mathrm H}\|\right)\lesssim \exp(ct) \|\phi_s-\phi_s^{\mathrm H}\|_\rmd\, ,
	\end{split}
	\] as well as
	\[
	\begin{split}
		\left\|\left(N^{-1} v\ast|\phi_s|^2\right)(\phi_s-\phi_s^{\mathrm H})\right\|_\rmd\le & N^{-1} \|\hat v  \wh{|\phi_s|^2}\|_1\|\phi_s-\phi_s^{\mathrm H}\|_\rmd\lesssim\|\phi_s-\phi_s^{\mathrm H}\|_\rmd\,.
	\end{split}
	\] Out of the remaining terms on the r.h.s. of \eqref{eq:Duhamel_phi}, the direct term is bounded by
	\[
	\begin{split}
		\|N^{-1}v\ast \r_{\g_s}\phi_s\|_\rmd = & N^{-1} \|(\hat v\wh{\r_{\g_s}})\ast\wh\phi_s\|_1 \le N^{-1}  \|\hat v\wh{\r_{\g_s}}\|_1\|\wh\phi_s\|_1 \,.
	\end{split}
	\] We use the identity $\wh{\r_{\g_s}}(\x)=(2\pi)^{-3/2}\int\hat\gamma_s(p,p-\x)\de p$ to see that
	\[
	\|\wh{\r_{\g_s}}\|_1\le(2\pi)^{-3/2}\int_{\mathbb{R}^6}|\hat\gamma_s(p,p-\x)|\de p \de \x\lesssim \|\gamma_s\|_\rmd\,,
	\] which implies
	\begin{equation*}
		\|N^{-1}v\ast \r_{\g_s}\phi_s\|_\rmd \lesssim  N^{-1} \|v\|_1\| \gamma_s \|_\rmd\left(\|\phi_s-\phi_s^{\mathrm H}\|_\rmd+\|\phi_s^{\mathrm H}\|_\rmd\right) 
		\lesssim  N^{-1}\|\gamma_s \|_\rmd\|\phi_s-\phi_s^{\mathrm H}\|_\rmd + N^{-1/2} \exp(ct) \|\gamma_s\|_\rmd\,.
	\end{equation*}
	Using \eqref{eq:exchange_fourier} another time, we can bound the remaining term on the r.h.s. of \eqref{eq:Duhamel_phi} by
	\[
	\begin{split}
		\|N^{-1}v\sharp(\gamma_s+\a_s)\phi_s\|_\rmd \le & N^{-1} \int_{\mathbb{R}^3}\left|\int_{\mathbb{R}^6}\hat v(\x)\Big(\hat\gamma(q-\x,p-\x)+\hat\a(q-\x,p-\x)\Big)\wh\phi_s(q)\de\x\de q\right|\de p\\
		\le & N^{-1} \|v\|_1\left(\|\gamma_s\|_\rmd+\|\a_s\|_\rmd\right)\left(\|\phi_s-\phi_s^{\mathrm H}\|_\rmd+\|\phi_s^{\mathrm H}\|_\rmd\right)\\
		\lesssim & N^{-1}\left(\|\gamma_s\|_\rmd+\|\a_s\|_\rmd\right)\|\phi_s-\phi_s^{\mathrm H}\|_\rmd + N^{-1/2} \exp(ct) \left(\|\gamma_s\|_\rmd + \|\a_s\|_\rmd\right)\,.
	\end{split}	
	\] In combination, \eqref{eq:Duhamel_phi} and the above estimates imply 
	\begin{equation}\label{eq:Duhamel_phi_final_1}
		\|\phi_t-\phi_t^{\mathrm H}\|_\rmd\lesssim \exp(ct) \int_0^t \left[ \|\phi_s-\phi_s^{\mathrm H}\|_\rmd+N^{-1/2}\left(\|\gamma_s\|_\rmd+\|\a_s\|_\rmd\right) \right] \de s + N^{-1}\int_0^t\|\phi_s-\phi_s^{\mathrm H}\|_\rmd\left(\|\gamma_s\|_\rmd+\|\a_s\|_\rmd\right)\de s\,.
	\end{equation}
	
	Let us define $u(t)=T_{\mathrm{c}}^{-3/2}(s)\|\gamma_t\|_\rmd+T_{\mathrm{c}}^{-3/2}(s)\|\a_t\|_\rmd+\|\phi_t-\phi_t^{\mathrm H}\|_\rmd$. We collect the bounds in \eqref{eq:Duhamel_final_bound_omega_1}, \eqref{eq:Duhamel_final_bound_alpha_1}, and \eqref{eq:Duhamel_phi_final_1}, and use the assumption $s \leq 2$ to find 
	\[
	\begin{split}
		u(t)\lesssim & 1 + \exp(ct) T_{\mathrm{c}}^{-3/2}(s) + \exp(ct) \int_0^t u(z) \de z+N^{-1} T_{\mathrm{c}}^{3/2}(s)\int_0^t u(z)^2\de z +N^{-1}\int_0^t u(z)^3\de z\\
		\lesssim &  \exp(ct)+ \exp(ct) \int_0^t u(z)\de z+ N^{-1}\int_0^t u(z)^3\de z\,,
	\end{split}
	\] 
	for every $t>0$ and some $c>0$ independent of $t$ and $N$. For any $t>0$, the assumptions of Lemma~\ref{lem:GronwallQuad} are satisfied with $T=t$, as long as $N \geq C \exp(c\exp(ct))$,
	and we conclude that, under this condition,
	\[
	u(t) \lesssim \exp(c  \exp(c t))\,,
	\] which proves the claim.
\end{proof}
\section{Fluctuation dynamics}
\label{sec:fluctdyn}
As explained in Section~\ref{sec:proofstrategy}, the proof of Theorem~\ref{thm:main_2} is based on the construction of a suitable fluctuation dynamics on the double Fock space $\mathscr{F}(\fh \oplus \fh)$ with $\fh = L^2(\mathbb{R}^3)$. In this section we construct the fluctuation dynamics and discuss some of its properties. We start by recalling a few well known facts about Weyl operators and Bogoliubov transformations. 
\subsection{Weyl operators and Bogoliubov transformations}
\label{sec:weylandBogtrafos}
We recall the definition of the Weyl operator $W(\phi)$ acting on $\mathscr{F}(\fh)$ in \eqref{eq:Weyltrafo} as well as the Weyl operator $\mathcal{W}(\phi)$ acting on $\mathscr{F}(\fh \oplus \fh)$ in \eqref{eq:doubleFPtrafos}. Both operators act as shifts on creation and annihilation operators, i.e.
\begin{equation}\label{eq:Weyl_shift}
	W(\phi)^* a_x W(\phi)= a_x+\phi(x)\,, \qquad W(\phi)^* a_x^* W(\phi)= a_x^*+\cc{\phi(x)}\,,
\end{equation} 
and
\begin{align}
	\cW(\phi)^* a_{\ell,x} \cW(\phi)&= a_{\ell,x}+\phi(x)\,, \qquad \mathcal{W}(\phi)^* a_{\ell,x}^* \mathcal{W}(\phi)= a_{\ell,x}^*+\cc{\phi(x)}\,, \nonumber\\
	\cW(\phi)^* a_{r,x} \cW(\phi)&= a_{r,x}+\cc{\phi(x)}\,, \qquad \mathcal{W}(\phi)^* a_{r,x}^* \mathcal{W}(\phi)= a_{r,x}^*+\phi(x)\,. \label{eq:Weyl_shift_double}
\end{align}
To express \eqref{eq:Weyl_shift_double} in a more compact form, we denote $\xx = ( \sigma, x ) \in \{ \ell, r \} \times \RR^3 $ and for a given $\phi\in\fh$ we define $\pphi:\{\ell,r\} \times \RR^3\to \CC$ by 
\[
\pphi(\xx) = \pphi(\sigma,x) = \begin{cases}
	\phi(x), & \sigma = \ell\,, \\
	\cc{\phi(x)}, & \sigma = r \,.
\end{cases}
\] 
This allows us to write \eqref{eq:Weyl_shift_double} as
\begin{equation}\label{eq:Weyl_shift_compact}
	\cW(\phi)^*a_\xx\cW(\phi)=a_\xx+\pphi(\xx)\,,\qquad \cW(\phi)^*a^*_\xx\cW(\phi)=a^*_\xx+\cc{\pphi(\xx)}\,.
\end{equation}

Next we recall some well known facts about Bogoliubov transformations. We are interested in the case where the one-particle Hilbert space is given by $\fh$ or by $\fh \oplus \fh = \fh^2$. The complex conjugate $\cc A$ of an operator $A\in\cB( \fh^n )$ with $n\in \{ 1,2 \}$ is defined as the operator whose integral kernel, in position space, is the complex conjugate of the kernel of $A$. For $n=1,2$ a bounded linear map $\n:(\fh^n \oplus \fh^n )\to( \fh^n \oplus \fh^n)$ of the form
\begin{equation}
	\n=\begin{pmatrix} 
		U & \cc V\\
		V & \cc U
	\end{pmatrix}
	\label{eq:nu}
\end{equation}
is called a symplectomorphism if
\begin{equation}\label{eq:relations_bog}
	\n^* \mathcal{S} \nu = \mathcal{S} \quad \text{ and } \quad \n \mathcal{S} \nu^* = \mathcal{S} \quad \text{ with } \quad \mathcal{S} = \begin{pmatrix}
		\mathds{1} & 0 \\ 0 & - \mathds{1} 
	\end{pmatrix}\,.
\end{equation} 
Here $\mathds{1}$ denotes the identity operator on $\fh^n$. If $V$ is a Hilbert--Schmidt operator, $\nu$ is implementable on $\mathscr{F}(\fh^n)$, see \cite{BachBretaux,Solovej}, i.e. there exists a unitary $T_\n$ on $\fock(\fh^n)$ such that 
\begin{align}
	T_\n^*a(f)T_\n &= a(Uf) + a^*(\cc{V} \ \cc{f})\,, \nonumber \\
	T_\n^*a^*(f)T_\n &= a^*(Uf) + a(\cc{V} \ \cc{f})\,, \label{eq:action_bogolubov}
\end{align}
hold for $f\in\fh^n$. The operator $T_\nu$ is called the Bogoliubov transformation corresponding to $\nu$. 

For Bogoliubov transformations on $\mathscr{F}(\fh^2)$ it is convenient to introduce a compact notation, similar to \eqref{eq:Weyl_shift_compact}. For a symplectomorphism $\n:(\fh^2 \oplus \fh^2 )\to( \fh^2 \oplus \fh^2)$ of the form \eqref{eq:nu}, with $U,V \in \mathcal{B}(\fh^2)$, we denote by $U(\xx,\yy)$, $V(\xx,\yy)$ the kernels\footnote{Since $V \in \mathcal{L}^2(\fh^2)$ its kernel is a function in $L^2(\mathbb{R}^3) \oplus L^2(\mathbb{R}^3)$. The kernel of $U$ is a distribution.} of $U,V$, that is, 
\begin{equation*}
	U(\xx,\yy) = U_{\sigma,\sigma'}(x,y) \,,
\end{equation*}
where $\xx = (\sigma,x)$ and $\yy = (\sigma',y)$, and the same for $V$. We also define $U_\xx(\yy)= U(\yy,\xx)$, $V_\xx(\yy)=V(\yy,\xx)$. With this notation, we can write 
\begin{align}
	T_\nu^*a_\xx T_\nu &= a(U_{\xx})+a^*(\cc{V_{\xx}}) \,, \nonumber \\
	T_\n^*a_\xx^* T_\n &= a^*(U_{\xx}) + a(\cc{V_{\xx}})\,. \label{eq:Tt_action}
\end{align} 
\subsection{Construction of the fluctuation dynamics}
\label{sec:constrcutionfluctdyn}
In this section we use the Hartree--Fock--Bogoliubov (HFB) equations in \eqref{eq:HFB} to define our fluctuation dynamics. As explained in Section~\ref{sec:proofstrategy}, we do not intend to derive the HFB equations, but rather use them as an intermediate step between the many-body evolution and the simple effective dynamics appearing in Theorem~\ref{thm:main_2} (or in Remark~\ref{rmk:restriction_s}.3).

First of all, we discuss how our reference state, associated with the solution of the HFB equations, can be conveniently expressed as a vector in $\mathscr{F} (\fh^2 )$. Let $(\phi, \gamma)$ satisfy Assumption~\ref{ass:scales} and let $(\phi_t, \gamma_t, \alpha_t)$ be the solution of the HFB equations \eqref{eq:HFB}, with initial datum $(\phi, \gamma,0)$. We denote by $G_t$ the density matrix of the quasi-free state on $\mathscr{F} (\frak{h})$, with 1-pdm $\g_t$ and pairing density $\a_t$, and we define its generalized 1-pdm by
\begin{equation}
	\Gamma_t^{(1)} = \begin{pmatrix}
		\gamma_t & \alpha_t \\ \overline{\alpha}_t & 1+\overline{\gamma}_t \end{pmatrix}
	\label{eq:gen1pdmt}
\end{equation} 
(the condensate will be added later on, through conjugation with a Weyl operator). As proven in \cite[Proposition~3.9]{BachBretaux}, we can write 
\begin{equation}
	\Gamma_t^{(1)} = \cU_t^* \Gamma_0^{(1)} \cU_t 
	\label{eq:diagallt}
\end{equation}
with 
\begin{equation*}
	\Gamma_0^{(1)} = \begin{pmatrix}
		\gamma & 0 \\ 0 & 1+\overline{\gamma} \end{pmatrix} 
\end{equation*}
and an implementable symplectomorphism
\begin{equation}\label{eq:def_symplecto_Ut}
	\cU_t = \begin{pmatrix}
		\pi_t & \cc{\th}_t\\
		\th_t & \cc{\pi}_t
	\end{pmatrix} 
\end{equation}
satisfying the equation
\begin{equation}\label{eq:Cauchy_problem_symplecto_Ut}
	\mathrm{i} \partial_t \cU^*_t = \cS \Lambda (\Gamma_t^{(1)}) \cU_t^* \quad \text{ with } \quad \L(\G_t^{(1)})=\begin{pmatrix}
		h(\gamma_t^{\phi_t}) & k(\a_t^{\phi_t})\\
		\cc{k}(\a_t^{\phi_t}) & \cc{h}(\gamma_t^{\phi_t})
	\end{pmatrix} \,,
\end{equation} 
with $h(\gamma), k(\alpha)$ defined in \eqref{eq:h(g)_k(g)_def} and the initial condition $\cU_0 = \mathds{1}$ (with $\mathcal{S}$ defined in \eqref{eq:relations_bog}). Raising $\cU_t$ to $\fh^2$, we obtain the map $R_t : \fh^2 \to \fh^2 $, defined by 
\begin{equation*}
	R_t = \begin{pmatrix}
		\Pi_t & \cc{\Th}_t\\
		\Th_t & \cc{\Pi}_t
	\end{pmatrix} \,,
\end{equation*} 
where we defined
\begin{equation*}
	\Pi_t=\begin{pmatrix}
		\pi_t & 0\\
		0 & \cc{\pi}_t
	\end{pmatrix} \quad \text{ and } \quad \Th_t = \begin{pmatrix}
		\th_t & 0\\
		0 & \cc{\th}_t 
	\end{pmatrix}\,.
\end{equation*} 
As $\cU_t$, the map $R_t$ is an implementable symplectomorphism, for all $t \in \bR$. Hence, we find a family of unitary operators $\cR_t$ on $\mathscr{F} (\fh^2)$, with 
\begin{equation*}
	\mathcal{R}_t^* a_{\xx} \mathcal{R}_t =  a (\Pi_{t,{\xx}}) + a^* (\overline{\Theta_{t,\xx}}) \quad \text{ and } \quad \mathcal{R}_t^* a^*_{\xx} \mathcal{R}_t =  a^* (\Pi_{t,{\xx}}) + a (\overline{\Theta_{t,\xx}}) \,.
\end{equation*}
Let us also introduce the unitary family $\mathcal{T}_t = \mathcal{R}_t \mathcal{T}(\gamma)$, with $\mathcal{T}(\gamma)$ as defined in \eqref{eq:def_T(gamma)}. We observe that $\mathcal{T}_t$ is again a family of Bogoliubov transformations, satisfying 
\begin{equation}\label{eq:action_time_dep_Bog}
	\mathcal{T}^*_t a_{\xx} \mathcal{T}_t = a (U_{t,{\xx}}) + a^* (\overline{V_{t,\xx}}) \quad \text{ and } \quad \mathcal{T}^*_t a^*_{\xx} \mathcal{T}_t = a^* (U_{t,{\xx}}) + a (\overline{V_{t,\xx}})
\end{equation}
with 
\begin{equation}\label{eq:def_Ut_Vt}
	U_t = U_0 \Pi_t + \overline{V}_0 \Theta_t \,, \quad \quad V_t = V_0 \Pi_t + \overline{U}_0 \Theta_t \,, 
\end{equation}	
and
\begin{equation}\label{eq:def_U0_V0}
	U_0 = \begin{pmatrix}
		\sqrt{1+\gamma} & 0 \\ 0 & \overline{\sqrt{1+\gamma}}
	\end{pmatrix}, \quad \quad V_0 = \begin{pmatrix}
		0 & \overline{\sqrt{\gamma}} \\ \sqrt{\gamma} & 0 
	\end{pmatrix} \,.
\end{equation}
In other words, $\mathcal{T}_t$ implements the symplectomorphism 
\begin{equation}\label{eq:symplectoTt}
	T_t = T_0 R_t = \begin{pmatrix}
		U_t & \overline{V}_t \\ V_t & \overline{U}_t
	\end{pmatrix} \,.
\end{equation}	
The following lemma provides us with bounds for the operator norms of $U_t$ and $V_t$.	
\begin{lemma}\label{lem:closeness_dynamics_V}
	Under the assumptions of Proposition~\ref{lem:closeness_dynamics}, there exists a constant $c>0$ such that
	\begin{equation} \label{eq:bound_Ut_Vt}
		\|U_t\| \lesssim \sqrt{T_{\mathrm{c}}(s)} \exp(ct) \quad \text{ and } \quad \|V_t\| \lesssim \sqrt{T_{\mathrm{c}}(s)} \exp(ct) 
	\end{equation} 
	hold. Here $\|\cdot\|$ denotes the norm of a linear operator on $\fh\oplus\fh$.
\end{lemma} 
\begin{proof}
	We first prove a bound on the operators $\pi_t$, $\th_t$ in \eqref{eq:def_symplecto_Ut}, from which the claim will follow. To do this we observe that, as a consequence of \eqref{eq:Cauchy_problem_symplecto_Ut}, $\pi_t$, $\th_t$ satisfy the system of ODEs
	\begin{align}
			\mathrm i \dt \pi_t^* &= h(\gpt)\pi_t^*+k(\apt)\cc{\th_t^*} \,,\nonumber \\
			\mathrm i \dt \th_t^* &= h(\gpt)\th_t^*+k(\apt)\cc{\pi_t^*} \,. \label{eq:system_pi_th}	
	\end{align} 
	Let $\psi$ with $\|\psi\|\le 1$ belong to an appropriate dense subset of $L^2(\RR^3)$. Using \eqref{eq:system_pi_th} and the fact that $h(\gpt)$ is self-adjoint we get
	\[
	\mathrm i\frac{\de}{\de t}\|\pi_t\psi\|^2=\mathrm i\frac{\de}{\de t}\langle\psi,\pi_t\pi_t^*\psi\rangle=\langle\psi,\left(\pi_tk(\apt)\cc{\th_t^*}-\cc{\th_t}k(\apt)^*\pi_t^*\right)\psi\rangle = 2\mathrm i \ \Im \langle\psi,\pi_tk(\apt)\cc{\th_t^*}\psi\rangle\,, 
	\] thus we can estimate
	\begin{equation}\label{eq:der_pi_bound}
		\frac{\de}{\de t}\|\pi_t\psi\|^2 \le 2\|k(\apt)\| \left(\|\pi_t\psi\|^2+\|\th_t\psi\|^2\right)\,.
	\end{equation} Proceeding in the same way, we also find 
	\begin{equation}\label{eq:der_th_bound}
		\frac{\de}{\de t}\|\th_t\psi\|^2 \le 2\|k(\apt)\| \left(\|\pi_t\psi\|^2+\|\th_t\psi\|^2\right)\,.
	\end{equation} The operator norm of $k(\apt)$ is easily bounded by 
	\[
	\|k(\apt)\|\le \|k(\apt)\|_{\cL^2}\le N^{-1}\|v\|_\infty \left( \int\left(|\a_t(x,y)|^2+|\phi_t(x)|^2|\phi_t(y)|^2\right)\de x\de y \right)^{1/2} \le c/2\,,
	\] for some constant $c>0$ independent of $t$, $N$ and $\psi$. To obtain the bound, we used $\Vert \alpha_t \Vert_{\mathcal{L}^2}^2 \leq \Vert \gamma_t \Vert_{\mathcal{L}^1}^2 + \Vert \gamma_t \Vert_{\mathcal{L}^1} \leq N(N+1)$ and $\Vert \phi_t \Vert^2 \leq N$. Inserting this into \eqref{eq:der_pi_bound} and \eqref{eq:der_th_bound} we find
	\[
	\frac{\de}{\de t}\left(\|\pi_t\psi\|^2+\|\th_t\psi\|^2\right) \le c \left(\|\pi_t\psi\|^2+\|\th_t\psi\|^2\right)\,,
	\] which together with the initial condition $\pi_0=\id$, $\th_0=0$ and an application of Gr\"onwall's Lemma lets us conclude
	\begin{equation}
		\|\pi_t\psi\|^2\le \exp(ct)\,,\qquad \|\th_t\psi\|^2 \le \exp(ct)\,.
	\end{equation} Therefore, taking the $\sup$ over $\psi$ we find
	\begin{equation}\label{eq:bound_pit_tht}
		\|\pi_t\|^2\le \exp(ct)\,,\qquad \|\th_t\|^2 \le \exp(ct)\,.
	\end{equation} Using \eqref{eq:bound_pit_tht}, \eqref{eq:def_Ut_Vt}, \eqref{eq:def_U0_V0} and point \textbf{(C)} of Assumption \ref{ass:scales}, we deduce \eqref{eq:bound_Ut_Vt}.
\end{proof}
We claim that $\mathcal{T}_t \Omega \in \mathscr{F} (\fh^2 )$ describes exactly the quasi-free state on $\mathscr{F}(\fh)$ with generalized 1-pdm \eqref{eq:gen1pdmt}. In fact, 
\begin{align}
	\langle \mathcal{U}^* \mathcal{T}_t \O, \left( a^*(g) a(f) \otimes \mathds{1} \right) \mathcal{U}^* \mathcal{T}_t \O \rangle &= \langle \mathcal{T}(\gamma) \O,\left( a_{\ell}^* (\pi_t g) + a_{\ell}(\cc{\th_t}\cc g) \right)\left( a_{\ell}(\pi_t f) +a^*_{\ell} (\cc{\th_t}\cc f) \right) \mathcal{T}(\gamma) \O\rangle \nonumber \\
	&= \langle f, \left[ \pi_t^* \g \pi_t + \th_t^*(\id+\cc{\g})\th_t \right] g\rangle \nonumber \\
	&=\langle f, \gamma_t g \rangle \,, \label{eq:1-pdmtimet}
\end{align}
with $\mathcal{U}$ in \eqref{eq:Utrafo}. To come to the last line we used \eqref{eq:diagallt}. A similar computation shows 
\begin{equation}
	\langle \mathcal{U}^* \mathcal{T}_t \O, \left( a(g) a(f) \otimes \mathds{1} \right) \mathcal{U}^* \mathcal{T}_t \O \rangle = \langle f, \alpha_t \overline{g} \rangle \,.
	\label{eq:pairingfunctiont}
\end{equation}
Next, we add the condensate. To this end, we apply the Weyl operator $\mathcal{W} (\phi_t)$. Proceeding as in \eqref{eq:1-pdmtimet} and \eqref{eq:pairingfunctiont}, and using \eqref{eq:Weyl_shift_double}, we find that the vector $\mathcal{W} (\phi_t) \mathcal{T}_t \Omega \in \mathscr{F} (\fh^2)$ describes the quasi-free state associated with the solution of the HFB equations. With this reference state, we are now ready to define the fluctuation dynamics as 
\begin{equation}
	\mathcal{U}^{\mathrm{fluct}}(t,s) = \mathcal{T}_t^* \mathcal{W}(\phi_t)^* \exp(-\mathrm{i} \mathcal{L}_N (t-s)) \mathcal{W}(\phi_s) \mathcal{T}_s
	\label{eq:fluctdyn}
\end{equation}
with $\mathcal{L}_N$ in \eqref{eq:Liouvillian}.
\subsection{Bound by the expectation of the number of particles}
\label{sec:Boundbyparticlenumber}
For $\xi \in \mathscr{F} (\fh^2)$, we define the fluctuation vector $\xi_t=\mathcal{U}^{\mathrm{fluct}}(t,0) \xi$. By definition, we have 
\begin{equation}
	\mathcal{W} (\phi_t) \cT_t \xi_t = e^{-\mathrm{i} \cL_N t} \cW (\phi) \cT_0 \xi = e^{-\mathrm{i} \cL_N t} \cW (\phi) \cT (\gamma) \xi \,,
\end{equation}
which is exactly the vector in $\mathscr{F}(\fh^2)$ associated with the state $\Gamma_{\xi,t}(\phi,\gamma)$ defined in \eqref{eq:Gammaxit} and considered in Theorem~\ref{thm:main_2}. Hence, the 1-pdm of the state $\Gamma_{\xi,t}(\phi,\gamma)$ can be written as 
\begin{equation*}
	\gamma_{\xi,t}(x,y) = \tr[ a^*_y a_x \Gamma_{\xi,t}(\phi,\gamma) ] = \langle \xi_t, \mathcal{T}^*_t \mathcal{W}^*(\phi_t) a^*_{\ell,y} a_{\ell,x} \mathcal{W}(\phi_t) \mathcal{T}_t \xi_t \rangle \,.
\end{equation*} 
A short computation that uses \eqref{eq:Weyl_shift_compact}, \eqref{eq:Tt_action} and
\begin{equation}\label{eq:ccr_V}
	\g_{t}(x,y) = \sum_{\sigma \in \{ \ell, r \}} \int_{\mathbb{R}^3} \cc {V_{t,\ell}((\sigma,z),x)}V_{t,\ell}((\sigma,z),y) \de z
\end{equation}
allows us to write the kernel of $\gamma_{\xi,t}$ as 
\begin{align}	
	\gamma_{\xi,t}(x,y) &= \phi_t(x) \overline{\phi_t(y)} + \g_t(x,y) + \langle \xi_t, \left( a^* (U_{t,\ell,y}) a(U_{t,\ell,x}) + a^* (\cc {V_{t,\ell,x}}) a(\cc {V_{t,\ell,y}})\right) \xi_t \rangle \nonumber \\ 
	&+\langle \xi_t, \left(  a (\cc {V_{t,\ell,y}}) a(U_{t,\ell,x})+a^* (U_{t,\ell,y}) a^*(\cc {V_{t,\ell,x}})\right) \xi_t \rangle \nonumber \\
	&+ \langle \xi_t, \left\{ \phi_t(x) \left( a^* (U_{t,\ell,y})  + a(\cc {V_{t,\ell,y}}) \right)  + \overline{\phi_t(y)} \left( a (U_{t,\ell,x})  + a^*(\cc {V_{t,\ell,x}}) \right) \right\} \xi_t \rangle \,. \label{eq:fluctdyn6}
\end{align}
With this representation for $\gamma_{\xi,t}$, we obtain the following lemma, which bounds the trace norm of $\gamma_{\xi,t} - |\phi_t \rangle \langle \phi_t| - \gamma_t$ in terms of the expectation of the number of particles operator 
\begin{equation}
	\cN=\cN_{\ell}+\cN_r+5 = \de\G(\id)+5\,.
	\label{eq:numberOp}
\end{equation}

\begin{lemma}\label{lem:bound_particle_number}
	Under the assumptions of Theorem~\ref{thm:main_2} and with $U_t, V_t$ as in \eqref{eq:def_Ut_Vt}, we have the bound
	\begin{align}
		\Vert \gamma_{\xi,t} - |\phi_t \rangle \langle \phi_t|-\gamma_t \Vert_{\mathcal{L}^1} \leq &\;  ( \|U_t\|^2+\|V_t\|^2) \langle \xi_t, \mathcal{N} \xi_t \rangle + 2 N^{1/2} ( \|U_t\|^2+\|V_t\|^2)^{1/2} \langle \xi_t, \mathcal{N} \xi_t \rangle^{1/2} 
		\nonumber \\
		&+ 2 \Vert \phi_t \Vert \cdot ( \|U_t\|^2+\|V_t\|^2)^{1/2} \langle \xi_t, \mathcal{N} \xi_t \rangle^{1/2} \,. \label{eq:bound_particle_number}
	\end{align}
\end{lemma}
\begin{proof} 
	Let $J \in \mathcal{B}(h)$. We take the trace of \eqref{eq:fluctdyn6} against $J$: this will allow us to prove \eqref{eq:bound_particle_number} by duality. Denoting by $J(x,y)$ the integral kernel of $J$, the contribution of the third term on the r.h.s. of \eqref{eq:fluctdyn6} is given by
	\begin{align}
		&\int_{\mathbb{R}^6} J(x,y) \left\langle \xi_t, \left[ a^* (\Uly)a(\Ulx) + a^* (\Vlx)a(\Vly) \right] \xi_t \right\rangle \de x \de y \nonumber \\
		&\hspace{1cm}= \int \left[ \cc{U_{t,\ell}(\zz_2,x)}J(x,y) U_{t,\ell}(\zz_1,y)  + \cc{V_{t,\ell}(\zz_2,x)}J(x,y) V_{t,\ell}(\zz_1,y)  \right] \left\langle \xi_t, a^*_{\zz_1}a_{\zz_2} \xi_t \right\rangle \de \zz_1\de\zz_2 \text{d}x\de y \nonumber \\
		&\hspace{1cm}=  \left\langle \xi_t, \de\G(J_1 + J_2)\xi_t \right\rangle\,, \label{eq:3}
	\end{align} 
	where we used the notation $\int \de \zz = \sum_{\sigma \in \{ \ell, r \}} \int_{\mathbb{R}^3} \de z$. The operators $J_1$ and $J_2$ are defined by
	\begin{align}
		J_1(\zz_1,\zz_2) =& \int_{\mathbb{R}^6} \cc{U_{t,\ell}(\zz_2,x)} J(x,y) \cc{U_{t,\ell}^*(y,\zz_1)} \text{d}x\de y, \quad J_2(\zz_1,\zz_2) = \int_{\mathbb{R}^6} \cc{V_{t,\ell}(\zz_2,x)} J(x,y) \cc{V_{t,\ell}^*(y,\zz_1)} \text{d}x\de y\,.
	\end{align}
	The absolute value of the r.h.s. of \eqref{eq:3} is bounded by
	\begin{align}
		\left| \left\langle \xi_t, \text{d} \Gamma ( J_1 + J_2 ) \xi_t \right\rangle\right| \leq (\Vert J_1 \Vert +\| J_2\| )\left\langle \xi_t, \mathcal{N} \xi_t \right\rangle \,.
		\label{eq:4}
	\end{align}
	With
	\begin{equation}
		\Vert J_1 \Vert \leq \Vert U_t \Vert^2 \Vert J \Vert \quad \text{ and } \quad \Vert J_1 \Vert \leq \Vert V_t \Vert^2 \Vert J \Vert
		\label{eq:5}
	\end{equation}
	we conclude that
	\begin{equation}
		\left| 	\int_{\mathbb{R}^6} J(x,y) \left\langle \xi_t, \left[ a^* (\Uly)a(\Ulx) + a^* (\Vlx)a(\Vly) \right] \xi_t \right\rangle \text{d}x\de y \right|  \leq \Vert J \Vert \left( \Vert U_t \Vert^2+\Vert V_t \Vert^2\right) \left\langle \xi_t, \mathcal{N}\xi_t \right\rangle \,. \label{eq:13}
	\end{equation}
	Next we consider the contribution of the fourth term on the r.h.s. of \eqref{eq:fluctdyn6}.
	
	Defining $J_3,J_4:\fh\to\fh\oplus\fh$ by
	\begin{equation}
		J_3(\zz,y)= \int_{\mathbb{R}^3} U_{t,l}(\zz,x) \cc{J(x,y)} \de x \,, \quad \quad J_4(\zz,x)= \int_{\mathbb R^3} J(x,y) U_{t,l}(\zz,y) \de y\,, \label{eq:13b}
	\end{equation}
	we write
	\begin{align}
		& \int_{\mathbb{R}^6} J(x,y) \langle \xi_t, \left[ a(\Vly) a(\Ulx) + a^*(\Uly) a^*(\Vlx) \right] \xi_t \rangle \de x \de y \nonumber \\
		&\hspace{3cm}= \int_{\mathbb{R}^3}  \langle \xi_t,  a(\Vly) a(J_{3,y}) \xi_t \rangle \de y +  \int_{\mathbb{R}^3}  \langle \xi_t,  a^*(J_{4,x}) a^*(\Vlx) \xi_t \rangle \de x \,. \label{eq:6} 
	\end{align} 
	The absolute value of the first term on the r.h.s. of \eqref{eq:6} is bounded by
	\begin{align} \label{2.17}
		\left( \int_{\mathbb{R}^3} \langle \xi_t, a(\Vly) a^*(\Vly) \xi_t \rangle \de y \right)^{1/2} \left( \int_{\mathbb{R}^3}  \langle \xi_t,  a^*(J_{3,y}) a(J_{3,y}) \xi_t \rangle \de y \right)^{1/2}\,.
	\end{align}
	Using the canonical commutation relations, \eqref{eq:ccr_V}, and $\tr\g_t \leq N$, we see that the first factor is not larger than
	\begin{equation}
		\| V_t \| \ \langle \xi_t ,\mathcal{ N}_r \xi_t \rangle^{1/2}+N^{1/2}\,.
	\end{equation}
	For the second factor we have the bound
	\begin{equation}
		\| J_3 \| \ \langle \xi_t ,\mathcal{ N} \xi_t \rangle^{1/2} \leq \| J \| \ \|U_t\| \ \langle \xi_t ,\mathcal{ N} \xi_t \rangle^{1/2}\,.
	\end{equation}
	The second term on the r.h.s. of \eqref{eq:6} can be bounded similarly by the same expression, and we thus find
	\begin{align}    
		&\left| \int_{\mathbb{R}^6} J(x,y) \langle \xi_t, \left[ a(\Vly) a(\Ulx) + a^*(\Uly) a^*(\Vlx) \right] \xi_t \rangle \de x \de y \right| \nonumber \\
		& \hspace{5cm} \leq 2 \left( \|J\| \ \|U_t\| \ \|V_t\| \ \langle \xi_t , \mathcal{N} \xi_t \rangle +  N^{1/2}\|J\| \ \|U_t\| \ \langle \xi_t , \mathcal{N} \xi_t \rangle^{1/2} \right)\,. \label{dust}
	\end{align} 
	It remains to provide bounds for the terms in the third line of \eqref{eq:fluctdyn6}. 
	
	We define the operator $J_5$ via its integral kernel
	\begin{equation}
		J_5(\zz,x) = \int_{\mathbb{R}^3} \overline{V_{t,l}(\zz,y)} \overline{J(x,y)} \de y 
		\label{dust2}
	\end{equation}
	and write
	\begin{align}
		\left| \int_{\mathbb{R}^6} \phi_t(x) J(x,y) \langle \xi_t, \left( a^*(\Uly)  + a(\Vly) \right) \xi_t \rangle \de(x,y) \right| &= \left| \int_{\mathbb{R}^3} \phi_t(x) \langle \xi_t, \left( a^*( J_{4,x} )  + a( J_{5,x}) \right) \xi_t \rangle \de x \right| \nonumber \\
		&\leq \Vert \phi_t \Vert_2 \ \Vert J \Vert \left( \Vert U_t \Vert^2 + \Vert V_t \Vert ^2\right)^{1/2} \langle \xi_t, \mathcal{N} \xi_t \rangle^{1/2}. \label{dust3}
	\end{align}
	The other term in the second line of \eqref{eq:fluctdyn6} can be estimated in the same way. In combination, \eqref{eq:fluctdyn6}, \eqref{eq:13}, \eqref{dust}, and \eqref{dust3} prove the claim.
\end{proof}
\subsection{Proof of Theorem~\ref{thm:main_2}}
\label{sec:proofoftheorem}
In the next proposition, whose proof is deferred to the next section, we control the growth of the expectation of the number of particles operator. 
\begin{proposition} \label{mainprop}
	Let $v$ satisfy the same assumptions as in Theorem \ref{thm:main_2}, and let the pair $(\phi,\gamma)$ satisfy Assumption~\ref{ass:scales} with $0 < s \leq 3/2$. For every $k \in \mathbb{N}$, there exist constants $c, C > 0 $ independent of $N,t$ such that
	\begin{align} \nonumber
		\langle \xi_t, \mathcal{N} \xi_t \rangle \leq & \exp(c \exp(c  \exp(c t))) \langle \xi, \cal N \xi\rangle  \\
		& + CN^{(22- k)/14}\exp(c\exp(c\exp(ct)))\left( \langle \xi, \mathcal{N}^{4} \xi \rangle^{1/2} + N^{10/7} \right) \big\langle \x, \cN^{k+2}\x \big\rangle^{1/2} \label{eq:Gronwall_bound_main}
	\end{align}
	for every $N\in\NN$, $t>0$.
\end{proposition}
\begin{remark}\label{rmk:higher_moments_cN}
Similar bounds can be established for higher moments of the number of particles operator on the fluctuation vector $\langle \xi_t, \mathcal{N}^j \xi_t \rangle $, for any $j\ge 2$. The proof is essentially the same as for \eqref{eq:Gronwall_bound_main}. We do not pursue this generalization here because we do not need it for our main result, and we want to keep the notation as simple as possible. 
\end{remark}
Next, we show how Proposition~\ref{mainprop} implies Theorem~\ref{thm:main_2}. 
\begin{proof}[Proof of Theorem~\ref{thm:main_2}]
	We combine \eqref{eq:bound_particle_number} and Lemma \ref{lem:closeness_dynamics_V} to see that
	\begin{equation}
		\Vert \gamma_{\xi,t} - |\phi_t \rangle \langle \phi_t | - \gamma_t \Vert_{\mathcal{L}^1} \lesssim \exp(ct) \left( T_{\mathrm{c}}(s) \langle \xi_t, \mathcal{N} \xi_t \rangle + \sqrt{N T_{\mathrm{c}}(s) } \langle \xi_t, \mathcal{N} \xi_t \rangle^{1/2} \right) \,.
		\label{eq:final1}
	\end{equation}
	An application of Proposition \ref{mainprop} with $k=42$ proves
	\begin{equation} \label{eq:excitation_number_bound}
		\langle \xi_t, \mathcal{N} \xi_t \rangle \lesssim \exp(c \exp(c \exp(ct)))\,,
	\end{equation} under the assumptions of Theorem \ref{thm:main_2}. Plugging \eqref{eq:excitation_number_bound} into \eqref{eq:final1} and using the bound 
	\begin{equation} \label{eq:Difference_HFB_H_final}
	\left\| \ |\phi_t\rangle\langle\phi_t| + \g_t - \left(|\phi_t^{\mathrm H}\rangle\langle\phi_t^{\mathrm H}| + e^{\mathrm{i}\D t}\g e^{-\mathrm{i}\D t} \right) \right\|_{\cL^1} \lesssim N^{1/2}T_c^{3/4}(s)\exp(ct)\,,
	\end{equation} which follows from \eqref{eq:closeness_dynamics}, we get \eqref{eq:thmbound1pdm}, concluding the proof of the theorem. 
\end{proof}
\begin{remark}
	Remark \ref{rmk:restriction_s}.3 is proven by using the bound \eqref{eq:Hartree_1pdm_bound} instead of \eqref{eq:Difference_HFB_H_final} at the end of the proof of Theorem \ref{thm:main_2}.
\end{remark}
\section{Study of the fluctuation dynamics}\label{sec:fluctuationdynamics}
\label{sec:studyfluctdyn}
The aim of this section is to prove Proposition~\ref{mainprop}. We will apply a Gr\"onwall argument, and therefore start our analysis with the computation of the generator of the fluctuation dynamics $\mathcal{U}^{\mathrm{fluct}}(t,s)$ in \eqref{eq:fluctdyn}. 
\subsection{Generator of the fluctuation dynamics}
\label{sec:generator}
The propagator of the fluctuations satisfies the equation
\begin{equation}
	\mathrm{i} \partial_t \mathcal{U}^{\mathrm{fluct}}(t,s) = \mathcal{G}_{N,t} \mathcal{U}^{\mathrm{fluct}}(t,s)
	\label{eq:genfluctdyn2}
\end{equation} 
with the time-dependent generator
\begin{equation}
	\mathcal{G}_{N,t} = \left( \mathrm{i} \partial_t \mathcal{T}_t^* \right) \mathcal{T}_t + \mathcal{T}_t^*\left( \mathrm{i} \partial_t \cW_t^* \right)\cW_t \mathcal{T}_t+\mathcal{T}_t^* \mathcal{W}_t^* \mathcal{L}_N \mathcal{W}_t \mathcal{T}_t \,,
	\label{eq:genfluctdyn3}
\end{equation}
where we used the shorthand $\cW_t=\cW(\phi_t)$.  In the following proposition we compute $\mathcal{G}_{N,t}$. To keep formulas at a reasonable length, we introduce the notation 
\begin{equation}
	v(\xx, \yy) = \begin{cases}
		v(x-y) & \text{ if } \sigma = \sigma' = \ell \,, \\ -v(x-y) & \text{ if } \sigma = \sigma' = r \,, \\ 0 & \text{ instead}\,, 
	\end{cases}
	\label{eq:v(xx)}
\end{equation}
where $\xx = (\sigma,x)$ and $\yy = (\sigma', y)$. We also recall the notation $\int f(\xx)\de\xx=\sum_{\s=l,r}\int f(\s,x)\de x$ for \\$f:\{\ell,r\}\times\RR^3\to\CC$, which was introduced in Section \ref{sec:fluctdyn} and will be used throught the present Section.
\begin{proposition}\label{prop:Generator}
	Let $v$ be an interaction potential satisfying the same assumptions as in Theorem \ref{thm:main_2}, and let $\phi\in H^3(\RR^3)$, $\gamma\in\cH^{3,1}$. The generator of the fluctuation dynamics $\mathcal{U}^{\mathrm{fluct}}(t,s)$ can be written as
	\begin{equation}
		\cG_{N,t} = \cI_{N,t}^{(1)}+ \cI_{N,t}^{(2)}+ \cI_{N,t}^{(3)} + \cI_{N,t}^{(4)}\,, 
		\label{eq:generatorformula}
	\end{equation}
	where 
	\begin{align}
		\cI_{N,t}^{(1)}
		= &\;\frac{1}{2N} \int v(\xx, \yy)\Big(a^*(U_{t,\xx})a^*(U_{t,\yy})a(U_{t,\yy})a(U_{t,\xx}) \nonumber \\
		&\;+ a^*(U_{t,\xx})a^*(\cc{V_{t,\yy}})a(\cc{V_{t,\yy}})a(U_{t,\xx})+ a^*(U_{t,\yy})a^*(\cc{V_{t,\xx}})a(\cc{V_{t,\xx}})a(U_{t,\yy}) \nonumber \\
		&\;+ a^*(\cc{V_{t,\xx}})a^*(\cc{V_{t,\yy}})a(\cc{V_{t,\yy}})a(\cc{V_{t,\xx}})+2a^*(U_{t,\xx})a^*(\cc{V_{t,\yy}})a(\cc{V_{t,\yy}})a(U_{t,\xx})\Big)\de\xx\de\yy \label{eq:I1}
	\end{align}
	collects the terms that are quartic in $a$ and $a^*$ and commute with $\cN$,
	\begin{align}
		\cI_{N,t}^{(2)} = \;& \frac{1}{2N} \int v(\xx, \yy)\Big(a^*(U_{t,\xx})a^*(U_{t,\yy})a^*(\cc{V_{t,\yy}})a^*(\cc{V_{t,\xx}})+a^*(U_{t,\xx})a^*(U_{t,\yy})a^*(\cc{V_{t,\yy}})a(U_{t,\xx}) \nonumber \\
		&\;+ a^*(U_{t,\xx})a^*(\cc{V_{t,\yy}})a^*(\cc{V_{t,\xx}})a(\cc{V_{t,\yy}}) + a^*(U_{t,\xx})a^*(U_{t,\yy})a^*(\cc{V_{t,\xx}})a(U_{t,\yy}) \nonumber \\ 
		&\; + a^*(U_{t,\yy})a^*(\cc{V_{t,\yy}})a^*(\cc{V_{t,\xx}})a(\cc{V_{t,\xx}}) +\hc\Big)\de\xx\de\yy \label{eq:I2}
	\end{align} 
	contains the quartic terms that do not commute with $\cN$,
	\begin{align}
		\cI_{N,t}^{(3)} = \;&\frac1{N}\int v(\xx, \yy)\Big\{\pphi(\xx)\Big(a^*(U_{t,\xx})a^*(U_{t,\yy})a^*(\cc{V_{t,\yy}}) + a^*(U_{t,\xx})a^*(U_{t,\yy})a(U_{t,\yy}) \nonumber \\
		&\; + a^*(U_{t,\xx})a^*(\cc{V_{t,\yy}})a(\cc{V_{t,\yy}})  + a^*(U_{t,\yy})a^*(\cc{V_{t,\yy}})a(\cc{V_{t,\xx}}) + a^*(\cc{V_{t,\yy}})a(\cc{V_{t,\xx}})a(\cc{V_{t,\yy}}) \nonumber \\
		&\; +  a^*(U_{t,\xx})a(\cc{V_{t,\yy}})a(U_{t,\yy})  +  a^*(U_{t,\yy})a(\cc{V_{t,\xx}})a(U_{t,\yy}) +  a(\cc{V_{t,\xx}})a(\cc{V_{t,\yy}})a(U_{t,\yy}) \Big) + \hc\Big\}\de\xx\de\yy \label{eq:I3}
	\end{align} the cubic terms, and 
	\[
	\cI_{N,t}^{(4)}=  \tr(V_t L V_t^*) + \frac1{2N}\int v(\xx, \yy)|\pphi_t(\xx)|^2|\pphi_t(\yy)|^2\de\xx\de\yy\,,
	\]
	with 
	\begin{equation}
		L = \begin{pmatrix}
			-\Delta & 0 \\ 0 & \Delta 
		\end{pmatrix} \,,
		\label{eq:L}
	\end{equation}
	is a time-dependent constant. 
\end{proposition}
\begin{remark}
	Comparing with the literature (especially \cite[Proposition~3.1]{Fermions}), it may seem surprising that the kinetic energy $\de\G(\D)$ of the excitations does not appear in \eqref{eq:generatorformula}. This is a consequence of a different definition of the symplectomorphism $T_t$ in \eqref{eq:symplectoTt}, which of course leads to a different Bogoliubov transformation $\cT_t$. In particular, from \eqref{eq:def_Ut_Vt}, we observe that the HFB propagator $\cU_t$ only acts from the right on the initial data \eqref{eq:def_U0_V0}. Effectively, this amounts to considering the fluctuation dynamics in the interaction picture.
\end{remark}
\begin{proof}
	We start our proof with the computation of $(\mathrm{i} \dpr_t \mathcal{T}_t^*) \mathcal{T}_t$. Our analysis for this term follows that in the proof of \cite[Proposition 3.1]{Fermions}. We observe that 
	\[
	\mathcal{T}_t^*A(f,g) \mathcal{T}_t=A(T_t(f,g)) \qquad \]
	holds for $f,g\in\fh\oplus\fh$, with the notation $A(f,g)=a(f)+a^*(\cc g)$. If we differentiate both sides w.r.t. time we find
	\[
	(\mathrm{i} \dpr_t \mathcal{T}_t^*) \mathcal{T}_t A(T_t(f,g))+A(T_t(f,g))\mathcal{T}_t^*(\mathrm{i} \dpr_t \mathcal{T}_t) = -A((\mathrm{i}\dpr_t T_t)(f,g)) \,.
	\] 
	We use $0 = \mathrm{i} \dpr_t( \mathcal{T}_t^* \mathcal{T}_t)=(\mathrm{i} \dpr_t \mathcal{T}_t^*) \mathcal{T}_t+\mathcal{T}_t^*(\mathrm{i}\dpr_t \mathcal{T}_t)$ to arrive at
	\begin{equation}\label{eq:idtT*T_commutator}
		\left[(\mathrm{i}\dpr_t\mathcal{T}_t^*)\mathcal{T}_t, A(T_t(f,g))\right]= 	-A((\mathrm{i}\dpr_tT_t)(f,g)) \,.
	\end{equation} Since $T_t$ is a symplectomorphism and \eqref{eq:idtT*T_commutator} holds for $f,g \in \fh \oplus \fh$, we conclude that $(\mathrm{i} \dpr_t\mathcal{T}_t^*)\mathcal{T}_t$ is quadratic in creation and annihilation operators. That is, we can write it as
	\begin{equation}\label{eq:idtTt*Tt}
		(\mathrm{i}\dpr_t\mathcal{T}_t^*)\mathcal{T}_t=\int C_t(\xx,\yy)a^*_\xx a_\yy\de \xx\de\yy+\frac{1}{2}\left(\int D_t(\xx,\yy)a_\xx^*a_\yy^*\de\xx\de\yy+\hc\right) \,,
	\end{equation} 
	where $C_t(\xx,\yy)$ and $D_t(\xx,\yy)$ denote the kernels of two operators $C_t$ and $D_t$ on $\fh\oplus\fh$. The operator $C_t$ is self-adjoint, i.e. 
	\[
	C_t(\yy,\xx) = \cc{C_t(\xx,\yy)}\,,
	\] and we assume 
	\[
	D(\xx,\yy) = D(\yy,\xx)\,,
	\] which will be justified a-posteriori. We use \eqref{eq:idtT*T_commutator} to determine the kernels. A straightforward computation shows that
	\[
	\begin{split}
		\left[\int C_t(\xx,\yy)a^*_\xx a_\yy\de \xx\de\yy, A(T_t(f,0))\right] =&\; -a(C_tU_tf)+a^*(C_t\cc{V_t} \ \cc{f})\,,\\
		\left[\frac12\int D_t(\xx,\yy)a^*_\xx a^*_\yy\de \xx\de\yy, A(T_t(f,0))\right] =&\; -a^*(D_t\cc{U_t} \ \cc{f})\,,\\
		\left[\frac12\int \cc{D_t(\xx,\yy)}a_\xx a_\yy\de \xx\de\yy, A(T_t(f,0))\right] =&\; a(D_tV_tf)\,.	
	\end{split}	
	\] 
	We insert this into \eqref{eq:idtT*T_commutator} and obtain the following system of equations for $C_t, D_t$:
	\[
	\begin{pmatrix}
		C_t & D_t
	\end{pmatrix} \begin{pmatrix}
		- U_t &  \cc{V_t} \\
		V_t & - \cc{U_t}
	\end{pmatrix} = \begin{pmatrix}
		- \mathrm{i}\dpr_tU_t & \mathrm{i}\dpr_t\cc{V_t}
	\end{pmatrix}.
	\] Using \eqref{eq:relations_bog} and \eqref{eq:Cauchy_problem_symplecto_Ut}, we see that
	\begin{equation*} \begin{split}
			\begin{pmatrix} C_t & D_t \end{pmatrix}	=& \begin{pmatrix} i\dpr_tU_t & -i\dpr_t\cc{V_t}	\end{pmatrix}  \begin{pmatrix} U_t^* & V_t^* \\ \cc{V^*_t} & \cc{U^*_t}	\end{pmatrix} \\
			=& - \begin{pmatrix} U_t & \cc{V_t} \end{pmatrix} \begin{pmatrix}	H_t & K_t \\	\cc {K_t} & \cc{H_t} \end{pmatrix}  \begin{pmatrix} U_t^* & V_t^* \\ \cc{V^*_t} & \cc{U^*_t}	\end{pmatrix}
		\end{split} 
	\end{equation*} 
	holds, and hence
	\begin{align}
		C_t =&\; - \left(U_tH_tU_t^*  +U_tK_t\cc{V^*_t} + \cc{V_t} \ \cc{K_t} \ U_t^*+\cc{V_t} \ \cc{H_t} \ \cc{V_t^*}\right) \,, \nonumber \\
		D_t = &\; - \left(U_tH_tV^*_t+U_tK_t\cc{U_t^*}+\cc{V_t} \ \cc{K_t}V_t^*+\cc{V_t} \ \cc{H_t} \ \cc{U_t^*}\right) \,. \label{eq:CtDt}
	\end{align}
	Observe that $D_t$ is indeed symmetric, that is, $\overline{D_t^*} = D_t$. 
	
	Next, we compute $\mathcal{T}_t^*(\mathrm{i} \dpr_t\cW_t^*)\cW_t\mathcal{T}_t$. 
	A short computation shows
	\[
	(\mathrm{i}\dpr_t\cW_t^*)\cW_t = -\left(a^*(\mathrm{i} \dpr_t\pphi_t)+a(\mathrm{i} \dpr_t\pphi_t)\right) \,,
	\] 
	which implies
	\begin{equation}
		\mathcal{T}_t^*(\mathrm{i}\dpr_t\cW_t^*)\cW_t\mathcal{T}_t = - \left(a^*(U_t(\mathrm{i}\dpr_t\pphi_t))+a^*(\cc{V_t}(-\mathrm{i}\dpr_t\cc{\pphi_t}))+\hc\right)\,.
	\end{equation} 
		We write $\cL_N = \cL_0+\cV_N$ with
		\[
		\begin{split}
			\cL_0 = &\; \de\G_\ell(-\D) - \de\G_r(-\D) = \de\G(L)\,,\\
			\cV_N = &\; \frac{1}{2N}\int v(x-y)a^*_{\ell,x}a^*_{\ell,y}a_{\ell,y}a_{\ell,x}\de x\de y - \frac1{2N}\int v(x-y)a^*_{r,x}a^*_{r,y}a_{r,y}a_{r,x}\de x\de y \\
			= & \; \frac1{2N}\int v(\xx, \yy)a^*_\xx a^*_\yy a_\yy a_\xx \de \xx\de \yy\,,
		\end{split}
		\] 
		and $L$ as defined in \eqref{eq:L}. A straightforward calculation shows that
		\[
		\cW_t^*\cL_0\cW_t = \cL_0 + a(L \pphi_t) + a^*(L \pphi_t)\,.
		\] 
		Using \eqref{eq:Tt_action} we can compute 
		\[
		\begin{split}
			\mathcal{T}_t^*\cL_0\mathcal{T}_t=&\; \de\G(U_t L U_t^*+\cc{V_t} L \cc{V_t^*}) + \tr(V_t L V_t^*)	\\
			&\;+\int a_\xx^*a_\yy^*(U_t L V_t^*)(\xx,\yy) \de \xx \de \yy + \hc\,,
		\end{split}
		\] 
		as well as 
		\[
		\mathcal{T}_t^*[a(L \pphi_t) + a^*(L \pphi_t) ]\mathcal{T}_t= a(U_t L \pphi_t+\cc{V_t} L \cc{\pphi_t}) + a^*(U_t L \pphi_t+\cc{V_t} L \cc{\pphi_t})\,.
		\] 
		We conclude that 
		\begin{align}
			\mathcal{T}_t^*\cW_t^* \cL_0\cW_t\mathcal{T}_t = &\; a(U_t L \pphi_t+\cc{V_t} L \cc{\pphi_t}) + a^*(U_t L \pphi_t+\cc{V_t} L \cc{\pphi_t}) + \tr(V_t L V_t^*) \nonumber \\
			&\; + \de\G(U_tLU_t^*+\cc{V_t}L\cc{V_t^*}) +\left(\int a_\xx^*a_\yy^*(U_tLV_t^*)(\xx,\yy) \de \xx \de \yy + \hc\right) \,. \label{eq:kineticterm}
		\end{align}
		
		Let us continue with the computation of $\mathcal{T}_t^*\cW_t^*\cV_N\cW_t\mathcal{T}_t$. Using \eqref{eq:Weyl_shift_double}, we obtain
		\[
		\begin{split}
			\cW_t^*a_\xx^*a_\yy^*a_\yy a_\xx\cW_t= & \; a_\xx^*a_\yy^*a_\yy a_\xx +\left(\pphi_t(\xx) a_\xx^*a_\yy^*a_\yy + \cc{\pphi_t(\xx)} a_\yy^*a_\yy a_\xx +\xy\right) \\
			& \; +\left(|\pphi_t(\yy)|^2a_\xx^*a_\xx+\pphi_t(\xx)\cc{\pphi_t(\yy)}a_\xx^*a_\yy +\xy\right) + \left(\pphi_t(\xx)\pphi_t(\yy)a_\xx^*a^*_\yy +\hc\right)\\
			&\; + \left(a_\xx^*\pphi_t(\xx)|\pphi_t(\yy)|^2 + a_\xx \cc{\pphi_t(\xx)}|\pphi_t(\yy)|^2+\xy\right) + |\pphi_t(\xx)|^2|\pphi_t(\yy)|^2 \,,
		\end{split}
		\] where $\xy$ is a shorthand notation for the expression preceding it with  $\xx$ and $\yy$ exchanged. We conjugate both sides with $\mathcal{T}_t$, apply \eqref{eq:Tt_action}, bring all terms to normal order and integrate the result against the potential $\frac1{2N}v(\xx, \yy)$. We omit the details of this lenghty but straightforward computation, which leads us to
		\[
		\mathcal{T}_t^*\cW_t^*\cV_N\cW_t\mathcal{T}_t= \sum_{i=1}^3\cI_{N,t}^{(i)} + \cQ + \mathcal{L} + C\,,
		\] where the operators $\cI_{N,t}^{(i)}$ are defined below \eqref{eq:generatorformula}. Moreover,
		\[ 
		\begin{split}
			\cQ = &\; \frac1{2N}\int v(\xx, \yy)\Bigg\{\pphi_t(\xx)\pphi_t(\yy)\Big[a^*(\Ux)a^*(\Uy)+a^*(\Ux)a(\Vy) \\
			&\qquad\qquad \qquad\qquad \qquad  + a^*(\Uy)a(\Vx)+a(\Vx)a(\Vy)\Big] +\hc \Bigg\}\de\xx\de\yy\\
			&\; + \frac1{N}\int v(\xx, \yy)|\pphi_t(\yy)|^2\Big[a^*(\Ux)a(\Ux)+a^*(\Ux)a^*(\Vx) + a(\Vx)a(\Ux)+a^*(\Vx)a(\Vx)\Big]\de\xx\de\yy\\
			&\; + \frac1{N}\int v(\xx, \yy) \pphi_t(\xx)\cc{\pphi_t(\yy)}\Big[a^*(\Ux)a(\Uy)+a^*(\Ux)a^*(\Vy) + a(\Vx)a(\Uy)+a^*(\Vy)a(\Vx)\Big]\de\xx\de\yy\\
			&\; + \frac1{2N}\int v(\xx, \yy)\Bigg\{\left(\Uy,\Vx\right)\Big[a^*(\Uy)a(\Vx) + a^*(\Ux)a^*(\Uy)\\
			&\qquad\qquad \qquad\qquad \qquad  + a(\Vx)a(\Vy)+ a^*(\Ux)a(\Vy)\Big]+\hc\Bigg\}\de\xx\de\yy\\
			&\; + \frac1{N}\int v(\xx, \yy)\left(\Vy,\Vy\right)\Big[a^*(\Ux)a(\Ux)+a^*(\Vx)a(\Vx) + a^*(\Ux)a^*(\Vx) + a(\Vx)a(\Ux)\Big]\de\xx\de\yy\\
			&\; + \frac1{N}\int v(\xx, \yy)\left(\Vy,\Vx\right)\Big[a^*(\Ux)a^*(\Vy)+ a(\Vx)a(\Uy) + a^*(\Vy)a(\Vx) + a^*(\Ux)a(\Uy)\Big]\de\xx\de\yy
		\end{split}
		\] 
		contains the quadratic terms, 
		\begin{align}
			\cL =  \frac1N\int &v(\xx, \yy)\Bigg\{\pphi_t(\xx)\Big[|\pphi_t(\yy)|^2+\left(\Vy,\Vy\right)\Big]\Big(a^*(\Ux)+a(\Vx)\Big) \nonumber \\
			&\; + \left(\Vx,\Uy\right)\Big(a(\Uy)+a^*(\Vy)\Big) + \left(\Vx,\Vy\right)\Big(a^*(\Uy)+a(\Vy)\Big)\Big] +\hc\Bigg\}\de\xx\de\yy 
			\label{eq:LL}
		\end{align}
		the linear terms, and
		\[
		C = \frac1{2N}\int v(\xx, \yy)|\pphi_t(\xx)|^2|\pphi_t(\yy)|^2\de\xx\de\yy
		\] is a time-dependent constant. Here we used the notation $(\cdot,\cdot)$ for the inner product in $(L^2(\RRR^3)\oplus L^2(\RRR^3),\de\zz)$, that is, 
		\begin{equation} \label{eq:innerproduct}
			(f,g)=\int \cc{f(\zz)}g(\zz)\de\zz\,.
		\end{equation}
		Eq.~\eqref{eq:diagallt} implies
		\begin{equation}\label{eq:scalar_prod_UV}
			(\Vy,\Vx)=V_t^*V_t(\xx,\yy)\eqqcolon\widetilde{\g}_t(\xx,\yy) \quad \text{ and } \quad (\Uy,\Vx)\eqqcolon\widetilde{\a}_t(\xx,\yy)\,,
		\end{equation} 
		where, $\widetilde{\g}_t(\xx,\yy)$ and $\widetilde{\alpha}_t(\xx,\yy)$ denote the kernels of the full 1-pdm $\widetilde{\g}_t \in \mathcal{B}(\fh \oplus \fh)$ and of the full pairing function $\widetilde{\alpha}_t \in \mathcal{B}(\fh \oplus \fh)$ of the vector state $\Psi_t = \mathcal{T}_t\O$ on $\mathscr{F}(\fh \oplus \fh)$, respectively. It should be highlighted that $\widetilde{\g}_t((x,\ell),(y,\ell)) = \gamma_t(x,y)$ and $\widetilde{\alpha}_t((x,\ell),(y,\ell)) = \alpha_t(x,y)$. We use $\cc{\widetilde{\g}_t(\xx,\yy)}=\widetilde{\g}_t(\yy,\xx)$ and $\widetilde{\a}_t(\xx,\yy)=\widetilde{\a}_t(\yy,\xx)$ to write $\mathcal{L}$ in \eqref{eq:LL} as
		\[
		\begin{split}
			\cL = & \; a^*\left(U_t\left(v\ast(|\pphi_t|^2+\r_{\g_t})\pphi_t\right)\right)+a\left(\cc{V_t\left(v\ast(|\pphi_t|^2+\r_{\g_t})\right)\pphi_t}\right) \\
			&\; +a^*\left(U_t\left((v\sharp\a_t)\cc{\pphi_t}\right)\right)+a\left(\cc{V_t}\left(\cc{(v\sharp\a_t)}\pphi_t\right)\right) \\
			&\; + a^*\left(U_t\left((v\sharp\g_t)\pphi_t\right)\right) + a\left(\cc{V_t\left((v\sharp\g_t)\pphi_t\right)}\right) + \hc \\
			= &\; a^*\Bigg(U_t\left( \left( v\ast|\pphi|^2+v\ast \r_{\g_t} + v\sharp\g_t \right)\pphi_t + (v\sharp\a_t)\cc{\pphi_t} \right) \\
			&\; + \cc{V_t}\left( \cc{ \left( v\ast|\pphi|^2+v\ast \r_{\g_t} + v\sharp\g_t \right)\pphi_t + (v\sharp\a_t)\cc{\pphi_t}} \right)\Bigg) +\hc\,.
		\end{split}
		\]
		
		In the final step we use the HFB equations to cancel all terms that are linear or quadratic in $a$ and $a^*$. We start by collecting all linear terms. Using the HFB equations in \eqref{eq:HFB}, we find
		\begin{align*}
			&\cL+ a(U_tL\pphi_t+\cc{V_t}L\cc{\pphi_t}) + a^*(U_tL\pphi_t+\cc{V_t}L\cc{\pphi_t})+ 	\mathcal{T}_t^*(\mathrm{i}\dpr_t\cW_t^*)\cW_t\mathcal{T}_t \\
			& \hspace{3cm} =\Bigg\{ a^*\Bigg(U_t\left( \left( L + v\ast|\pphi|^2+v\ast \r_{\g_t} + v\sharp\g_t \right)\pphi_t + (v\sharp\a_t)\cc{\pphi_t} \right) \\
			&\hspace{3.5cm} + \cc{V_t}\left( \cc{  \left(L + v\ast|\pphi|^2+v\ast \r_{\g_t} + v\sharp\g_t \right)\pphi_t + (v\sharp\a_t)\cc{\pphi_t}} \right)\Bigg) \\
			&\hspace{3.5cm} -a^*\left(U_t(\mathrm{i}\dt\pphi_t)+\cc{V_t}(\cc{\mathrm{i}\dt\pphi_t})\right) \Bigg\} +\hc = 0\,.
		\end{align*} 
		Similarly, we write 
		\begin{align*}
			\cQ = & \; \frac{1}{2}\left(\int a^*_\xx a^*_\yy \left( U_t (v\sharp(|\pphi_t\rangle\langle\cc{\pphi_t}|+\a_t))\cc{U_t^*} + \cc{V_t} (\cc{v\sharp(|\pphi_t\rangle\langle\cc{\pphi_t}|+\a_t})V_t^* \right)(\xx,\yy)\de\xx\de\yy +\hc \right) \\
			& \;+ \int a^*_\xx a_\yy \left( U_t (v\sharp(|\pphi_t\rangle\langle\cc{\pphi_t}|+\a_t)\cc{V_t^*} + \cc{V_t} (\cc{v\sharp(|\pphi_t\rangle\langle\cc{\pphi_t}|+\a_t})U_t^* \right)(\xx,\yy) \de\xx\de\yy\\
			&\; + \left(\int a^*_\xx a^*_\yy \left( U_t (v\ast(|\pphi_t|^2+\r_{\g_t})) V_t^*  \right)(\xx,\yy) \de\xx\de\yy +\hc\right)\\
			& \;+ \int a^*_\xx a_\yy \left( U_t (v\ast(|\pphi_t|^2+\r_{\g_t})) U_t^* + \cc{V_t} (v\ast(|\pphi_t|^2+\r_{\g_t})) \cc{V_t^*}  \right)(\xx,\yy) \de\xx\de\yy \\
			&\; + \left(\int a^*_\xx a^*_\yy \left( U_t (v\sharp(|\pphi_t\rangle\langle\pphi_t|+\g_t))V_t^* \right)(\xx,\yy)\de\xx\de\yy +\hc\right)  \\
			& \;+ \int a^*_\xx a_\yy \left( U_t (v\sharp(|\pphi_t\rangle\langle\pphi_t|+\g_t))U_t^* + \cc{V_t} (\cc{v\sharp(|\pphi_t\rangle\langle\pphi_t|+\g_t}))\cc{V_t^*} \right)(\xx,\yy) \de\xx\de\yy \,. \\
		\end{align*}
		When we combine these terms with the quadratic terms coming from $\cW_t^*\mathcal{T}_t^*\cL_0\mathcal{T}_t\cW_t$ we get
		\begin{align*}
			&\de\G(U_tLU_t^*+\cc{V_t}L\cc{V_t^*}) +\left(\int a_\xx^*a_\yy^*(U_tLV_t^*)(\xx,\yy)\de\xx\de\yy + \hc\right) +\cQ \\
			&\hspace{3cm} = \int a_\xx^*a_\yy \left(U_t H_t U_t^* + U_tK_t\cc{V_t^*} + \cc{V_t} \ \cc{K_t}U_t^* + \cc{V_t} \ \cc{H_t} \ \cc{V_t^*} \right) (\xx,\yy)\de\xx\de\yy\\
			& \hspace{3.5cm} +\frac{1}{2}\left(\int a_\xx^*a_\yy^* \left(2U_tH_tV_t^* + U_tK_t\cc{U_t^*} + \cc{V_t} \ \cc{K_t}V_t^*\right)(\xx,\yy)\de\xx\de\yy+\hc\right) \\
			& \hspace{3cm} = -(\mathrm{i}\dt \mathcal{T}_t^*)\mathcal{T}_t\,.
		\end{align*}
		The last equality follows from \eqref{eq:idtTt*Tt}, \eqref{eq:CtDt}, and 
		\begin{equation}
			\int a_\xx^*a_\yy^* \left(\cc{V_tH_tU_t^*}\right)(\xx,\yy)\de\xx\de\yy = \int a_\xx^*a_\yy^* \left(U_tH_tV_t^*\right)(\yy,\xx)\de\xx\de\yy=\int a_\xx^*a_\yy^* \left(U_tH_tV_t^*\right)(\xx,\yy)\de\xx\de\yy\,.
			\label{eq:AAA1}
		\end{equation}
		Eq.~\eqref{eq:AAA1} is a consequence of the canonical commutation relations and $H_t^*=H_t$. Thus all linear and quadratic terms in the generator cancel. This concludes the proof of Proposition~\ref{prop:Generator}. 
	\end{proof}
	\subsection{The truncated fluctuation dynamics} \label{sec4.1}
	As explained in Section~\ref{sec:proofstrategy}, we will apply a Gr\"onwall argument to control the expected number of particles in the fluctuation vector $\xi_t$. The main problem to overcome in this approach is that the operator norms of $U_t$ and $V_t$ in the generator $\cG_{N,t}$ grow with $N$. To solve this problem, we introduce an auxiliary fluctuation dynamics $\cU_\eta^{\textrm{fluct}}$, with a time-dependent generator $\cG_{N,t}^{(\eta)}$  defined similarly as $\cG_{N,t}$, but with a cutoff on the number of particles. We will then prove a bound on the growth of the expectation and of the moments of the number of particles operator with respect to $\cU_\eta^{\textrm{fluct}}$; later we will compare with the true fluctuation dynamics $\cU^{\textrm{fluct}}$. A similar strategy was used in \cite{RodSchlein2009} (at $T=0$) and in \cite{Fermions} (for fermionic systems).
	
	We choose $\eta \in (0,1)$, recall the definition of $\mathcal{N}$ in \eqref{eq:numberOp}, and define 
	\[
	\cal G^{(\eta)}_{N,t}\deq \sum_{i=1}^{4}\cI_{N,t}^{(i,\eta)},
	\] where
	\[
	\begin{split}
		\cI_{N,t}^{(1,\eta)}
		= &\;\frac{1}{2N} \int v(\xx, \yy)\Big(a^*(U_{t,\xx})a^*(U_{t,\yy})a(U_{t,\yy})a(U_{t,\xx})\\
		&\;+ a^*(U_{t,\xx})a^*(\cc{V_{t,\yy}})a(\cc{V_{t,\yy}})a(U_{t,\xx})+ a^*(U_{t,\yy})a^*(\cc{V_{t,\xx}})a(\cc{V_{t,\xx}})a(U_{t,\yy}) \\
		&\;+ a^*(\cc{V_{t,\xx}})a^*(\cc{V_{t,\yy}})a(\cc{V_{t,\yy}})a(\cc{V_{t,\xx}})+2a^*(U_{t,\xx})a^*(\cc{V_{t,\yy}})a(\cc{V_{t,\yy}})a(U_{t,\xx})\Big)\mathds{1}(\cN \le N^{\eta})\de\xx\de\yy \,,\\
		\cI_{N,t}^{(2,\eta)} = \;& \frac{1}{2N} \int v(\xx, \yy)\Bigg\{\Big(a^*(U_{t,\xx})a^*(U_{t,\yy})a^*(\cc{V_{t,\yy}})a^*(\cc{V_{t,\xx}})+a^*(U_{t,\xx})a^*(U_{t,\yy})a^*(\cc{V_{t,\yy}})a(U_{t,\xx}) \\
		&\;+ a^*(U_{t,\xx})a^*(\cc{V_{t,\yy}})a^*(\cc{V_{t,\xx}})a(\cc{V_{t,\yy}}) + a^*(U_{t,\xx})a^*(U_{t,\yy})a^*(\cc{V_{t,\xx}})a(U_{t,\yy}) \\ 
		&\; + a^*(U_{t,\yy})a^*(\cc{V_{t,\yy}})a^*(\cc{V_{t,\xx}})a(\cc{V_{t,\xx}})\Big)\mathds{1}(\cN \le N^{\eta})+\hc\Bigg\}\de\xx\de\yy\,,\\
		\cI_{N,t}^{(3,\eta)} = \;&\frac1{N}\int v(\xx, \yy)\Bigg\{\pphi(\xx)\Big(a^*(U_{t,\xx})a^*(U_{t,\yy})a^*(\cc{V_{t,\yy}}) + a^*(U_{t,\xx})a^*(U_{t,\yy})a(U_{t,\yy}) \\
		&\; + a^*(U_{t,\xx})a^*(\cc{V_{t,\yy}})a(\cc{V_{t,\yy}})  + a^*(U_{t,\yy})a^*(\cc{V_{t,\yy}})a(\cc{V_{t,\xx}}) + a^*(\cc{V_{t,\yy}})a(\cc{V_{t,\xx}})a(\cc{V_{t,\yy}})  \\
		&\; +  a^*(U_{t,\xx})a(\cc{V_{t,\yy}})a(U_{t,\yy})  +  a^*(U_{t,\yy})a(\cc{V_{t,\xx}})a(U_{t,\yy}) +  a(\cc{V_{t,\xx}})a(\cc{V_{t,\yy}})a(U_{t,\yy}) \Big)\mathds{1}(\cN \le N^{\eta}) + \hc\Bigg\}\de\xx\de\yy\,,
	\end{split}
	\] and $\cI_{N,t}^{(4,\eta)}=\cI_{N,t}^{(4)}$.
	By $\mathds{1}(\cN \le N^{\eta})$ we denote the spectral projection of the number of particles operator onto the linear subspace of $\mathscr{F}(\fh \oplus \fh)$, where $\mathcal{N} \le N^{\eta}$ holds. We also introduce the fluctuation dynamics $\mathcal U_{\eta}^{\mathrm{fluct}}(t;s)$ as the unique solution to the equation
	\begin{equation}\label{eq:equation_u_alpha}
		\mathrm{i} \dt \mathcal U_{\eta}^{\mathrm{fluct}}(t;s) = \mathcal{G}^{(\eta)}_{N,t} \mathcal U_{\eta}^{\mathrm{fluct}}(t;s) \quad \text{ with } \quad \cU_{\eta}^{\mathrm{fluct}}(s;s) = \mathds{1}\,,
	\end{equation}
	which is well-defined because for every $t>0$ the Hamiltonian $\mathcal{G}_{N,t}^{(\eta)}$ is a bounded self-adjoint operator on $\mathscr{F}(\fh \oplus \fh)$. The fluctuation vector related to the above dynamics will be denoted by $\xi_{\eta,t}= \mathcal U_{\eta}^{\mathrm{fluct}}(t;0) \xi$. The goal of this section is to prove the following proposition.
	\begin{proposition}\label{lem:gronwall_N}
		Let the interaction potential $v$ satisfy the same assumptions as in Theorem \ref{thm:main_2}, and let the initial datum $(\phi,\gamma)$ satisfy Assumption \ref{ass:scales} with $0 < s \leq 3/2$. Moreover, let $\eta\in(0,1/7]$. For every $j\in\NN$, there exists a constant $c > 0$ independent of $N \in \mathbb{N}$ and $t>0$ such that
		\begin{equation}\label{eq:Gronwall_final_bound}
			\langle \xi_{\eta,t}, \cN^j \xi_{\eta,t} \rangle \le \exp(c\exp(c  \exp(ct))) \langle\x,\cN^j\x\rangle
		\end{equation}
		for every $N\in\NN$, $t>0$. 
	\end{proposition}
	\begin{proof} We start by noting that
		\begin{equation} \label{eq:N_time_derivative}
			\mathrm{i} \frac{\de }{\de t} \langle \xi_{\eta,t}, \cN^j \xi_{\eta,t} \rangle  = -\sum_{k=0}^{j-1} \langle \xi_{\eta,t}, \cN^k  [\cal G^{(\eta)}_{N,t},\cal N] \cN^{j-k-1} \xi_{\eta,t} \rangle\,.
		\end{equation}  A simple computation shows
		\begin{equation} \label{4.41}
			[\cal G^{(\eta)}_{N,t},\cal N]= \widetilde{\cal I}_{N.t}^{(2,\eta)}+\widetilde{\cal I}_{N,t}^{(3,\eta)} \,,
		\end{equation}
		where
		\begin{align}
			\widetilde{\mathcal{I}}^{(2,\eta)}_{N,t} = \;& - \frac{1}{N} \int v(\xx, \yy)\Bigg\{\Big(2a^*(U_{t,\xx})a^*(U_{t,\yy})a^*(\cc{V_{t,\yy}})a^*(\cc{V_{t,\xx}})+a^*(U_{t,\xx})a^*(U_{t,\yy})a^*(\cc{V_{t,\yy}})a(U_{t,\xx}) \nonumber \\
			&\;+ a^*(U_{t,\xx})a^*(\cc{V_{t,\yy}})a^*(\cc{V_{t,\xx}})a(\cc{V_{t,\yy}}) + a^*(U_{t,\xx})a^*(U_{t,\yy})a^*(\cc{V_{t,\xx}})a(U_{t,\yy}) \nonumber \\ 
			&\; + a^*(U_{t,\yy})a^*(\cc{V_{t,\yy}})a^*(\cc{V_{t,\xx}})a(\cc{V_{t,\xx}})\Big)\mathds{1}(\cN \le N^{\eta})-\hc\Bigg\}\de\xx\de\yy \label{eq:I2tilda} \,,
		\end{align} 
		and 
		\begin{align}
			\widetilde{\cal I}_{N,t}^{(3,\eta)}=\;& - \frac1{N}\int v(\xx, \yy)\Bigg\{\pphi(\xx)\Big(3a^*(U_{t,\xx})a^*(U_{t,\yy})a^*(\cc{V_{t,\yy}}) + a^*(U_{t,\xx})a^*(U_{t,\yy})a(U_{t,\yy}) \nonumber \\
			&\; + a^*(U_{t,\xx})a^*(\cc{V_{t,\yy}})a(\cc{V_{t,\yy}})  + a^*(U_{t,\yy})a^*(\cc{V_{t,\yy}})a(\cc{V_{t,\xx}}) - a^*(\cc{V_{t,\yy}})a(\cc{V_{t,\xx}})a(\cc{V_{t,\yy}}) \label{eq:I3tilda} \\
			&\; -  a^*(U_{t,\xx})a(\cc{V_{t,\yy}})a(U_{t,\yy})  -  a^*(U_{t,\yy})a(\cc{V_{t,\xx}})a(U_{t,\yy})  -3 a(\cc{V_{t,\xx}})a(\cc{V_{t,\yy}})a(U_{t,\yy}) \Big)\mathds{1}(\cN \le N^{\eta}) - \hc\Bigg\}\de\xx\de\yy\,. \nonumber
		\end{align}
		We insert the first term on the r.h.s. of \eqref{eq:I2tilda} into \eqref{eq:N_time_derivative} and estimate it by
		\begin{align}  
			&\sum_{k=0}^{j-1} \frac{2}{N}\int \Big|v(\xx, \yy)\big\langle\xi_{\eta,t},
			\cal N^{k}	a^*(\Ux)a^*(\Uy)a^*(\Vy)a^*(\Vx)\mathds{1}(\cN \le N^{\eta})\cN^{j-k-1}\xi_{\eta,t}\big \rangle\Big| \de\xx\de\yy \nonumber \\
			&\hspace{3cm}=\sum_{k=0}^{j-1}\frac{2}{N}\int \Big|v(\xx, \yy)\big\langle \xi_{\eta,t},
			\cN^{k}(\cal N-4)^{j/2-k-1}	a^*(\Ux)a^*(\Uy) \nonumber \\
			&\hspace{5cm}\times a^*(\Vy)a^*(\Vx) \mathds{1}(\cN \le N^{\eta})\cN^{j/2}\xi_{\eta,t} \big \rangle\Big| \de\xx\de\yy \nonumber \\
			&\hspace{3cm} \le 2j \ \bigg(	\frac{1}{N}\int  |v(\xx, \yy)| \ \big\|a(\Ux)a(\Uy)(\cN-4)^{j/2-k-1}\cN^{k}\xi_{\eta,t}\big\|^2 \de\xx\de\yy\bigg)^{1/2} \nonumber \\
			&\hspace{5cm}\times	\bigg( \frac{1}{N}\int |v(\xx, \yy)| \ \big\|a^*(\Vy)a^*(\Vx)  \mathds{1}(\cN \le N^{\eta})\cN^{j/2}\xi_{\eta,t}\big\|^2 \de\xx\de\yy \bigg)^{1/2}\,. \label{5.45} 
		\end{align}
		The square of the first factor on the r.h.s. of \eqref{5.45} can be estimated by
		\begin{align}  
			&\frac{1}{N}\int |v(\xx, \yy)| \big\langle\xi_{\eta,t},\cN^k(\cN-4)^{j/2-k-1} a^*(\Uy)a^*(\Ux)a(\Ux)a(\Uy)(\cN-4)^{j/2-k-1}\cN^{k}\xi_{\eta,t}\big\rangle \de\xx\de\yy\ \nonumber \\
			&\hspace{3cm}\le \frac{\|v\|_\infty}{N}\int\big\langle\xi_{\eta,t}\cN^k(\cN-4)^{j/2-k-1} a^*(\Uy)\de\G(U_tU_t^*)a(\Uy)(\cN-4)^{j/2-k-1}\cN^{k}\xi_{\eta,t}\big\rangle \de\yy \nonumber \\
			&\hspace{3cm}\le \frac{\|v\|_\infty\|U_t\|^2}{N}\int\big\langle\xi_{\eta,t}\cN^k(\cN-4)^{j/2-k-1}a^*(\Uy)(\cN-3) a(\Uy)(\cN-4)^{j/2-k-1}\cN^{k}\xi_{\eta,t}\big\rangle \de\yy \nonumber \\
			&\hspace{3cm}= \frac{\|v\|_\infty\|U_t\|^2}{N}\int\big\langle\xi_{\eta,t}\cN^k(\cN-4)^{j/2-k-1/2}a^*(\Uy) a(\Uy)(\cN-4)^{j/2-k-1/2}\cN^k\xi_{\eta,t}\big\rangle \de\yy \nonumber \\
			&\hspace{3cm}\le \frac{\|v\|_\infty\|U_t\|^4}{N}\big\langle\xi_{\eta,t}\cN^{j}\xi_{\eta,t}\big\rangle\,. \label{5.46}
		\end{align}
		To treat the second term on r.h.s. of \eqref{5.45} we first note that
		\begin{align} 
			a(&\Vx)a(\Vy)a^*(\Vy)a^*(\Vx) \nonumber \\ 
			= & \; a^*(\Vx)a^*(\Vy)a(\Vy)a(\Vx) + a^*(\Vx)a(\Vx) \widetilde{\g}_t(\yy,\yy) + a^*(\Vy)a(\Vy) \widetilde{\g}_t(\xx,\xx)  \nonumber \\
			& \; + a^*(\Vx)a(\Vy) \widetilde{\g}_t(\yy,\xx) + a^*(\Vy)a(\Vx) \widetilde{\g}_t(\xx,\yy) + \widetilde{\g}_t(\xx,\xx) \widetilde{\g}_t(\yy,\yy) + \widetilde{\g}_t(\xx,\yy)\widetilde{\g}_t(\yy,\xx)\,, \label{eq:commutatorsVVVV}
		\end{align}
		with $\widetilde{\g}_t$ as defined in \eqref{eq:scalar_prod_UV}. We insert the first term on the r.h.s. of \eqref{eq:commutatorsVVVV} back into the second factor on the r.h.s. of \eqref{5.45} and find
		\begin{multline*}
			\frac{1}{N}\int |v(\xx, \yy)| \ \big\|a(\Vy)a(\Vx) \mathds{1}(\cN \le N^{\eta})\cN^{j/2}\xi_{\eta,t}\big\|^2 \de\xx\de\yy \\
			\lesssim \frac{\|v\|_{\infty}\|V_t\|^4}{N}\big\| \mathcal{N}^{j/2+1}\mathds{1}(\cN \le N^{\eta})\xi_{\eta,t}\big\|^2\leq \frac{\| v \|_{\infty}\|V_t\|^4}{N^{1-2\eta}}\big\langle \xi_{\eta,t}, \mathcal{N}^j \xi_{\eta,t} \big\rangle\,.
		\end{multline*}
		To estimate the contributions from the last two terms on the r.h.s.\,of \eqref{eq:commutatorsVVVV}, we note that, by Cauchy-Schwarz, $| \widetilde{\g}_t(\xx,\yy) |^2 \leq \widetilde{\g}_t(\xx,\xx) \widetilde{\g}_t(\yy,\yy)$. Hence we can bound
		\begin{align}
			&\bigg|\frac{1}{N}\int v(\xx, \yy) \left[ \widetilde{\g}_t(\xx,\xx)\widetilde{\g}_t(\yy,\yy)+\widetilde{\g}_t(\xx,\yy)\widetilde{\g}_t(\yy,\xx) \right] \de \xx \de \yy \bigg| \nonumber  \\
			&\hspace{3cm}\leq 4 \Vert v \Vert_1N^{-1} \sup_{x \in \mathbb{R}^3} |\g_t(x,x)| \int  \g_t(x,x) \de x \lesssim \Vert v \Vert_1 \exp(c \exp(c t)) T_{\mathrm{c}}^{3/2}(s)\,, \label{5.47c}
		\end{align} where we also used the bound 
		\begin{equation*}
			\sup_{\xx \in \{l,r\} \times \RR^3 } \int v(\xx,\yy)\de\yy \le \| v \|_1 \,,
		\end{equation*} which follows immediately from the definition \eqref{eq:v(xx)} of $v(\xx,\yy)$. To get to the second line we used the fact that for $\xx=(\s,x)$ we have 
		\[ \widetilde{\g}_t(\xx,\xx) = \begin{cases}
			\g_t(x,x), & \text{if }\sigma=\ell\,,\\
			\cc{\g_t(x,x)}, & \text{if }\sigma=r \,,
		\end{cases}\] which lets us control the diagonal of $\widetilde\g_t$ with that of $\g_t$. 
		The second estimate follows from Lemma~\ref{lem:diluteness_HFB} and $\tr \gamma_t \leq N$. Similar considerations show that the contributions from all remaining terms on the r.h.s. of \eqref{eq:commutatorsVVVV} can be bounded by a constant times
		\begin{equation*}
			N^{-1+\eta} \exp(c t \exp(c t))T_{\mathrm{c}}^{3/2}(s) \ \Vert v \Vert_1 \Vert V_t \Vert^2 \langle \xi_{\eta,t}, \mathcal{N}^j \xi_{\eta,t} \rangle\,.
			\label{5.47d}
		\end{equation*} 
		The above considerations show
		\begin{align} 
			&\frac{1}{N}\int |v(\xx, \yy)| \ \big\|a^*(\Vy)a^*(\Vx)\cal \mathds{1}(\cN \le N^{\eta})\cN^{j/2}\xi_{\eta,t}\big\|^2 \de\xx\de\yy \nonumber \\
			&\hspace{3cm}\lesssim  \left[\frac{\Vert v \Vert_{\infty} \Vert V_t \Vert^4 }{N^{1-2\eta}} + \exp(c t \exp(ct)) T_{\mathrm{c}}^{3/2}(s) \ \Vert v \Vert_1 \left( 1 + N^{-1+\eta} \Vert V_t \Vert^2 \right) \right]\langle \xi_{\eta,t}, \mathcal{N}^j \xi_{\eta,t} \rangle\,, \label{5.47e}
		\end{align}
		which combined with \eqref{5.46} implies
		\begin{align} 
			\eqref{5.45} \lesssim & \left( \Vert v \Vert_1 + \Vert v \Vert_{\infty} \right) \langle \xi_{\eta,t}, \mathcal{N}^j \xi_{\eta,t} \rangle \ \Bigg[ \frac{\|U_t\|^2 \ \|V_t\|^2}{N^{1-\eta}}  
			+\frac{\exp(c \exp(c t)) T_{\mathrm{c}}^{3/4}(s) \ \|U_t\|^2}{N^{1/2}} \left( 1 + N^{-1/2 + \eta/2} \Vert V_t \Vert \right) \Bigg] \nonumber \\
			\lesssim &   \exp(c \exp(c t)) \langle \xi_{\eta,t}, \mathcal{N}^j \xi_{\eta,t} \rangle  \,. \label{guoan}
		\end{align}
		In the second step, we used the definition \eqref{eq:idealgascrittemp} of $T_\mathrm{c}(s)$ and Lemma \ref{lem:closeness_dynamics_V} to see that 
		\[ N^{-1/2}T_\mathrm{c}^{3/4}(s)\| U_t \|^2 \lesssim \exp(ct) N^{-1/2}T_{\mathrm c}^{7/4}(s)\lesssim \exp(ct) N^{-1/2} N^{\frac 7 4 \frac{2s}{6+3s}} \lesssim \exp(ct)
		\] for $0<s\le 3/2$ (this determines the restriction on the range of $s$ that we can treat), and the condition $\eta\le 1/7$, again with Lemma \ref{lem:closeness_dynamics_V}, to show 
		\[
		N^{-1+\eta} \| U_t \|^2 \| V_t \|^2\,, N^{-1/2+\eta/2} \| V_t \| \lesssim \exp(ct) \,.
		\] Next, we consider the second term on the r.h.s. of \eqref{eq:I2tilda}. An application of Cauchy-Schwarz yields
		\begin{multline}   \label{biubiubiu}
			\sum_{k=0}^{j-1}	\frac{1}{N}\int  \Big|v(\xx, \yy)\ \big\langle\xi_{\eta,t},
			\cal N^{k}	a^*(\Ux)a^*(\Uy)a^*(\Vy)a(\Ux)\mathds{1}(\cN \le N^{\eta})\cN^{j-k-1}\xi_{\eta,t}\big \rangle\Big| \de\xx\de\yy \\
			\leq j \ \Bigg(	\frac{1}{N}\int  |v(\xx, \yy)| \ \big\|a(\Ux)a(\Uy)(\cN-4)^{j/2-k-1}\cN^k\xi_{\eta,t}\big\|^2\de\xx\de\yy\Bigg) \\
			\times	\Bigg(\frac{1}{N}\int |v(\xx, \yy)| \ \big\|a^*(\Vy)a(\Ux) \mathds{1}(\cN \le N^{\eta})\cN^{j/2}\xi_{\eta,t}\big\|^2 \de\xx\de\yy \Bigg)^{1/2}\,.
		\end{multline}
		The first factor on the r.h.s.\,can be estimated similarly as the first term on the r.h.s.\,of \eqref{5.45}. The second factor is estimated by
		\begin{align*}
			&\frac{1}{N}\int  |v(\xx, \yy) | \ \big\langle\xi_{\eta,t},\cN^{j/2}\mathds{1}(\cN \le N^{\eta})a^*(\Ux)a(\Vy)a^*(\Vy)a(\Ux)\mathds{1}(\cN \le N^{\eta})\cN^{j/2}\xi_{\eta,t}\big\rangle \de\xx\de\yy \\
			&\hspace{1cm}= \frac{1}{N}\int |v(\xx, \yy)| \ \big\langle\xi_{\eta,t},\cN^{j/2}\mathds{1}(\cN \le N^{\eta})a^*(\Ux)a^*(\Vy)a(\Vy)a(\Ux)\mathds{1}(\cN \le N^{\eta})\cN^{j/2}\xi_{\eta,t}\big\rangle \de\xx\de\yy \\
			&\hspace{1.5cm}+ \frac{1}{N}\int  |v(\xx, \yy)| \ \widetilde{\g}_t(\yy,\yy) \ \big\langle\xi_{\eta,t},\cN^{j/2}\mathds{1}(\cN \le N^{\eta})a^*(\Ux)a(\Ux)\mathds{1}(\cN \le N^{\eta})\cN^{j/2}\xi_{\eta,t}\big\rangle \de\xx\de\yy \\
			&\hspace{1cm}\lesssim \left(\frac{\Vert v \Vert_{\infty} \Vert U_t \Vert^2 \Vert V_t \Vert^2 }{ N^{1-2\eta}}+ \Vert v \Vert_1 \exp(c\exp(ct)) T_{\mathrm{c}}^{3/2}(s) \ \frac{ \Vert U_t \Vert^2}{N^{1-\eta}  } \right)\langle \xi_{\eta,t}, \mathcal{N}^j \xi_{\eta,t}\rangle\,. 
		\end{align*}
		Using Lemma \ref{lem:closeness_dynamics_V} and the assumptions on $s$ and $\eta$ as above we conclude
		\begin{equation} 
			\eqref{biubiubiu} 
			\lesssim \exp(c \exp(c t)) \langle \xi_{\eta,t}, \cal N^j \xi_{\eta,t}\rangle \,. \label{5.47g}  
		\end{equation} The contribution of all remaining terms on the r.h.s. of \eqref{eq:I2tilda} can be estimated in the same way.
		
		Let us now consider the terms on the r.h.s. of \eqref{eq:I3tilda}. The contribution arising from the first of these terms can be bounded, with Cauchy-Schwarz, by
		\begin{align}
			\sum_{k=0}^{j-1} \frac{2}{N} \int &\Big|\pphi(\xx)v(\xx, \yy) \big\langle\xi_{\eta,t},
			\cal N^{k}	a^*(\Ux)a^*(\Uy)a^*(\Vy)\mathds{1}(\cN \le N^{\eta})\cN^{j-k-1}\xi_{\eta,t}\big \rangle\Big| \de\xx\de\yy \nonumber \\
			&\leq j \ \bigg( \frac{1}{N}\int |v(\xx, \yy)| \ \big\|a(\Ux)a(\Uy)(\cN-4)^{j/2-k-1}\cN^k\xi_{\eta,t}\big\|^2 \de\xx\de\yy \bigg ) \nonumber \\
			&\hspace{3cm}\times	\bigg( \frac{1}{N}\int |v(\xx, \yy)| \ |\pphi(\xx)|^2 \big\|a^*(\Vy)\mathds{1}(\cN \le N^{\eta})\cN^{j/2}\xi_{\eta,t}\big\|^2 \de\xx\de\yy \bigg)^{1/2}\,. \label{eq:term_phi_3*}
		\end{align}
		The first factor can again be estimated as the first factor on the r.h.s. of \eqref{5.45}. We estimate the other one by 
		\begin{align*}
			&\frac{1}{N}\int |v(\xx, \yy)| \ |\pphi(\xx)|^2 \big\langle\xi_{\eta,t},\cN^{j/2}\mathds{1}(\cN \le N^{\eta})a(\Vy)a^*(\Vy)\mathds{1}(\cN \le N^{\eta})\cN^{j/2}\xi_{\eta,t}\big\rangle \de\xx\de\yy \\
			&\hspace{3cm}=\frac{1}{N}\int |v(\xx, \yy)| \ |\pphi(\xx)|^2 \big\langle\xi_{\eta,t},\cN^{j/2}\mathds{1}(\cN \le N^{\eta})a^*(\Vy)a(\Vy)\mathds{1}(\cN \le N^{\eta})\cN^{j/2}\xi_{\eta,t}\big\rangle \de\xx\de\yy \\
			&\hspace{3.5cm}+ \frac{1}{N}\int |v(\xx, \yy)| \ |\pphi(\xx)|^2 \ \widetilde{\g}_t(\yy,\yy) \ \big\langle\xi_{\eta,t},\cN^{j}\mathds{1}(\cN \le N^{\eta})\xi_{\eta,t}\big\rangle \de\xx\de\yy \\
			&\hspace{3cm}\lesssim \left[ \|v\|_\infty N^{\eta} \|V_t\|^2 + \|v\|_1 \exp(c\exp(ct))T_{\mathrm{c}}^{3/2}(s) \right] \langle \xi_{\eta,t},\cN^j\xi_{\eta,t}\rangle\,,
		\end{align*} 
		where we used Lemma~\ref{lem:diluteness_HFB} and $\|\phi_t\|_2^2 \leq N$. We conclude with the help of Lemma \ref{lem:closeness_dynamics_V} and $\eta\le 1/7$ that 
		\begin{equation*}
			\eqref{eq:term_phi_3*} \lesssim \exp(c \exp(c t))  \langle \xi_{\eta,t}, \mathcal{N}^j \xi_{\eta,t} \rangle \,.
		\end{equation*} 
		The contribution of the other terms on the r.h.s. of \eqref{eq:I3tilda} can be bounded similarly. 
		
		Collecting all bounds, we find
		\begin{align*}
			\left|\dt \langle\xi_{\eta,t},\cN^j\xi_{\eta,t}\rangle \right| \lesssim \exp(c\exp(ct)) \langle\xi_{\eta,t},\cN^j\xi_{\eta,t}\rangle \,,
		\end{align*}
		and we conclude by an application of Gr\"onwall's lemma.
	\end{proof}
	\subsection{Weak bounds on the original dynamics}
	\label{sec:weakbounds}
	In order to prove that $\cU_\eta^{\textrm{fluct}}$ remains close to the full fluctuation dynamics, we need some rough a-priori bounds on the growth of the number of particles operator when conjugated with  $\cU^{\textrm{fluct}}$. 
	
	\begin{lemma} \label{lem5.1}
		Let $v$ be an interaction potential satisfying the same assumptions as in Theorem \ref{thm:main_2}, and let $\phi\in H^3(\RR^3)$, $\gamma\in\cH^{3,1}$ with $n(\phi,\gamma)=N$. For $j \in \mathbb N$ and $t > s > 0$, we have
		\begin{align} 
				\cal U^{\mathrm{fluct}}(t,s)^* \cal N^{2j}\cal U^{\mathrm{fluct}}(t,s)  \lesssim & \big(\|U_t\|^2+\|V_t\|^2\big)^{2j}\big(\|U_s\|^2+\|V_s\|^2\big)^{2j} \ \cal N^{2j} \nonumber \\
				&+ \left(1+ \big(\|U_t\|^2+\|V_t\|^2\big)^{2j} \right)N^{2j}\,.  \label{eq:weak_bounds_3}
		\end{align}
	\end{lemma}
	
	\begin{proof}
		We claim that the bounds
		\begin{align}
			\mathcal{W}_t^* \cal N^{2j} \mathcal{W}_t \lesssim & \cal N^{2j} + N^{2j}\,,  &	\mathcal{W}_t \cal N^{2j} \mathcal{W}_t^* \lesssim & \cal N^{2j} + N^{2j} \label{eq:weak_bounds_Weyl_2} \\
			\mathcal{T}_t^* \cal N^{2j} \mathcal{T}_t \lesssim & \big(\|U_t\|^2+\|V_t\|^2\big)^{2j} \ \cal N^{2j} + N^{2j} \,, & \mathcal{T}_t \cal N^{2j} \mathcal{T}_t^* \lesssim & \big(\|U_t\|^2+\|V_t\|^2\big)^{2j} \ \cal N^{2j} + N^{2j} \label{eq:weak_bounds_2} 
		\end{align} hold for every $j\in\NN$ and $t>0$. Assuming these bounds for the moment, \eqref{eq:weak_bounds_3} is easily proved by a repeated application of \eqref{eq:weak_bounds_Weyl_2}, \eqref{eq:weak_bounds_2}. Indeed, by the definition of the fluctuation dynamics  we have 
		\[
		\begin{split}
			\cal U^{\mathrm{fluct}}(t,s)^* \cal N^{2j}\cal U^{\mathrm{fluct}}(t,s) = & \cT_s^*\cW_s^* e^{it\cL_N(t-s)}\cW_t\cT_t \cN^{2j} \cT_t^*\cW_t^* e^{-it\cL_N(t-s)}\cW_s\cT_s \\
			\lesssim & \big(\|U_t\|^2+\|V_t\|^2\big)^{2j} \ \cT_s^*\cW_s^* e^{it\cL_N(t-s)}\cW_t\cal N^{2j}\cW_t^* e^{-it\cL_N(t-s)}\cW_s\cT_s + N^{2j}\\
			\lesssim &  \big(\|U_t\|^2+\|V_t\|^2\big)^{2j} \ \cT_s^*\cal N^{2j}\cT_s + \big(1+\|U_t\|^2+\|V_t\|^2\big)^{2j}N^{2j}\\
			\lesssim &  \big(\|U_t\|^2+\|V_t\|^2\big)^{2j}\big(\|U_s\|^2+\|V_s\|^2\big)^{2j} \cal N^{2j} + (1+ \big(\|U_t\|^2+\|V_t\|^2\big)^{2j} )N^{2j}\,,
		\end{split}
		\] where in the second estimate we also used the fact that the Liouvillian $\cL_N$ defined in \eqref{eq:Liouvillian} commutes with the number of particles operator. Thus we just need to prove \eqref{eq:weak_bounds_Weyl_2}, \eqref{eq:weak_bounds_2}. The bounds in \eqref{eq:weak_bounds_Weyl_2} follow by a straightforward adaptation of the arguments in \cite[Lemma 3.6]{RodSchlein2009}. 
		
		The proof of \eqref{eq:weak_bounds_2} is also inspired by the proof of the same lemma, but since the adaptation is less straightforward in this case, we provide the details. We start with the first bound in \eqref{eq:weak_bounds_2}, and we proceed by induction on $j \in \mathbb{N}$. For the base case, we use \eqref{eq:action_time_dep_Bog} to compute $\mathcal{T}_t^* \mathcal{N} \mathcal{T}_t$, bring all terms to normal order, and find
		\begin{equation}
			\mathcal{T}_t^* \cal N \mathcal{T}_t = \int \left[ a^*(U_{t,\xx}) a(U_{t,\xx}) + a^*(\Vx) a(\Vx) + a^*(\Ux) a^*(\Vx) + a(\Vx) a(\Ux) + ( \Vx, \Vx ) \right] \de \xx + 5 \,,
			\label{eq:weakbounds1}
		\end{equation}
		where $(\cdot,\cdot)$ denotes the inner product defined in \eqref{eq:innerproduct}. To obtain a bound for $\cT_t^*\cN^2\cT_t$, we square \eqref{eq:weakbounds1} and apply Cauchy-Schwartz. This yields 
		\begin{align}
			\cT_t^*\cN^2\cT_t\lesssim&1+\int \Big[ a^*(U_{t,\xx}) a(U_{t,\xx}) a^*(U_{t,\yy}) a(U_{t,\yy}) +  a^*(\Vx) a(\Vx) a^*(\Vy) a(\Vy) + ( \Vx, \Vx ) ( \Vy, \Vy ) \nonumber \\
			&\hspace{2cm}+ a^*(U_{t,\xx}) a^*(\Vx) a(\Vy) a(U_{t,\yy}) + a(\Vy) a(U_{t,\yy}) a^*(U_{t,\xx}) a^*(\Vx) \Big] \de \xx \de \yy \,. \label{eq:weakbounds2}
		\end{align}
		The first two terms in the integral are bounded by $(\Vert U_t \Vert^4 + \Vert V_t \Vert^4 ) \mathcal{N}^2$, and we use
		\begin{equation}
			\int (\Vx, \Vx) \de \xx = \int_{\mathbb{R}^3} \left[ \widetilde{\gamma}_t((x,\ell),(x,\ell)) + \widetilde{\gamma}_t((x,r),(x,r)) \right] \de x = 2 \int_{\mathbb{R}^3} \gamma_t(x,x) \de x \leq 2 N
			\label{eq:weakbounds3}
		\end{equation} 
		to see that the third term in the integral is bounded by $4 N^2$. To bound the first term in the second line of \eqref{eq:weakbounds2} we write it as
		\begin{equation*}
			\int a^*(U_{t,\xx}) a(\Vy) a^*(\Vx) a(U_{t,\yy}) \de \xx \de \yy + \int a^*(U_{t,\xx}) a(U_{t,\yy}) (\Vy,\Vx) \de \xx \de \yy \,.
		\end{equation*}
		Using Cauchy-Schwartz and \eqref{eq:weakbounds3}, we check that the first term is bounded by $\Vert U_t \Vert^2 \Vert V_t \Vert^2 \mathcal{N}^2 + \Vert U_t \Vert^2 N$. The second term equals the second quantization of the operator $U_t V_t^* V_t U_t^*$. Accordingly, it is bounded by $\Vert U_t \Vert^2 \Vert V_t \Vert^2 \mathcal{N}$. To obtain a bound for the last term on the r.h.s. of \eqref{eq:weakbounds2}, we write it as
		\begin{equation}
			\int \Big[ a(\Vx) a^*(U_{t,\yy}) a(U_{t,\yy}) a^*(\Vx) + a^*(\Vx) a(\Vy) (U_{t,\yy},U_{t,\xx}) + (\Vy,\Vx) \ (U_{t,\yy},U_{t,\xx}) \Big] \de \xx \de \yy \,.
			\label{eq:weakbounds4}
		\end{equation}
		An application of Cauchy-Schwartz shows that the first term is bounded by $2 \Vert U_t \Vert^2 \Vert V_t \Vert^2 \mathcal{N}^2 + \Vert V_t \Vert^2 N$. The second operator is the second quantization of $V_t \cc{U_t^*U}_t \overline{V_t^*}$, and is therefore also bounded by $\Vert U_t \Vert^2 \Vert V_t \Vert^2 \mathcal{N}$. A computation that uses \eqref{eq:relations_bog} and \eqref{eq:scalar_prod_UV} shows that the last term in \eqref{eq:weakbounds3} can be written as 
		\begin{equation*}
			\tr[ V_t^* V_t U_t^* U_t ] = \tr[ \widetilde{\gamma}_t (1+\widetilde{\gamma}_t) ] \leq 2 N + \left( \tr \widetilde{\g}_t \right)^2 \leq 2 N + 4 N^2. 
		\end{equation*}
		To obtain the second bound we used $\Vert \cdot \Vert_{\mathcal{L}^2} \leq \Vert \cdot \Vert_{\mathcal{L}^1}$. Collecting the above estimates we find
		\begin{equation}
			\mathcal{T}_t^* \mathcal{N}^2 \mathcal{T}_t \lesssim \left(\Vert U_t \Vert^2 + \Vert V_t \Vert^2 \right)^2 \mathcal{N}^2 + N^2\,.
			\label{eq:weakbounds5} 
		\end{equation}
		It remains to consider the induction step. Let us denote $\mathrm{ad}_A(B) = [B,A]$ and $X_t = \mathcal{T}_t^* \mathcal{N} \mathcal{T}_t$. From \eqref{eq:weakbounds1} we have  
		\begin{equation*}
			\mathrm{ad}^n_{\mathcal{N}}(X_t) = (-2)^n \int a^*(U_{t,\xx}) a^*(\Vx) \de \xx + 2^n \int a(\Vx) a(U_{t,\xx}) \de \xx \,.
		\end{equation*} Therefore, for any $n\in\NN$ we can bound
		\begin{equation}
			| \mathrm{ad}^n_{\mathcal{N}}(X_t) |^2 \lesssim \left(\Vert U_t \Vert^4 + \Vert U_t \Vert^4 \right) \mathcal{N}^2 + N^2\,.
			\label{eq:weakbounds6}
		\end{equation} 		
		From the induction assumption, we find 
		\begin{align}
			\mathcal{T}_t^* \mathcal{N}^{2j+2} \mathcal{T}_t = X_t \mathcal{T}_t^* \mathcal{N}^{2j} \mathcal{T}_t X_t \lesssim& X_t \left( C_t^{2j} \ \cal N^{2j} + N^{2j} \right) X_t \lesssim C_t^{2j} X_t \mathcal{N}^{2j} X_t + C_t^2 \mathcal{N}^2 N^{2j} + N^{2j+2} \nonumber \\
			&\lesssim C_t^{2j} X_t \mathcal{N}^{2j} X_t + C_t^{2j+2} \mathcal{N}^{2j+2} + N^{2j+2} \,,
			\label{eq:weakbounds7}
		\end{align} where we defined $C_t \deq\|U_t\|^2+\|V_t\|^2$.
		With Cauchy-Schwartz, \eqref{eq:weakbounds5}, \eqref{eq:weakbounds6} and the commutator expansion
		\begin{equation*}
			[A^n , B ] = \sum_{j=0}^{n-1} \binom{n}{m} A^m \mathrm{ad}_A^{n-m}(B)\,,
		\end{equation*}
		we obtain
		\begin{align}
			X_t \mathcal{N}^{2j} X_t &\leq 2 \mathcal{N}^j X_t^2 \mathcal{N}^j + 2 \left| [\mathcal{N}^j,X_t] \right|^2
			\leq 2 \mathcal{N}^j X_t^2 \mathcal{N}^j + 4^{j+1} \sum_{n=0}^{j-1} \mathcal{N}^n \left| \mathrm{ad}^{j-n}_{\mathcal{N}}(X_t) \right|^2 \mathcal{N}^n \nonumber \\
			&\lesssim C_t^2 \mathcal{N}^{2j+2} + N^2 \mathcal{N}^{2j}  \,. \label{eq:weakbounds8}
		\end{align}
		Inserting \eqref{eq:weakbounds8} into \eqref{eq:weakbounds7} we get 
		\[ \cT_t^* \cN^{2j+2} \cT_t \leq C_t^{2j+2} \cN^{2j+2} + \cN^{2j+2}\,, \]
		which proves the induction step.
		
		Finally, the second bound in \eqref{eq:weak_bounds_2} follows from the first bound and the observation that $\cT_t^*$ is also a Bogoliubov tranformation, implementing the symplectomorphism 
		\[
		T_t^{-1} = \cS T_t^*\cS = \begin{pmatrix}
			U_t^* & -V_t^* \\
			-\cc{V_t^*} & \cc{U_t^*}
		\end{pmatrix}\,.
		\]
		This concludes the proof of the lemma. 
	\end{proof}
	\subsection{Comparison of the two fluctuation dynamics}
	\label{sec:comparisonfluctuationdynamics}
	In this section we compare the fluctuation dynamics $\cal U^{\mathrm{fluct}}$ and $\cal U^{\mathrm{fluct}}_{\eta}$. The goal is to prove the following lemma. In the whole section we use the shorthand notations $C_t \deq\|U_t\|^2+\|V_t\|^2$ and $T_{\mathrm c}\deq T_{\mathrm c}(s)$. 
	\begin{lemma} \label{lem5.3}
	Let the assumptions of Lemma \ref{lem5.1} be satisfied, and let $k\in\NN$, $\eta\in(0,1/7]$. There exists a constant $c>0$ independent of $t, N$ such that
	\begin{align} \nonumber
		\big|	\big\langle \xt, & \cal N\big(\xt-\xat\big) \big\rangle \big| + \big|	\big\langle \xat, \cal N\big(\xt-\xat\big) \big\rangle \big| \lesssim N^{-1/2 - \eta k/2} \exp(c \exp(c\exp(ct))) \\ \label{5.69}
		&\qquad\times \left( T_{\mathrm{c}}^{7} \langle \xi, \mathcal{N}^{4} \xi \rangle^{1/2} + T_{\mathrm{c}}^{5} N^{2} \right) \left( T_{\mathrm{c}} N^{-1/2} + T_{\mathrm{c}}^{1/2} N^{-\eta/2} + T_{\mathrm{c}}^{3/4} N^{-\eta} \right)  \big\langle \x, \cN^{k+2}\x \big\rangle^{1/2}
	\end{align}
	where we recall that $\xi_t = \cU^{\mathrm{fluct}} (t;0) \xi$ and $\xi_{\eta,t} = \cU^{\mathrm{fluct}} _{\eta} (t;0) \xi$.
\end{lemma}

\begin{proof}
	We only prove a bound for the first term on the l.h.s. of \eqref{5.69}. A bound for the second term can be obtained with the same arguments. We start by noting that
	\begin{align}
		\big\langle \xt, \cal N \big(\xt-\xat\big) \big\rangle &=	\big\langle \xt, \cal N \big(\cU^{\mathrm{fluct}}(t,0)-\cU^{\mathrm{fluct}}_{\eta}(t,0)\big)\x  \big\rangle \nonumber \\
		&=-\ii \int_0^t \dd s \big\langle \xt, \cal N\cal U^{\mathrm{fluct}}(t,s)\left( \cal G_{N,s}-\cal G_{N,s}^{(\eta)} \right)\cal U^{\mathrm{fluct}}_{\eta}(s,0)\x \big\rangle \nonumber \\
		&=-\ii \int_0^t \dd s \big\langle \xt, \cal N \cal U^{\mathrm{fluct}}(t,s)\left(\sum_{i=1}^3\Iplus{i}\right)\cal U^{\mathrm{fluct}}_{\eta}(s,0)\x \big\rangle\,, \label{eq:comparison_1}
	\end{align} where $\Iplus{i}$ is defined exactly as $\cI_{N,t}^{(i,\eta)}$, with $\id(\cN\le N^\eta )$ replaced by $\id(\cN> N^\eta )$. The contribution from the term with four $U$ operators appearing in $\Iplus{1}$ can be estimated by
	\begin{align} 
		&\frac{2}{N}\int_0^t \int \Big|v(\xx, \yy)\big\langle\xt,\cN\cU^{\mathrm{fluct}}(t,s) a^*(\Uxs)a^*(\Uys)a(\Uys)a(\Uxs)\mathds{1}(\mathcal{N} > N^{\eta})\xas\big \rangle\Big| \de\xx\de\yy \de s \nonumber \\
		&\hspace{1cm}\le \frac{2}{N}\int_0^t \Bigg(\int  |v(\xx, \yy)| \ \big\langle\xt,\cN\cU^{\mathrm{fluct}}(t,s) a^*(\Uxs)a^*(\Uys)a(\Uys)a(\Uxs)\cU^{\mathrm{fluct}}(t,s)^*\cN\xt\big\rangle \de\xx\de\yy\Bigg)^{1/2} \nonumber \\
		&\hspace{1.5cm}\times \Bigg(\int  |v(\xx, \yy)| \ \big\langle\xas,\mathds{1}(\mathcal{N} > N^{\eta})a^*(\Uxs)a^*(\Uys)a(\Uys)a(\Uxs)\mathds{1}(\mathcal{N} > N^{\eta})\xas\big \rangle \de\xx\de\yy \Bigg)^{1/2} \de s\,. \label{eq:comparison_1a}
	\end{align} 
	An application of Lemma~\ref{lem5.1} shows that
	\begin{align} 
		&\int  |v(\xx, \yy)| \ \big\langle\xt,\cN\cU^{\mathrm{fluct}}(t,s) a^*(\Uxs)a^*(\Uys)a(\Uys)a(\Uxs)\cU^{\mathrm{fluct}}(t,s)^*\cN\xt\big\rangle \de\xx\de\yy \nonumber \\
		&\hspace{2cm}\le \|v\|_\infty\|U_s\|^4\big\langle\xt,\cN\cU^{\mathrm{fluct}}(t,s) \cN^2 \cU^{\mathrm{fluct}}(t,s)^*\cN\xt\big\rangle \nonumber \\
		&\hspace{2cm}\lesssim \|v\|_\infty C_s^2 \left[ C_s^2 C_t^2 \langle \xt,\cN^{4} \xt \rangle  +  C_t^2N^2 \langle \xt, \mathcal{N}^{2} \xt \rangle \right] \nonumber \\
		&\hspace{2cm}\lesssim \|v\|_\infty C_s^2 \left[ \left( C_s^2 C_t^{6} C_0^{4} + C_t^{6} C_0^{4} \right) \langle \xi,\cN^{4} \xi \rangle + \left( C_s^2 C_t^{6} + C_t^{4} \right) N^{4} \right] \,.
		\label{eq:comparison 1b}
	\end{align} 	
	To bound the second factor, we use
	\begin{equation} \label{eq:trick_inequality}
		\mathds{1}(\mathcal{N} > N^{\eta}) \leq (\cal N/N^{\eta})^k 
	\end{equation}
	with $k \in \NN\setminus \{0\}$. Using \eqref{eq:trick_inequality} and Lemma~\ref{lem:gronwall_N}, we obtain
	\begin{align} 
		&\int |v(\xx, \yy)| \ \big\langle\xas,\mathds{1}(\mathcal{N} > N^{\eta})a^*(\Uxs)a^*(\Uys)a(\Uys)a(\Uxs)\mathds{1}(\mathcal{N} > N^{\eta})\xas\big \rangle \de\xx\de\yy \nonumber \\
		&\hspace{1.5cm}\le \|v\|_\infty \|U_s\|^4 N^{-\eta k}\big\langle\xas,\cN^{k+2}\xas\big\rangle \le \|v\|_\infty \|U_s\|^4 N^{-\eta k} \exp(c\exp(c\exp(cs)))\langle\x,\cN^{k+2}\x \rangle\,. \label{eq:comparison1c}	
	\end{align} 
	We note that Lemma~\ref{lem:closeness_dynamics_V} implies $C_t \lesssim T_{\mathrm{c}} \exp(ct)$. When combined with \eqref{eq:comparison 1b} and \eqref{eq:comparison1c}, this gives
	\begin{equation}
		\eqref{eq:comparison_1a} \lesssim \frac{ \|v\|_\infty T_{\mathrm{c}}}{N^{1+\eta k /2}} \exp(c\exp(c\exp(ct))) \langle \xi, \mathcal{N}^{k+2} \xi \rangle^{1/2} \left[ T_{\mathrm{c}}^{7} \langle \xi, \mathcal{N}^{4} \xi \rangle^{1/2} + T_{\mathrm{c}}^{5} N^{2} \right] \,. \label{eq:comparison1aa}
	\end{equation}
	When inserted into \eqref{eq:comparison_1}, the other terms in $\Iplus{1}$ can be estimated in the same way by the r.h.s. of \eqref{eq:comparison1aa}.
	
	Next, we consider the terms arising from $\Iplus{2}$. We start by considering the term with four creation operators. We bound its contribution to \eqref{eq:comparison_1} by 
	\begin{align} 
		&\frac{2}{N}\int_0^t \int  \Big|v(\xx, \yy)\big\langle\xt,\cN\cU^{\mathrm{fluct}}(t,s) a^*(\Uxs)a^*(\Uys)a^*(\Vys)a^*(\Vxs)\mathds{1}(\mathcal{N} > N^{\eta})\xas\big \rangle\Big| \de\xx\de\yy \de s \nonumber \\
		&\hspace{0.5cm}\le \frac{2}{N}\int_0^t \Bigg(\int  |v(\xx, \yy)| \ \big\langle\xt,\cN\cU^{\mathrm{fluct}}(t,s) a^*(\Uxs)a^*(\Uys)a(\Uys)a(\Uxs)\cU^{\mathrm{fluct}}(t,s)^*\cN\xt\big\rangle \de\xx\de\yy\Bigg)^{1/2} \nonumber \\
		&\hspace{1cm}\times \Bigg(\int |v(\xx, \yy)| \ \big\langle\xas,\mathds{1}(\mathcal{N} > N^{\eta})a(\Vxs)a(\Vys)a^*(\Vys)a^*(\Vxs)\mathds{1}(\mathcal{N} > N^{\eta})\xas\big \rangle \de\xx\de\yy \Bigg)^{1/2} \de s\,. \label{eq:comparison_2}
	\end{align} 
	A bound for the first factor has been obtained in \eqref{eq:comparison 1b}. To estimate the other factor, we use the identity \eqref{eq:commutatorsVVVV}. The normal-ordered term has already been considered in \eqref{eq:comparison1c}. We use Proposition~\ref{lem:diluteness_HFB} and estimates similar as the ones above to bound the remaining contributions by
	\begin{equation*}
		\Vert v \Vert_1 \exp(c \exp(c\exp(cs))) \left[ T_{\mathrm{c}}^{5/2} N^{-\eta (k+1)} +  T_{\mathrm{c}}^{3/2} N^{1-\eta (k+2)} \right] \langle \xi, \mathcal{N}^{k+2} \xi \rangle\,.
	\end{equation*}
	Hence, we have 
	\begin{align*}
		\eqref{eq:comparison_2}\lesssim& \frac{\Vert v \Vert_1 + \Vert v \Vert_{\infty}}{N^{1+\eta k /2}} \exp(c \exp(ct\exp(ct))) \left[ T_{\mathrm{c}}^{7} \langle \xi, \mathcal{N}^{4} \xi \rangle^{1/2} + T_{\mathrm{c}}^{5} N^{2} \right] \\
		&\times \Big[  T_{\mathrm{c}} + T_{\mathrm{c}}^{5/4} N^{-\eta/2} +  N^{-\eta} T_{\mathrm{c}}^{3/4} N^{1/2} \Big] \langle \xi, \mathcal{N}^{k+2} \xi \rangle^{1/2} \,.
	\end{align*} 
	The remaining terms in $\Iplus{2}$ can be estimated analogously.
	
	We are left to consider the contributions from $\Iplus{3}$. We insert the term with three creation operators into \eqref{eq:comparison_2} and find
	\begin{align*}
		&\Bigg|	\frac{3}{N} \int_0^t \int  v(\xx, \yy)\pphi_s(\xx) \ \big\langle\xt,\cN\cU^{\mathrm{fluct}}(t,s) a^*(\Uxs)a^*(\Uys)a^*(\Vys)\mathds{1}(\mathcal{N} > N^{\eta})\xas\big \rangle \de\xx\de\yy \de s \Bigg|\\
		&\hspace{3cm} \le \frac{3}{N}\int_0^t \Bigg(\int |v(\xx, \yy)| \ \big\| a(\Uys)a(\Uxs)\cU^{\mathrm{fluct}}(t,s)^*\cN\xt  \big\|^2 \de\xx\de\yy \Bigg)^{1/2} \\
		&\hspace{4cm} \times\Bigg(\int |v(\xx, \yy)| \ |\pphi_s(\xx)|^2 \big\| a^*(\Vys)\id(\cN> N^\eta )\xas \big\|^2 \de\xx\de\yy \Bigg)^{1/2} \de s \\
		&\hspace{0.5cm}\lesssim \frac{(\Vert v \Vert_1 + \Vert v \Vert_{\infty}) }{N^{1/2+\eta k/2}} \exp(c \exp(c\exp(ct))) \left[ T_{\mathrm{c}}^{7} \langle \xi, \mathcal{N}^{7} \xi \rangle^{1/2} + T_{\mathrm{c}}^{5} N^{2} \right] \\
		&\hspace{1cm}\times \left[ T_{\mathrm{c}}^{1/2} N^{-\eta/2}  + T_{\mathrm{c}}^{3/4}  N^{-\eta} \right] \langle \xi, \mathcal{N}^{k+2} \xi \rangle^{1/2} \,.
	\end{align*}
	The remaining terms in $\Iplus{3}$ can be estimated in a similar way by the r.h.s. of the above equation.
	
	Collecting the above bounds, and using $N^{-1/2}T_{\mathrm c}^{3/4}=N^{-1/2} T_{\mathrm{c}}^{3/4}(s) \lesssim 1$ (which follows from the condition $s \leq 3/2$), we obtain
	\begin{align*}
		\big|	\big\langle \xt, \cal N \big(\xt-\xat\big) \big\rangle \big| \lesssim& N^{-1/2 - \eta k/2} \exp(c \exp(c\exp(ct))) \left( T_{\mathrm{c}}^{7} \langle \xi, \mathcal{N}^{4} \xi \rangle^{1/2} + T_{\mathrm{c}}^{5} N^{2} \right) \\
		&\times \left( T_{\mathrm{c}} N^{-1/2} + T_{\mathrm{c}}^{1/2} N^{-\eta/2} + T_{\mathrm{c}}^{3/4} N^{-\eta} \right)  \big\langle \x, \cN^{k+2}\x \big\rangle^{1/2}\,.
	\end{align*}
\end{proof}
	\subsection{Proof of Proposition \ref{mainprop}}
	\label{sec:proofmainprop}
	We are now ready to prove our final bound for the growth of the number of excitations in the fluctuation vector. For $\eta\in (0,1/7]$ we have, by Proposition~\ref{lem:gronwall_N} and Lemma~\ref{lem5.3},
	\begin{align*}
		\big\langle\xt,\cN\xt\big\rangle \le & \big\langle\xat,\cN\xat\big\rangle + \Big|\big\langle\xat,\cN(\xt-\xat)\big\rangle\Big| +	\Big|\big\langle(\xt-\xat),\cN\xt\big\rangle\Big| \\
		\lesssim &  \exp(c \exp(c\exp(ct))) \big\langle\x,\cN\x\big\rangle + N^{-1/2 - \eta k/2} \exp(c \exp(c\exp(ct))) \\ \label{5.69}
		&\qquad\times \left( T_{\mathrm{c}}^{7} \langle \xi, \mathcal{N}^{4} \xi \rangle^{1/2} + T_{\mathrm{c}}^{5} N^{2} \right) \left( T_{\mathrm{c}} N^{-1/2} + T_{\mathrm{c}}^{1/2} N^{-\eta/2} + T_{\mathrm{c}}^{3/4} N^{-\eta} \right)  \big\langle \x, \cN^{k+2}\x \big\rangle^{1/2} \,.
	\end{align*} Choosing $\eta = 1/7$ and using $T_{\mathrm c}= T_{\mathrm c}(s)\le T_{\mathrm c}(3/2)\lesssim N^{2/7}$ we arrive at \eqref{eq:Gronwall_bound_main}. 
	
	\textbf{Acknowledgments.}
	It is a pleasure to thank Vedran Sohinger for discussions related to Remark~\ref{remark:HartreeH2}. A.~D. would also like to thank Robert Seiringer and Jakob Yngvason for inspiring discussions. A. D. gratefully acknowledges funding from the European Union’s Horizon 2020 research and innovation programme under the Marie Sklodowska-Curie grant agreement No 836146 and from the Swiss National Science Foundation through the Ambizione grant PZ00P2 185851. B. S. gratefully acknowledges partial support from the NCCR SwissMAP and from the Swiss National Science Foundation through the Grant ``Dynamical and energetic properties of Bose-Einstein condensates’’. B. S. and M. C. gratefully acknowledge partial support from the European Research Council through the ERC-AdG CLaQS.

\vspace{0.5cm}
(Marco Caporaletti) Institute of Mathematics, University of Zurich \\ 
Winterthurerstrasse 190, 8057 Zurich, Switzerland \\ 
E-mail address: \texttt{marco.caporaletti@math.uzh.ch}\\
\\
(Andreas Deuchert) Institute of Mathematics, University of Zurich \\ 
Winterthurerstrasse 190, 8057 Zurich, Switzerland \\ 
E-mail address: \texttt{andreas.deuchert@math.uzh.ch}\\
\\
(Benjamin Schlein) Institute of Mathematics, University of Zurich \\ 
Winterthurerstrasse 190, 8057 Zurich, Switzerland \\ 
E-mail address: \texttt{benjamin.schlein@math.uzh.ch}

\end{document}